\documentclass[11pt, twoside, table]{predoc2}

\usepackage[utf8]{inputenc}
\usepackage[T1]{fontenc}
\usepackage{newtxtext} 
\usepackage[scaled=.95]{cabin} 
\usepackage[varqu,varl]{inconsolata} 
\usepackage{amsmath}
\usepackage{tabularx}
\usepackage{xltabular}
\usepackage[varg]{newtxmath}
\usepackage{changepage}
\usepackage[title]{appendix}

\usepackage{url}            
\usepackage{graphicx}
\usepackage{float}
\usepackage{booktabs}       
\usepackage{amsfonts}       
\usepackage{nicefrac}       
\usepackage{microtype}      
\usepackage{xcolor}         

\usepackage{mathtools}
\usepackage{bbm}                        
\usepackage{eucal}                      

\usepackage{xcolor}

\usepackage{enumitem}
\usepackage{array}
\usepackage{booktabs}
\usepackage{tabularx}

\usepackage{amsthm}

\newtheorem{proposition}{Proposition}
\newtheorem{corollary}{Corollary}
\newtheorem{lemma}{Lemma}

\newcommand\BlackBox{}

\newcommand\numberthis{\addtocounter{equation}{1}\tag{\theequation}}

\let\originalleft\left
\let\originalright\right
\renewcommand{\left}{\mathopen{}\mathclose\bgroup\originalleft}
\renewcommand{\right}{\aftergroup\egroup\originalright}

\newcommand*{\+}[1]{#1}  

\DeclareMathOperator{\N}{N}
\DeclareMathOperator{\E}{E}
\DeclareMathOperator{\SD}{SD}
\DeclareMathOperator{\Var}{Var}
\DeclareMathOperator{\Cov}{Cov}
\DeclareMathOperator{\skewness}{skewness}
\DeclareMathOperator{\Poisson}{Poisson}
\DeclareMathOperator{\normal}{normal}
\DeclareMathOperator{\Dirichlet}{Dirichlet}

\newcommand*{\p}{p}

\newcommand*{\diff}{\mathop{}\kern -2pt\mathrm{d}}

\newcommand*{\eye}{\mathrm{I}}

\newcommand*{\ones}{\mathbbm{1}}  

\newcommand*{\bcdot}{\boldsymbol{\cdot}}

\newcommand*{\x}{\ensuremath{x}}
\newcommand*{\y}{\ensuremath{y}}
\newcommand*{\Short}{{L}}

\newcommand*{\Ma}{{\ensuremath{\mathrm{M}_a}}}
\newcommand*{\Mb}{{\ensuremath{\mathrm{M}_b}}}
\newcommand*{\Mk}{{\ensuremath{\mathrm{M}_k}}}
\newcommand*{\Msa}{{\ensuremath{a}}}
\newcommand*{\Msb}{{\ensuremath{b}}}
\newcommand*{\Msk}{{\ensuremath{k}}}
\newcommand*{\Msd}{{\ensuremath{a-b}}}
\newcommand*{\Mfa}{{\ensuremath{\mathrm{M_A}}}}
\newcommand*{\Mfb}{{\ensuremath{\mathrm{M_B}}}}
\newcommand*{\Mfsa}{{\ensuremath{\mathrm{A}}}}
\newcommand*{\Mfsb}{{\ensuremath{\mathrm{B}}}}
\newcommand*{\Mfsd}{{\ensuremath{\mathrm{A-B}}}}

\renewcommand{\mid}{|}

\newcommand*{\yobs}{\ensuremath{y}}

\newcommand*{\eelpd}[1]{{\text{e-elpd}\bigr(#1\bigl)}}

\newcommand*{\elpd}[2]{{\mathrm{elpd}\bigr(#1 \mid #2\bigl)}}
\newcommand*{\elpdi}[3]{{\mathrm{elpd}_{#3}\bigr(#1 \mid #2\bigl)}}

\newcommand*{\elpdHat}[2]{{\widehat{\mathrm{elpd}}_\mathrm{\scriptscriptstyle LOO}\bigr(#1 \mid #2\bigl)}}

\newcommand*{\elpdHatErr}[2]{{\mathrm{err}_\mathrm{\scriptscriptstyle LOO}\bigr(#1 \mid #2\bigl)}}

\newcommand*{\elpdHatErrEE}[3]{{\mathrm{err}_\mathrm{\scriptscriptstyle LOO}^{#3}\bigr(#1 \mid #2\bigl)}}

\newcommand*{\elpdHati}[3]{{\widehat{\mathrm{elpd}}_{\mathrm{\scriptscriptstyle LOO},\, #3}\bigr(#1 \mid #2\bigl)}}

\newcommand*{\seHat}[2]{{\widehat{\mathrm{SE}}_\mathrm{\scriptscriptstyle LOO}\bigr(#1 \mid #2\bigl)}}

\newcommand*{\Md}{{\ensuremath{\mathrm{M}_a,\mathrm{M}_b}}}

\newcommand*{\Mfd}{{\ensuremath{\mathrm{M_A,M_B}}}}

\newcommand*{\elpdC}[2]{{\mathrm{elpd}\bigr(#1 \mid #2\bigl)}}

\newcommand*{\elpdHatC}[2]{{\widehat{\mathrm{elpd}}_\mathrm{\scriptscriptstyle LOO}\bigr(#1 \mid #2\bigl)}}

\newcommand*{\elpdHatUnkC}[2]{{\mathrm{unc}_\mathrm{\scriptscriptstyle LOO}\bigr(#1 \mid #2\bigl)}}

\newcommand*{\elpdHatUnkHatC}[2]{{\widehat{\mathrm{unc}}_\mathrm{\scriptscriptstyle LOO}\bigr(#1 \mid #2\bigl)}}

\newcommand*{\elpdHatErrEEC}[3]{{\mathrm{err}_\mathrm{\scriptscriptstyle LOO}^{#3}\bigr(#1 \mid #2\bigl)}}

\newcommand*{\elpdHatErrC}[2]{{\mathrm{err}_\mathrm{\scriptscriptstyle LOO}\bigr(#1 \mid #2\bigl)}}

\newcommand*{\elpdHatErrHatC}[2]{{\widehat{\mathrm{err}}_\mathrm{\scriptscriptstyle LOO}\bigr(#1 \mid #2\bigl)}}

\newcommand*{\elpdHatiC}[3]{{\widehat{\mathrm{elpd}}_{\mathrm{\scriptscriptstyle LOO},\, #3}\bigr(#1 \mid #2\bigl)}}

\newcommand*{\seHatC}[2]{{\widehat{\mathrm{SE}}_\mathrm{\scriptscriptstyle LOO}\bigr(#1 \mid #2\bigl)}}

\newcommand*{\elpdPlain}{{\ensuremath{\mathrm{elpd}}}}
\newcommand*{\eelpdPlain}{{\ensuremath{\text{e-elpd}}}}

\newcommand*{\elpdHatPlain}{{\ensuremath{\widehat{\mathrm{elpd}}_\mathrm{\scriptscriptstyle LOO}}}}

\newcommand*{\elpdHatErrPlain}{{\mathrm{err}_\mathrm{\scriptscriptstyle LOO}}}

\newcommand*{\seHatPlain}{{\ensuremath{\widehat{\mathrm{SE}}_\mathrm{\scriptscriptstyle LOO}}}}

\newcommand{\para}[1]{\paragraph{#1}}

\newlength\figurecontrolwidth
\title{Uncertainty in {Bayesian} Leave-One-Out Cross-Validation Based Model Comparison}

\author[1]{Tuomas Sivula}
\affil[1]{Aalto University, Finland}
\author[2]{Måns Magnusson\thanks{Most of the work was done while at Aalto University.}}
\affil[2]{Uppsala University, Sweden}
\author[1]{Asael Alonzo Matamoros}
\author[1]{Aki Vehtari}

\makeatletter
\def\runauthor{Sivula, Magnusson, Matamoros, and Vehtari}
\def\runtitle{Uncertainty in {Bayesian} Leave-One-Out Cross-Validation Based Model Comparison}
\makeatother

\begin{document}

\maketitle

\begin{abstract}
{It is useful to estimate the expected predictive performance of models planned to be used for prediction.  We focus on leave-one-out cross-validation (LOO-CV), which has become a popular method for estimating predictive performance of Bayesian models. Given two models, we are interested in comparing the predictive performances and associated uncertainty, which can also be used to compute the probability of one model having better predictive performance than the other model. We study the properties of the Bayesian LOO-CV estimator and the related uncertainty quantification for the predictive performance difference,
and analyse when a normal approximation of this uncertainty is well calibrated and whether taking into account higher moments could improve the approximation.}
We provide new results of the properties both theoretically in the linear regression case and empirically for {hierarchical linear, latent linear, and spline models} and discuss the challenges.
We show that problematic cases include: comparing models with similar predictions, misspecified models, and small data. In these cases, there is a weak connection between the distributions of the LOO-CV estimator and its error. We show that
that the problematic skewness of the error distribution {for the difference}, which occurs when the models make similar predictions, does 
not fade away when the data size grows to infinity in certain situations. Based on the results, we also provide some practical recommendations for the users of Bayesian LOO-CV for {comparing predictive performance of models}.

\end{abstract}

\begin{keywords}
  Bayesian computation, model comparison, leave-one-out cross-validation, uncertainty, asymptotics
\end{keywords}

\section{Introduction}
\label{sec_intro}
We are often interested in the predictive performance of Bayesian models for new, unseen data. Given two models, we are then also interested in comparing the predictive performances and the probability that one has better predictive performance than the other model.
We cannot directly compute the predictive performance for unseen data. We can estimate it using, for example, cross-validation \citep{Geisser:1975,Geisser+Eddy:1979,Gelfand+Dey+Chang:1992,Bernardo+Smith:1994,Gelfand:1996,Vehtari+Ojanen:2012} and then, in the model comparison, take into account the uncertainty related to the difference of the predictive performance estimates for the different models \citep{Vehtari+Lampinen:2002,Vehtari+Ojanen:2012}.

Leave-one-out cross-validation (LOO-CV) has become a popular approach for estimating Bayesian predictive performance; For example, \texttt{loo} R package \citep{vehtari2022loopkg}, which implements a fast LOO-CV computation \citep{Vehtari+Gelman+Gabry:2017_practical,Vehtari+etal:2024:PSIS}, has been downloaded more than 4 million times from RStudio CRAN mirror alone. The \texttt{loo} package uses a normal approximation to quantify the uncertainty in the predictive performances and the difference in the predictive performance of two models.  The uncertainty in the predictive performance estimates is due to approximating the unknown future data distribution with a finite number of re-used observations.
To draw rigorous conclusions about the difference in predictive performance, we need to assess the accuracy of the estimated uncertainty, a problem that recently also has attracted attention in the frequentist setting \citep[e.g. see][]{austern2020asymptotics,bayle2020cross,bates2023cross}. How well is the uncertainty quantification calibrated when repeatedly applied to a new, comparable problem? Can there be some settings in which the uncertainty is, in general, poorly quantified? Are there some general characteristics that make it hard to estimate the uncertainty? This paper carefully analyses these properties and provides some practical guidance for modellers in the Bayesian setting. 

\subsection{Our Contributions}
\label{sec_contributions}

We provide new theoretical and empirical results for the uncertainty quantification in Bayesian LOO-CV model predictive performance comparison, and illustrate the challenges of quantifying it. We focus on analysing \textit{the difference in the predictive performance} of the  LOO-CV estimator of the \emph{expected log pointwise predictive density} (elpd) in two {linear} model comparisons. %
What matters is which model has better predictive performance, how much better it is, and what the associated uncertainty is. With our focus on predictive performance, we do not need to assume that one of the models is the true model, and thus, we do not consider the probability of selecting the true model. 
We focus on the finite sample size behaviour but also investigate asymptotic properties. We discuss how predictive performance comparison can be used for model selection.

We formulate the underlying uncertainty %
and present the two ways of analysing it: the normal approximation and the Bayesian bootstrap \citep[i.e. Dirichlet process approximation;][]{Rubin:1981a,Lo:1987a}. We analyse the properties of the error distribution and the approximations of that distribution in typical normal linear regression problem settings over possible data sets.
Based on this analysis, we identify when these uncertainty estimates can perform poorly: the models make similar predictions (Scenario~1), the models are misspecified with outliers in the data (Scenario~2), or the number of observations is small (Scenario~3).
The consequences of these problematic cases are:
\begin{enumerate}[topsep=0pt,partopsep=2pt,itemsep=2pt,parsep=2pt]
    \item When the models make similar predictions (Scenario~1), there is not much difference in the predictive performance and we can use either model for prediction.
    \item Model misspecification in model comparison (Scenario~2) should be avoided by proper model checking and expansion before using LOO-CV, and thus this should not happen for any final comparisons.
    \item LOO-CV can not reliably detect small differences in the predictive performance if the number of observations is small (Scenario~3).
\end{enumerate}

We have derived analytical results for normal linear regression with random covariates and demonstrate experimentally the same behaviour with a fixed covariate, hierarchical linear, (Poisson) generalised linear, and spline models (these types of models probably cover more than 90\% models used in applied Bayesian modelling). The underlying reasons and consequences are the same for Bayesian $K$-fold-CV, and we demonstrate similar behaviour in experimental results.
In a non-Bayesian context, \citet{arlot_celisse_2010_cv_survey} provide several results for different cross-validation approaches, where most of the discussed results are similar to the results presented here. To the best of our knowledge, these are the first results in a Bayesian domain, including pre-asymptotic behaviour.

\section{Problem Setting}
\label{sec_problem_setting}

For each positive integers $n > 0$, $\y=(\y_1, \y_2, \dots, \y_n)$ is generated from $\p_\text{true}(\y\mid\x)$, representing the true data generating process for $\y$ conditional on covariates $\x=(\x_1, \x_2, \dots, \x_n)$. Here $\y_i$ are assumed exchangeable conditionally on the covariates \citep[see, e.g.,][Section 5.2]{bda_book}.
For evaluating models \Mk{}$\in \{\Ma{},\Mb{} \}$, we consider the \emph{expected log pointwise predictive density}~\citep{Vehtari+Ojanen:2012},
a measure of predictive accuracy for another data set $\tilde{\y}=(\tilde{\y}_1, \tilde{\y}_2, \dots, \tilde{\y}_n)$, independent of $\y$, and generated from the same true data generating process $\p_\text{true}(\y|\x)$ as: 
\begin{align}\label{eq_elpd}
 \elpd{\Mk}{\y} &= \sum_{i=1}^n \elpdi{\Mk}{\y}{i} = \sum_{i=1}^n \E_{\tilde{\y}_i,\tilde{\x}_i} \left[ \log \p_\Mk(\tilde{\y}_i \mid \tilde{\x}_i, \y) \right] \nonumber\\ & = \sum_{i=1}^n \int \p_\mathrm{true}(\tilde{\y}_i\mid\tilde{\x}_i)\p_\mathrm{true}(\tilde{\x}_i) \log \p_\Mk(\tilde{\y}_i \mid \tilde{\x}_i, \y) \diff \tilde{\y}_i\tilde{\x}_i\,,
\end{align}
where $\log \p_\Mk(\tilde{\y}_i \mid \y)$ is the logarithm of the posterior predictive density for the model \Mk{} fitted for data set $\y$. If $\tilde{\x}$ are fixed,~\eqref{eq_elpd} simplifies to $\sum_{i=1}^n \int \p_\mathrm{true}(\tilde{\y}_i\mid\tilde{\x}_i) \log \p_\Mk(\tilde{\y}_i \mid \tilde{\x}_i, \y) \diff \tilde{\y}_i$.
Although we do not consider data shifts, covariate shifts can be included in the model for future data.
A more detailed discussion of the covariate setup is presented by  \citet{Vehtari+Lampinen:2002} and  \citet{Vehtari+Ojanen:2012}. We omit the covariates for brevity most of the time in the notation.
Here, the observations are considered pointwise to maintain comparability with the given data set~\citep[p.~168,][]{bda_book}.
A summary of notation used in the paper is presented in Table~\ref{tab_list_of_notation}. We can use different score functions, but for simplicity, we use the strictly proper and local log  score~\citep{Gneiting+Raftery:2007,Vehtari+Ojanen:2012} throughout the paper.

For evaluating model \Mk{} in the context of a specific data generating process in general, the respective measure of predictive performance is the expectation of $\elpd{\Mk}{\y}$ over all possible data sets $\y$ we might have observed:
\begin{align}
    \label{eq_eelpd}
    \text{expected } \elpd{\Mk}{\y} = \E_{\y}\left[ \elpd{\Mk}{\y} \right] \,.
\end{align}
The $\elpd{\Mk}{\y}$ in Equation~\eqref{eq_elpd}, conditioned on $\y$, can be considered as an estimate for the measure in Equation~\eqref{eq_eelpd}. Our focus is in \eqref{eq_elpd}, which useful in the application-oriented model-building workflow \citep{Gelman+etal:2020:workflow} when evaluating models conditional on the observed data. Measure~\eqref{eq_eelpd} is of interest in algorithm-oriented experiments when analysing the performance of models in the context of a problem set in general~\citep[e.g.][]{Dietterich:1998,bengio_Grandvalet_2004}. 
We further discuss these measures' differences and their uncertainties in Appendix~\ref{app_sec_theory_elpd_vs_eelpd}.

\begingroup
 \setlength{\tabcolsep}{0pt}
 \begin{table}[t]
   \small
   \begin{center}
 \begin{tabularx}{5.25in}{@{}p{0.84in}p{4.41in}@{}}
    $n$ &
    number of observations in a data set
 \\
    $\y$ &
    data set of $n$ observations from $\p_\text{true}(\y)$
 \\
    $\tilde{\y}$ &
    another independent analogous data set of $n$ observations from $\p_\text{true}(\y)$
 \\
    \Mk{} &
    model variable indicating model $k$
 \\
    $\p_\text{true}(\y)$ &
    distribution representing the true data generating process for $\y$ and $\tilde{\y}$
 \\
    $\p_\Mk(\tilde{\y}_i \mid \yobs)$ &
    posterior predictive distribution with model $\Mk$
 \\
    $\elpdPlain$ &
    expected log pointwise predictive density score, see Eq. \eqref{eq_elpd} and \eqref{eq_elpd_diff_exact}
 \\
    $\elpdHatPlain$ &
    LOO-CV approximation to $\elpdPlain$, see Eq. \eqref{eq_elpd_loo} and \eqref{eq_elpd_diff}
 \\ 
    $\elpdHatErrPlain$ &
    LOO-CV approximation error for $\elpd{\cdots}{\yobs}$, see Eq. \eqref{eq_elpdhat_err}
 \\ \addlinespace[.5mm]
    $p(\elpdHatErrPlain)$ &
    the true distribution of uncertainty in $\elpdHatErrPlain$
 \\
 $\hat{p}(\elpdHatErrPlain)$ &
    approximate distribution $\hat{p}(\elpdHatErrPlain) \approx p(\elpdHatErrPlain)$
 \\
    $\seHatPlain$ &
    estimator for the standard deviation of $\elpdHat{\cdots}{\y}$
 \end{tabularx}
   \end{center}
 \caption{Notation used.}\label{tab_list_of_notation}
 \end{table}
\endgroup

\subsection{Bayesian Cross-Validation}
As the true data generating process $\p_\text{true}(\y)$ is usually unknown,~\eqref{eq_elpd} needs to be approximated~\citep{Bernardo+Smith:1994,Vehtari+Ojanen:2012}.
If we had independent test data $\tilde{\y}=(\tilde{\y}_1, \tilde{\y}_2, \dots, \tilde{\y}_n) \sim \p_\mathrm{true}(\tilde{y})$
  , that is, observations from the same data generating process as $y$,
we could estimate~\eqref{eq_elpd} as 
\begin{equation}\label{eq_elpd_test}
 \widehat{\mathrm{elpd}}_\mathrm{test}(\Mk \mid {\yobs}) = \sum_{i=1}^n \log \p_\Mk\left(\tilde{y}_i \mid \yobs\right).
\end{equation}

When independent test data are not available, which is often the case in practice,
a popular strategy is cross-validation, in which a finite number of observations are re-used as a proxy for the unobserved independent data \citep{Geisser:1975}. The data is divided into parts, which are used as out-of-sample validation sets for the model fitted using the remaining observations. 
In leave-one-out cross-validation (LOO-CV), each observation is one validation set, and we approximate $\elpd{\Mk}{\yobs}$ as
\begin{align}\label{eq_elpd_loo}
 \elpdHat{\Mk}{\yobs} & = \sum_{i=1}^n \elpdHati{\Mk}{\yobs}{i}
                        = \sum_{i=1}^n \log \p_\Mk\left(\yobs_i \mid \yobs_{-i}\right)\,, \\
\intertext{where}
 \elpdHati{\Mk}{\yobs}{i} & = \log \p_\Mk\big(\yobs_i \big| \yobs_{-i}\big)  = \log \int \p_\Mk\big(\yobs_i \big| \theta_k \big)\p_\Mk \big(\theta_k \big| \yobs_{-i}\big) \diff \theta_k
\end{align}
is the LOO predictive log density for the $i$th observation $\yobs_i$ with model \Mk{} and parameters $\theta_k$, given the data except the $i$th observation, denoted as $\yobs_{-i}$.
The observations $\yobs_i$ are assumed to be exchangeable (conditionally on covariates).
The bias of~\eqref{eq_elpd_loo} tends to decrease when $n$ grows \citep{Watanabe:2010d}.
The stability of the learning algorithm affects the variance of CV estimators \citep[][Section 5.2.1]{arlot_celisse_2010_cv_survey}.
As the log score is smooth and integration over the posterior smooths out sharp changes, Bayesian LOO-CV tends to have lower variance than Bayesian $K$-fold-CV \citep{Vehtari+Gelman+Gabry:2017_practical}.
The naive approach would fit the model separately for each fold $\p_\Mk\big(\yobs_i \big| \yobs_{-i}\big)$. In practice, we use more efficient methods such as Pareto smoothed importance sampling \citep{Vehtari+etal:2024:PSIS}, implicitly adaptive importance sampling \citep{Paananen+etal:2021:implicit} and sub-sampling \citep{Magnusson+etal:2019,Magnusson+etal:2020} to estimate $\elpd{\Mk}{\yobs}$ more efficiently.

\para{Predictive performance comparison}
For comparing two models, \Ma{} and \Mb{}, given the same data $\yobs$, we estimate the difference in their expected predictive performance,
\begin{align*}
 \elpd{\Md}{\yobs} &= \elpd{\Ma}{\yobs} - \elpd{\Mb}{\yobs}
 \numberthis\label{eq_elpd_diff_exact}
 \intertext{as}
 \elpdHat{\Md}{\yobs} &= \elpdHat{\Ma}{\yobs} - \elpdHat{\Mb}{\yobs}
 \\
 &= \sum^n_{i=1} \left( \log \p_\Ma\big(\yobs_i \big| \yobs_{-i}\big) - \log \p_\Mb\big(\yobs_i \big| \yobs_{-i}\big) \right)
 \\
 &= \sum^n_{i=1} \elpdHati{\Md}{\yobs}{i}
 \,.
  \numberthis
  \label{eq_elpd_diff}
\end{align*}

\subsection{Uncertainty in Cross-Validation Estimators}
\label{sec_intro_uncertainty}
The true distribution $\p_\mathrm{true}(\tilde{\y},\tilde{\x})$ needed to compute~\eqref{eq_elpd} and \eqref{eq_elpd_diff_exact} is unknown, but we can model it and find the posterior distribution for~\eqref{eq_elpd} and \eqref{eq_elpd_diff_exact}.
We use a minimal assumption model, that is, a flat Dirichlet process prior \citep{Lo:1987a}, to model the unknown $\p_\mathrm{true}(\tilde{\y},\tilde{\x})$. Conditioning on $\y_i,\x_i$ and using the cross-validation terms, we obtain the posterior for~\eqref{eq_elpd_diff_exact} (and similarly for~\eqref{eq_elpd}) as
\begin{align}
  p(\elpd{\Md}{\y}) \approx n \sum_{i=1}^n w_i \,  \elpdHati{\Md}{\yobs}{i}  ,
  \label{eq:dirichlet_posterior}
\end{align}
where $w \sim \Dirichlet_n(1, \ldots, 1)$. If $\tilde{\x}$ is fixed and assuming the pointwise scores  $\elpdHati{}{}{i}$ are exchangeable, the flat Dirichlet process prior is used for the scores, and the posterior is the same. The Dirichlet process posterior approaches the true distribution as $n \rightarrow \infty$ \citep{Lo:1987a}.
There are two practical ways to approximate the Dirichlet process posterior in practice: the normal distribution and the Bayesian bootstrap.

\para{Normal Approximation}
The mean and variance of the Dirichlet process posterior are available analytically. Assuming $\Var[\elpdHati{\Mk}{\yobs}{i}]$ is finite and using the result by \citet{Lo:1987a}, the posterior $p(\elpd{\Mk}{\y})$ can be approximated with the following normal distribution
\begin{align}\label{eq_normalapprox_d}
    \hat{p}(\elpd{\Md}{\y}) &= \N\left(\elpdHat{\Md}{\yobs}, \, \seHat{\Md}{\yobs} \right)\,, 
\end{align}
where $\N(\mu,\sigma)$ is the normal distribution, $\elpdHat{\Md}{\yobs}$ is the sample mean \eqref{eq_elpd_diff}; and $\seHat{\Md}{\yobs}$ is the sample standard error defined as
\begin{align}\label{eq_naive_estimator_main}
\seHat{\Md}{\yobs}
    = \sqrt{\frac{n}{n-1}\sum_{i=1}^n \biggl(
        \elpdHati{\Md}{\yobs}{i}
        - \frac{1}{n}\sum_{j=1}^n \elpdHati{\Md}{\yobs}{j}
    \biggr)^2}.
\end{align}
Similar normal approximation has been used, but without the above Bayesian posterior justification, for example, for cross-validation performance of a single model by \citet{Breiman+Friedman+Olshen+Stone:1984}, and for performance difference by \citet{Dietterich:1998} in a non-Bayesian algorithm-oriented experiments context, and by \citet{Vehtari+Lampinen:2002} in a Bayesian context for given data.
We provide the conditions when the normal approximation~\eqref{eq_normalapprox_d} is well calibrated. Assuming the normal approximation is well-calibrated,
it can be used to further estimate $\hat{p}\left(\elpd{\Md}{\y} >0\right)$, the probability that model $\Ma{}$ has better predictive performance than model $\Mb{}$.

\para{Bayesian Bootstrap Approximation}
The posterior~\eqref{eq:dirichlet_posterior} can also be approximated using Monte Carlo to draw from the Dirichlet distribution. This approach is also known as the Bayesian bootstrap \citep{Rubin:1981a}. \citet{Vehtari+Lampinen:2002} proposed to use the Bayesian bootstrap for the performance difference in a Bayesian context. \citet{Weng:1989:second_order_property_BB} shows that the Bayesian bootstrap produces a more accurate posterior approximation than normal approximation or bootstrap with multinomial weights. Although the focus in this paper is on the normal approximation, we demonstrate that the Bayesian bootstrap does not perform better than the normal approximation, and we discuss why.

\para{Assessing the Approximation Accuracy}
In the two-model comparison, we define the error in the LOO-CV estimate as
\begin{align}\label{eq_elpdhat_err}
 \elpdHatErr{\Md}{\y} = \elpdHat{\Md}{\y} - \elpd{\Md}{\y}.
\end{align}
We compare the approximated error distribution $\hat{p}(\elpdHatErr{\Md}{\y})$ to the actual known values $\elpdHatErr{\Md}{\y}$ and the true distribution $p(\elpdHatErr{\Md}{\y})$ analytically in a simple case and with simulation in other cases.
We use the probability integral transform (PIT) method~\citep[see e.g.][]{Gneiting+Balabdaoui+Raftery:2007,Sailynoja+etal:2022:graphical} to analyse how well $\hat{p}(\elpdHatErr{\Md}{\y})$ is calibrated with respect to $p(\elpdHatErr{\Md}{\y})$. When many data sets $y$ are simulated from known $\p_\text{true}(\y)$, PIT values from a perfectly calibrated $\hat{p}(\elpdHatErr{\Md}{\y})$ will be uniformly distributed. If $\hat{p}(\elpdHatErr{\Md}{\y})$ is well calibrated, then the probabilities of one model having better predictive performance than the other will also be well calibrated.
For some specific simple data generating processes and models, we can analytically derive the moments of 
$p(\elpdHatErr{\Md}{\y})$, and compare them to the moments of the approximated uncertainty $\hat{p}(\elpdHatErr{\Md}{\y})$ to get insights into why the calibration can be far from perfect in some scenarios.
We also consider the distribution of $\elpdHat{\Md}{\yobs}$ as a statistic over possible data sets $\yobs$, and call it the sampling distribution. This follows the standard definition used in frequentist statistics.

\subsection{Problems in Quantifying the Uncertainty}
\label{sec_intro_problems_in_uncertainty}

In this section, we review the main previously known challenges related to quantifying uncertainty in LOO-CV, specifically for predictive performance differences.

\para{No Unbiased Estimator for the Variance}
\citet{bengio_Grandvalet_2004} show that there is no generally unbiased estimator for the variance of $\elpdHat{\Mk}{\y}$ nor $\elpdHat{\Md}{\y}$. As each observation is part of $n-1$ ``training'' sets, the contributing terms in $\elpdHat{\cdot}{\yobs}$ are not independent. The naive variance estimator used to compute $\seHat{\Md}{\y})$ in~\eqref{eq_naive_estimator_main} is biased \citep[see e.g.][]{Sivula+etal:2022:unbiased}. Even though it is possible to derive unbiased estimators for certain models \citep{Sivula+etal:2022:unbiased}, an exact unbiased estimator is not required if the bias is negligible. Based on experimental results, the variance of $\elpdHat{\Mk}{\y}$ can be greatly underestimated when $n$ is small, if the model is misspecified, or if there are outliers in the data \citep{bengio_Grandvalet_2004,varoquaux2017166,varoquaux201868}.
We show that under-estimation of the variance also holds for $\elpdHat{\Md}{\y}$, even more when the models have similar predictions.

\para{Potentially High Skewness}
The distribution  $p(\elpdHat{\Md}{\y})$ can be highly skewed, which would affect the usefulness  of the normal approximation. We show that estimating the skewness of $\elpdHat{\Md}{\y}$ from the contributing terms $\elpdHati{\Md}{\yobs}{i}$ is a challenging task. To capture higher moments, \citet{Vehtari+Lampinen:2002} proposed to use BB  \citep{Rubin:1981a}, which in theory should be more accurate \citep{Weng:1989:second_order_property_BB}, but it also has problems with heavy-tailed distributions, as the approximation is essentially truncated at the extreme observed values~\citep[as already noted by][]{Rubin:1981a}.
Furthermore, as we show in this paper, the mismatch between distributions $p(\elpdHatErr{\Md}{\y})$ and the distribution of the contributing terms $\elpdHati{\Md}{\yobs}{i}$, means that we are not able to obtain useful information about the higher moments.
In our experiments, there was no practical benefit to using BB instead of normal approximation.

\para{Mismatch Between Contributing Terms and Error Distributions}
We construct the approximated distribution $\hat{p}(\elpd{\Md}{\y})$ using information available in the terms $\elpdHati{\Md}{\yobs}{i}$, but we show in this paper that 
the connection between the true distribution $p(\elpdHatErr{\Md}{\y})$ and the distribution of the terms $\elpdHati{\Md}{\yobs}{i}$
can be weak.
This is because, in addition to $\elpdHati{\Md}{\yobs}{i}$, the distribution $p(\elpdHatErr{\Md}{\y})$ is affected by the dependent term $\elpd{\Md}{\y}$, as seen in~\eqref{eq_elpdhat_err}. 
We show that even if the true distribution of the contributing terms $\elpdHati{\Md}{\yobs}{i}$ is known, it may not help in producing a good approximation for $p(\elpdHatErr{\Md}{\y})$.

\para{Asymptotic Uncertainty in the Difference}
\citet{shao1993linear} shows that for the nested least squares linear models, asymptotically, all the models that include the true model will have the same predictive squared error and the standard deviation goes down at the same speed as the differences. In this case, even if the predictive performance difference approaches 0, there remains uncertainty about which model has the best predictive performance (\citeauthor{shao1993linear} also discusses model selection inconsistency, which is not relevant for this paper, as we focus on the predictive performance and do not assume that a true model exists).
 We provide finite case and asymptotic results in the Bayesian context with log score and analyse higher moments of uncertainty for the predictive performance difference. Our results show that in models with asymptotically the same performance, the magnitude of the uncertainty goes down at the same speed as the difference.

\para{Effect of Model Misspecification}
Finally, model misspecification and outliers in the data affect the results in complex ways. \citet{bengio_Grandvalet_2004} demonstrate that given a well-specified model without outliers in the data, the correlation between measures for individual observations may subside as $n$ grows. They also demonstrate that if the model is misspecified and there are outliers in the data, the correlation may significantly affect the total variance even with large $n$. 
We show that outliers affect the constants in the moment terms, and thus, larger $n$ is required to achieve good calibration.

\begin{figure}[tb!]
 \centering
 \includegraphics[width=0.89\figurecontrolwidth]{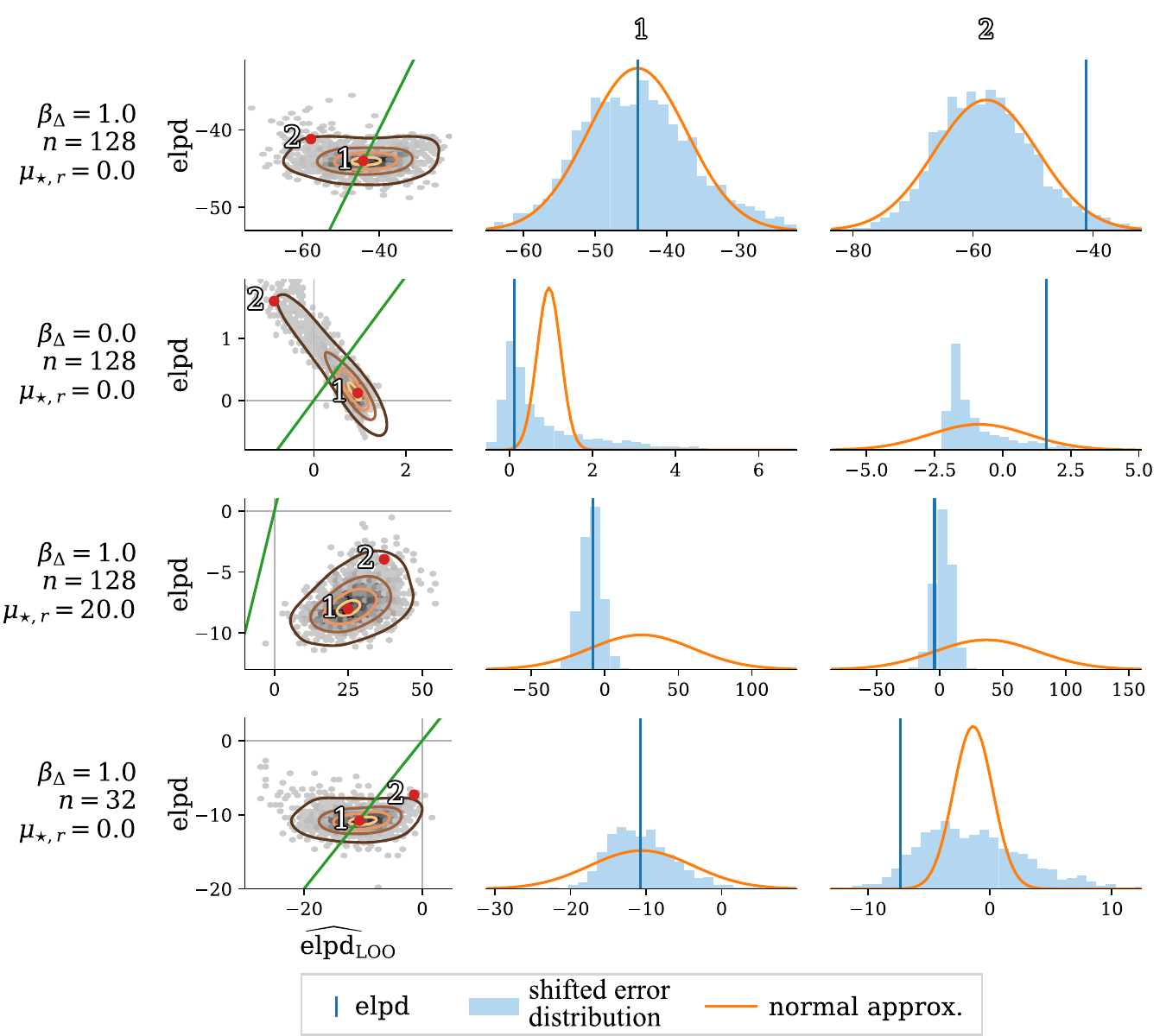}
 \caption{%
  Demonstration of the uncertainty quantification in a simulated normal linear regression. Two realisations in each setting are illustrated in more detail: (1) near the mode and (2) at the tail area of the distribution of the predictive performance and its estimate. Parameter $\beta_\Delta$ controls the difference in the predictive performance of the models, $n$ is the size of the data set, and $\mu_{\star,\mathrm{r}}$ is the magnitude of an outlier observation. The experiments are described in Section~\ref{sec_experiments}. In the first column, the green diagonal line indicates where $\elpd{\Md}{\y}=\elpdHat{\Md}{\y}$ and the brown-yellow lines illustrate density isocontours estimated with the Gaussian kernel method with bandwidth 0.5.
  In the second and third columns, the yellow line shows the normal approximation to the uncertainty, and the blue histogram illustrates the corresponding target, the error distribution located at $\elpdHat{\Md}{\yobs}$. See more explanations in the main text.
 }
 \label{fig_demonstration}
 \vspace{-6pt}
\end{figure}

\para{Demonstration of the Uncertainty Quantification}
Figure~\ref{fig_demonstration} demonstrates normal approximation~\eqref{eq_normalapprox_d} in different simulated linear regression cases. We later demonstrate that similar behaviour also occurs in other settings. The selected example realisations represent the behaviour near the mode and at the tail area of the distribution of the predictive performance difference and its estimate. The normal approximation is good when the difference is relatively big, but in other cases, it can be inaccurate. 
The BB approximation was similar to the normal approximation in all the experimented cases, and the results are not shown in the figure.
\begin{enumerate}[topsep=0pt,partopsep=2pt,itemsep=2pt,parsep=2pt]
 \item In the first case, the normal approximation $\hat{p}(\elpd{\Md}{\y})$ is close to the error distribution $\p(\elpdHatErr{\Md}{\y})$, and correctly indicates that the model $\Mb$ has better predictive performance.
 \item In the second case, the models have similar predictive performance (Scenario~1), and the distribution $\p(\elpd{\Md}{\y})$ is skewed. In the case near the mode, the uncertainty is underestimated, and the normal approximation $\hat{p}(\elpd{\Md}{\y})$ incorrectly indicates that the model $\Ma$ has slightly better predictive performance. In the case of the tail area, the uncertainty is overestimated, which is not harmful as it emphasises the uncertainty of the sign of the performance difference. 
 \item In the third case, there is an outlier observation in the data set (Scenario~2) and the estimator $\elpdHat{\Md}{\y}$ is biased. Poor calibration is inevitable with any symmetric approximate distribution. The variance in the uncertainty is overestimated in both cases. However, precise variance estimation would narrow the estimated uncertainty, making it to have worse calibration.
 \item In the last case, the number of observations is small (Scenario~3). The case near the mode illustrates an undesirable overestimation of the uncertainty. The model $\Mb$ has better predictive performance, and the difference is estimated correctly, but the overestimated uncertainty indicates that the sign of the difference is not certain. In the tail, the uncertainty is underestimated, suggesting that the models might have similar predictive performance. In reality, the model $\Mb$ is better.
\end{enumerate}

While inaccurately representing $\p(\elpdHatErr{\Md}{\y})$ in some cases, the obtained approximated $\hat{p}(\elpd{\Md}{\y})$ can be useful in practice if the problematic cases are considered carefully, as discussed in Section~\ref{sec_contributions} and summarised in Section~\ref{sec_conclusions}. More detailed experiments are presented in Section~\ref{sec_experiments}.

\section{Theoretical Analysis using Bayesian Linear Regression} \label{sec_analytic_case}
To study the uncertainty related to the approximation error, we examine it given a normal linear regression model as the known data generating process. Let $\p_\text{true}(\y)$ be 
\begin{gather*}
  \+y = \+X \+\beta + \+\varepsilon,\\
  \+\varepsilon \sim \operatorname{N}\left(\+\mu_\star, \; \+\Sigma_\star\right), \numberthis \label{eq_analytic_data_gen_process}
\end{gather*}
where $\+y \in \mathbb{R}^n$ and $\+X \in \mathbb{R}^{n \times d}$ are the dependent variable and design matrix respectively, $\+\beta \in \mathbb{R}^d$ a vector of the unknown covariate effect parameters, $\varepsilon \in \mathbb{R}^n$ is the vector of errors normally distributed and denoted as residual noise, with underlying parameters $\mu_\star \in \mathbb{R}^n$, and $\+\Sigma_\star \in \mathbb{R}^{n \times n}$ a positive definite matrix, and hence there exist a unique matrix $\+\Sigma_\star^{1/2}$ such that $\+\Sigma_\star^{1/2} \+\Sigma_\star^{1/2} = \+\Sigma_\star$. Let the vector $\+\sigma_\star \in \mathbb{R}^n$ contain the square roots of the diagonal of $\+\Sigma_\star$. The process can be modified to generate outliers by controlling the magnitude of the respective values in $\mu_{\star}$. Under this model, we can analytically study the effect of uncertainty in different situations.

\subsection{Models}
We compare two normal linear regression models \Ma{} and \Mb{}, with subsets of covariates $d_\Ma$ and $d_\Mb$, respectively.
We assume
$d_\Ma \neq d_\Mb$. 
Otherwise, $\elpd{\Md}{\y}$ and $\elpdHat{\Md}{\y}$ would be trivially 0. We write the models $\Mk \in \{\Ma, \Mb\}$ as
\begin{align}\label{eq_analytic_case_models}
 \+y | \widehat{\+\beta}_{d_\Mk}, \+X_{[\bcdot, {d_\Msk}]}, \tau & \sim
 \operatorname{N}\left(X_{[\bcdot,{d_\Msk}]} \widehat{\+\beta}_{d_\Mk}, \; \tau^2\eye\right)\,,
\end{align}
where $\widehat{\+\beta}_{d_\Msk} \in \mathbb{R}^{|d_\Msk|}$ is the respective estimated unknown model parameter. In both models, the noise variance $\tau^2$ is fixed, and a non-informative uniform prior on $\widehat{\+\beta}_{d_\Msk}$ is applied. The resulting posterior and posterior predictive distributions are normal (Appendix~\ref{app_sec_norm_lin_reg_case_study}). Neither model needs to have the same structure as the data generating process.

\subsection{Controlling the Similarity of the Predictive Performances}

Let $\+\beta_\Delta$ denote the coefficients of the data generating model for the non-shared covariates, that is, the covariates included in one model but not the other. If  $\+\beta_\Delta = 0$, both models are similar in the sense that they both include the same model with most non-zero effects, but the noise in the non-effective covariates affects the resulting predictive performance. Situations in which the models are close in predictive performance often arise in practice, for example, in variable selection. As discussed in Section~\ref{sec_problem_setting}, analysing the uncertainty in the model comparison can be problematic in these situations (Scenario~1).

\subsection{Properties for Finite Data}
By applying the specified model setting, data generating process, and score function in $\elpdHat{\Mfd}{\y}$ and $\elpd{\Mfd}{\y}$, we can derive a simplified form for these and for the approximation error $\elpdHatErr{\Mfd}{\y}$. Based on Lemmas~\ref{lem:elpd_chisq} and~\ref{lem_analytic_moments}, we draw some conclusions about their properties and behaviour with finite $n$. The asymptotic behaviour is inspected later in Section~\ref{sec_analytic_case_ninf}. Further details and results are in Appendix~\ref{app_sec_norm_lin_reg_case_study}.

\begin{lemma}\label{lem:elpd_chisq}
 Let the data generating process be as defined in~\eqref{eq_analytic_data_gen_process} and models \Ma{} and \Mb{} be as defined in~\eqref{eq_analytic_case_models}. Given the design matrix $\+X$, the approximation error $\elpdHatErr{\Md}{\y}$ has the following quadratic form:
 \begin{align}
  \elpdHatErr{\Md}{\y}=\+\varepsilon^\transp \+A \+\varepsilon + \+b^\transp \+\varepsilon + c\,,
 \end{align}
 where $\varepsilon$ is the residual noise defined in~\eqref{eq_analytic_data_gen_process}, and for given values of $\+A \in \mathbb{R}^{n \times n}$, $\+b \in \mathbb{R}^{n}$, and $c \in \mathbb{R}$. Similarly, $\elpdHat{\Md}{\y}$ and $\elpd{\Md}{\y}$ have analogous quadratic forms with different values for $\+A, \+b$, and $c$.
\end{lemma}
\begin{proof}
\vspace{-6pt}
 \renewcommand{\BlackBox}{} \renewcommand{\qedsymbol}{}
 See appendices~\ref{app_sec_analytic_elpd}, \ref{app_sec_analytic_loocv}, and~\ref{app_sec_analytic_error}.
 \vspace{-6pt}
\end{proof}
The quadratic factorisation presented in Lemma~\ref{lem:elpd_chisq} allows us to efficiently compute the first moments for the variable of interest $\elpdHatErr{\Md}{\y}$, and therefore analyse properties for finite data in the linear regression case.
\begin{lemma}\label{lem_analytic_moments}
The mean $m_1$, variance $\overline{m}_2$, third central moment $\overline{m}_3$, and skewness $\widetilde{m}_3$ of the variable of interest $Z = \elpdHatErr{\Md}{\y}$ presented in Lemma~\ref{lem:elpd_chisq} for a given covariate matrix $\+X$ are
\begin{align*}
  m_1 &= \E\left[Z\right]\\
  &= \operatorname{tr}\left(\+\Sigma_\star^{1/2}\+A\+\Sigma_\star^{1/2}\right)
  + c
  + \+b^\transp \+\mu_\star
  + \+\mu_\star^\transp \+A \+\mu_\star
  \numberthis
  \label{eq_analytic_1_moment}
  \\
  \overline{m}_2 &= \Var\left[Z\right]
  \\
  &= 2 \operatorname{tr}\left(\left(\+\Sigma_\star^{1/2}\+A\+\Sigma_\star^{1/2}\right)^2\right)
    + \+b^\transp \+\Sigma_\star \+b + 4 \+b^\transp \+\Sigma_\star \+A \+\mu_\star + 4 \+\mu_\star^\transp \+A \+\Sigma_\star \+A \+\mu_\star
  \numberthis
  \label{eq_analytic_2_moment}
  \\
  \overline{m}_3 &= \E\Big[\left(Z-\E\left[Z\right]\right)^3\Big]
  \\
  &= 8 \operatorname{tr}\left(\left(\+\Sigma_\star^{1/2}\+A\+\Sigma_\star^{1/2}\right)^3\right)
    + 6 \+b^\transp \+\Sigma_\star \+A \+\Sigma_\star \+b + 24 \+b^\transp \+\Sigma_\star \+A \+\Sigma_\star \+A \+\mu_\star + 24 \+\mu_\star^\transp \+A \+\Sigma_\star \+A \+\Sigma_\star \+A \+\mu_\star
  \numberthis
  \label{eq_analytic_3_moment}
  \\ 
  \widetilde{m}_3 &= \overline{m}_3 / (\overline{m}_2)^{3/2}
  \numberthis
  \label{eq_analytic_3_moment_skew}\,.
\end{align*}
\end{lemma}
\begin{proof}
\vspace{-6pt}
 See Appendix~\ref{app_sec_analytic_moments}.
 \renewcommand{\BlackBox}{} \renewcommand{\qedsymbol}{} 
\vspace{-6pt}
\end{proof}

\para{No Effect by the Shared Covariates}
The distributions of $\elpdHat{\Md}{\y}$, $\elpd{\Md}{\y}$, and the error do not depend on the commonly shared covariate effects $\+\beta_\text{shared}$. For example, if an intercept is included in both models, the intercept coefficient does not affect the comparison. We summarise this in the following proposition:
\begin{proposition}\label{cor_analytic_not_depend_shared_covariates}
 The distribution of the variables of interest presented in Lemma~\ref{lem:elpd_chisq} do not depend on the commonly shared covariate effects $\+\beta_\text{shared}$.
\end{proposition}
\begin{proof}
\vspace{-6pt}
\renewcommand{\BlackBox}{} \renewcommand{\qedsymbol}{}
 See appendices~\ref{app_sec_analytic_elpd_d}, \ref{app_sec_analytic_loocv_d}, and~\ref{app_sec_analytic_error}.
\vspace{-6pt}
\end{proof}

\para{Non-Shared Covariates}
The skewness of the distribution of the error $\elpdHatErr{\Md}{\y}$ will asymptotically converge to 0 when the models \Ma{} and \Mb{} become more dissimilar (the magnitude of the effects of the non-shared covariates $\+\beta_\Delta$ grows). The larger the difference, the better a normal distribution approximates the uncertainty. If the models capture the true data generating process comprehensively (no outliers in the data), and all covariates are included in at least one of the models, then the skewness of the error has its extremes when the models are, more or less, identical in predictive performance (around $\+\beta_\Delta=0$). We summarise this result in the following proposition.
\begin{proposition}\label{pro_analytic_beta_to_inf}
 Consider skewness $\widetilde{m}_3$ for variable $\elpdHatErr{\Md}{\y}$.
 Let $\+\beta_\Delta = \beta_\mathrm{r} \+\beta_\text{rate} + \+\beta_\text{base}$,
 where $\beta_\mathrm{r} \in \mathbb{R}$, $\+\beta_\text{rate} \in \mathbb{R}^{k} \setminus \{0\}$, $\+\beta_\text{base} \in \mathbb{R}^{k}$, and $k$ is the number of non-shared covariates.
 Now,
 \begin{align} \lim_{\beta_\mathrm{r} \rightarrow \pm \infty} \widetilde{m}_3 = 0\,.
 \end{align}
 Furthermore, if $\mu_\star = 0$, $\+\beta_\text{base} = 0$, and ${d_\Msa}\cup{d_\Msb} = \mathbb{U}$, $\widetilde{m}_3$ as a function of $\beta_\mathrm{r}$ is a continuous even function with extremes at $\beta_\mathrm{r} = 0$ and situational at $\beta_\mathrm{r} = \pm r$, where the definition of the latter extreme and the condition for their existence are given in Appendix~\ref{app_sec_analytic_effect_of_model_difference}.
\end{proposition}
\begin{proof} 
\vspace{-6pt}
\renewcommand{\BlackBox}{} \renewcommand{\qedsymbol}{}
 See Appendix~\ref{app_sec_analytic_effect_of_model_difference}.
\vspace{-6pt}
\end{proof}
The behaviour of the moments with regard to the non-shared covariates' effects is illustrated graphically in Figures~\ref{fig_analytic_skew_b_n} and~\ref{fig_analytic_zscore_skew_n_b}. Figure~\ref{fig_analytic_skew_b_n} shows that the problematic skewness near $\+\beta_\Delta=0$ occurs particularly with nested models. Similar behaviour can be observed with unconditional design matrix $\+X$ in Figure~\ref{fig_analytic_zscore_skew_n_b_tot} in Appendix~\ref{app_sec_analytic_graph}, and additionally with unconditional model variance $\tau$ in the simulated experiment results in Section~\ref{sec_experiments}. In a non-nested comparison setting, problematic skewness near $\+\beta_\Delta=0$ occurs, particularly when there is a difference in the effects of the included covariates between the models.
\begin{figure}[tb!]
 \centering
 \includegraphics[width=0.9\figurecontrolwidth]{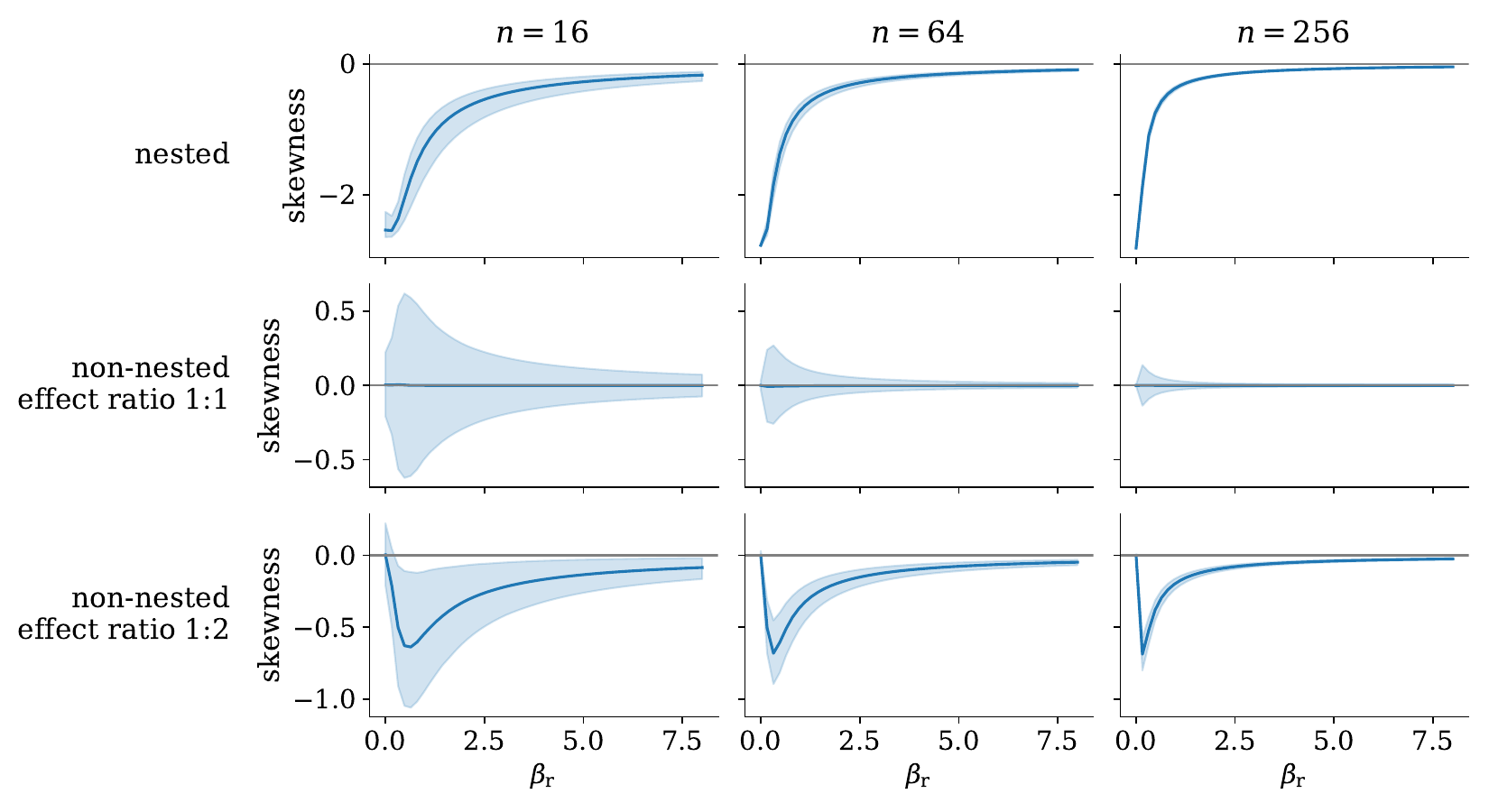}
 \caption{%
 The skewness conditional on the design matrix $\+X$ for the error $\elpdHatErr{\Md}{\y}$ as a function of a scaling factor $\beta_\mathrm{r} \in \mathbb{R}$ for the magnitude of the non-shared effects: $\+\beta_\Delta = \beta_\mathrm{r} \+\beta_\text{rate}$. Models have an intercept and one shared covariate. Top: the model $\Mb$ has one additional covariate. Middle: models $\Ma$ and $\Mb$ each have one additional covariate with equal effects. Bottom: models $\Ma$ and $\Mb$ have one additional covariate with an effect ratio of 1:2. The solid lines correspond to the median, and the shaded area illustrates the 95 \% interval based on 2000 simulated $\+X$s. The problematic skewness occurs, particularly with the nested models (top row) when $\beta_r$ is close to 0 so that the models make similar predictions (Scenario~1). In the non-nested case, the extreme skewness decreases when $n$ grows, more noticeably in the case of equal effects, but in the nested case, the extreme skewness stays high when $n$ grows.
 }\label{fig_analytic_skew_b_n}
\end{figure}
\begin{figure}[tb!] \centering
 \includegraphics[width=0.8\figurecontrolwidth]{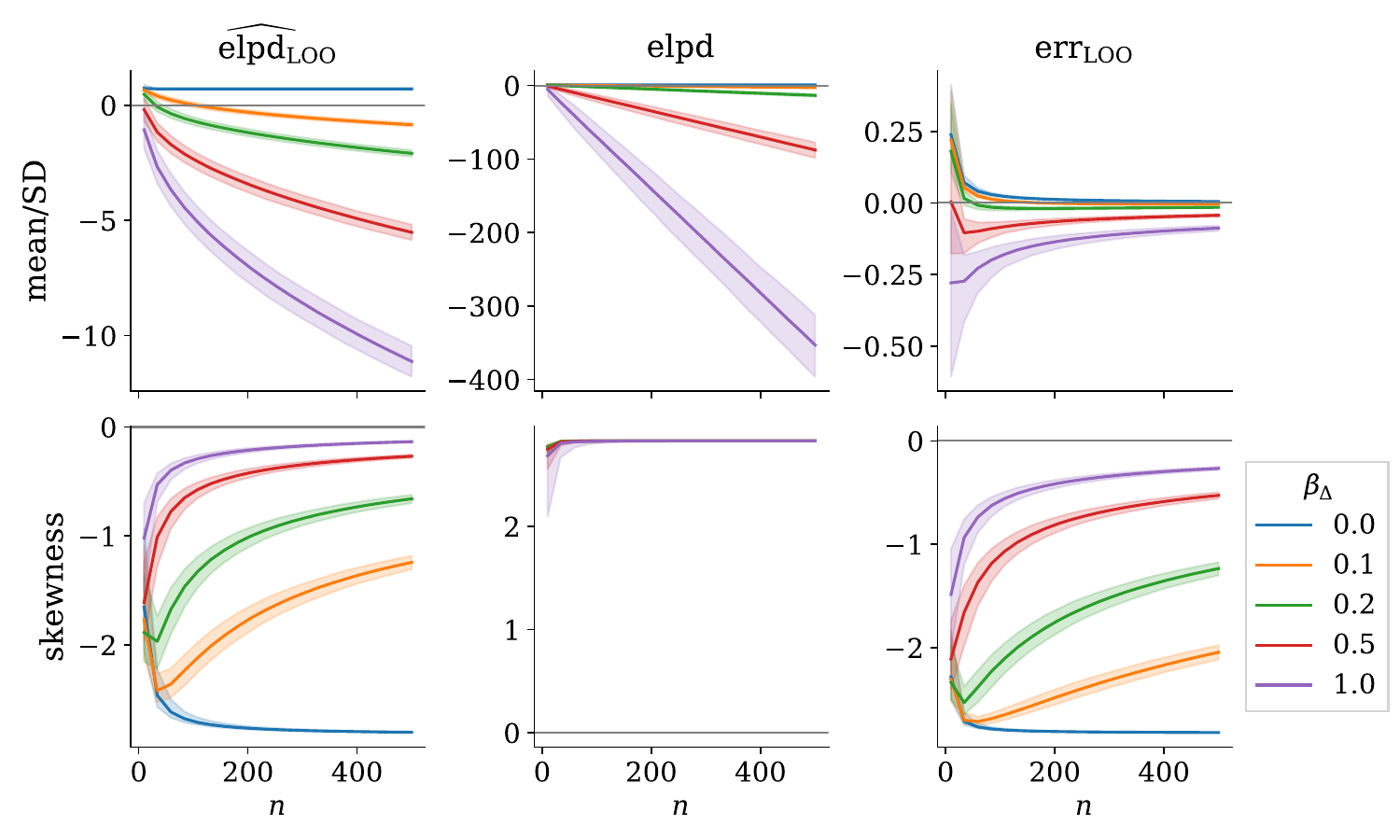}
 \caption{The mean relative to the standard deviation and skewness conditional on the design matrix $\+X$ for $\elpdHat{\Md}{\y}$, $\elpd{\Md}{\y}$, and for the error $\elpdHatErr{\Md}{\y}$ as a function of the data size $n$. The relative mean serves as an indicator of how far away the distribution is from 0. The true model has an intercept and two covariates. One of the covariates with true effect $\beta_\Delta$ is included only in model $\Mb$. The solid lines correspond to the median, and the shaded area illustrates the 95\% interval based on 2000 simulated $\+X$s. The problematic skewness of the error occurs with small $n$ and $\beta_\Delta$. When $\beta_\Delta = 0$, the magnitude of skewness does not decrease when $n$ grows. The relative mean of the error approaches zero when $n$ grows.}
\label{fig_analytic_zscore_skew_n_b}
\end{figure}

\para{Outliers}
Outliers in the data impact the moments of the distribution of the error $\elpdHatErr{\Md}{\y}$ in a fickle way. 
Depending on the data $\+X$, covariate effect vector $\+\beta$, and on the outlier design vector $\mu_\star$, scaling the outliers can affect the bias of the error quadratically, linearly or not at all. The variance is affected quadratically or not at all. If the scaling affects the variance, the skewness asymptotically converges to zero. We summarise these results in the following proposition.
\begin{proposition}\label{prop_analytic_effect_of_mu}
 Consider the mean $m_1$, the variance $\overline{m}_2$, and the third central moment $\overline{m}_3$ for the variable $\elpdHatErr{\Md}{\y}$.
 Let $\+\mu_\star = \mu_{\star,\mathrm{r}} \+\mu_{\star,\text{rate}} + \+\mu_{\star,\text{base}}$,
 where $\mu_{\star,\mathrm{r}} \in \mathbb{R}$, $\+\mu_{\star,\text{rate}} \in \mathbb{R}^{n} \setminus \{0\}$, and $\+\mu_{\star,\text{base}} \in \mathbb{R}^{n}$. Now $m_1$ is a second or first-degree polynomial or constant as a function of $\mu_{\star,\mathrm{r}}$. Furthermore, $\overline{m}_2$ and $\overline{m}_3$ are either both second-degree polynomials or both constants and thus, if not constant, the skewness
 \begin{align}
   \lim_{\mu_{\star,\mathrm{r}} \rightarrow \pm \infty} \widetilde{m}_3 = \lim_{\mu_{\star,r} \rightarrow \pm \infty} \frac{\overline{m}_3}{(\overline{m}_2)^{3/2}} = 0 \,.
 \end{align}
\end{proposition}
\begin{proof}
\vspace{-6pt}
\renewcommand{\BlackBox}{} \renewcommand{\qedsymbol}{} 
 See Appendix~\ref{app_sec_analytic_effect_of_outliers}.
\vspace{-6pt}
\end{proof}
As demonstrated in Figure~\ref{fig_analytic_zscore_skew_mu_b}, while the skewness decreases, the relative bias increases, and the approximation gets increasingly bad. When $\+\mu_\star \neq 0$, the problematic skewness of the error $\elpdHatErr{\Md}{\y}$ may occur with any level of non-shared covariate effects $\+\beta_\Delta$. This behaviour is shown in Appendix~\ref{app_sec_analytic_three_skew}.

\para{Residual Variance} 
The skewness of the error $\elpdHatErr{\Md}{\y}$ converges to a constant value when the true residual variance grows. 
When the observations are uncorrelated, and they have the same residual variance so that $\Sigma_\star = \sigma_\star^2 \eye$, the skewness converges to a constant, determined by the design matrix $\+X$ when $\sigma_\star^2 \rightarrow \infty$. We summarise this behaviour in the following proposition.
\begin{proposition}\label{prop_analytic_effect_residual_var}
 For the data generating process defined in Equation~\eqref{eq_analytic_data_gen_process}, let $\+\Sigma_\star = \sigma_\star^2 \eye_{n}$, and consider the skewness $\widetilde{m}_3$ for the variable $\elpdHatErr{\Md}{\y}$.
 Then,
 \begin{align}
    \lim_{\sigma_\star \rightarrow \infty} \widetilde{m}_3 = 2^{3/2} \frac{
        \operatorname{tr}\left(\+A_\mathrm{err}^3\right)}{
        \operatorname{tr}\left(\+A_\mathrm{err}^2\right)^{3/2}}\,.
 \end{align}
\end{proposition}
\begin{proof}
\vspace{-6pt}
\renewcommand{\BlackBox}{} \renewcommand{\qedsymbol}{} 
 See Appendix~\ref{app_sec_analytic_effect_of_data_variance}.
\vspace{-6pt}
\end{proof}
\begin{figure}[tb!]\centering
 \includegraphics[width=0.74\figurecontrolwidth]{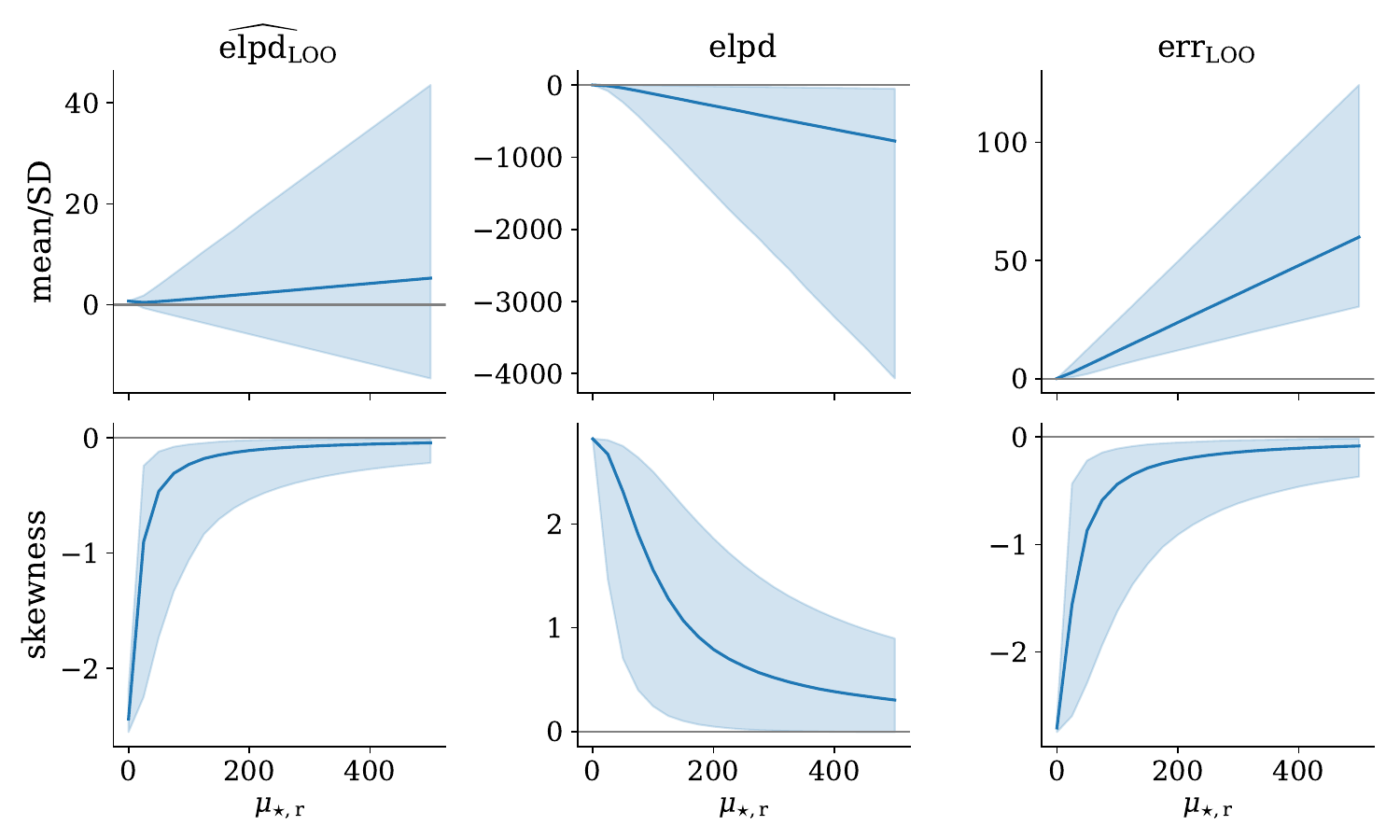}
 \caption{Illustration of the mean relative to the standard deviation and skewness conditional on the design matrix $\+X$ for $\elpdHat{\Md}{\y}$, $\elpd{\Md}{\y}$, and for the error $\elpdHatErr{\Md}{\y}$ as a function of a scaling factor $\mu_{\star,\mathrm{r}}$ for the magnitude of one outlier observation. The data consists of an intercept and two covariates, one of which has no effect and is considered only in the model $\Mb$. The illustrated behaviour is also similar to other levels of effect for the non-shared covariate. The solid lines correspond to the median, and the shaded area illustrates the 95 \% confidence interval based on 2000 independently simulated $\+X$s from the standard normal distribution. The skewness of all the inspected variables approaches zero when $n$ grows. However, at the same time, the bias of the estimator increases, thus making the analysis of the uncertainty hard.}
\label{fig_analytic_zscore_skew_mu_b}
\end{figure}
\subsection{Asymptotic Behaviour as a Function of the Data Size} \label{sec_analytic_case_ninf}
Following the setting defined in~\eqref{eq_analytic_data_gen_process} and~\eqref{eq_analytic_case_models}, by inspecting the moments in an example case, where a null model is compared to a model with one covariate, we can further draw some interesting conclusions about the behaviour of the moments when $n \rightarrow \infty$, namely:
\begin{proposition}\label{prop:elpd_asymptotic} 
 Let the setting be defined as in~\eqref{eq_analytic_data_gen_process} and~\eqref{eq_analytic_case_models}. In addition, let $\beta_\Delta \in \mathbb{R}$ be the true effect of the sole non-shared covariate that controls the similarity of the model performances, $\tau^2$ is the model variance, and $\Sigma_\star = s_\star^2 \eye$ is the true residual variance, then
\begingroup
 \allowdisplaybreaks
 \begin{align}
  \lim_{n\rightarrow\infty} \frac{\E\left[\elpdHat{\Md}{\y}\right]}{\SD\left[\elpdHat{\Md}{\y}\right]}
    &= \begin{dcases}
        \frac{\tau^2}{\sqrt{2} s_\star^2} \,, & \text{when } \beta_\Delta = 0 \,, \\
        - \infty & \text{otherwise,}
    \end{dcases}
  \\
     \lim_{n\rightarrow\infty} \frac{\E\left[\elpd{\Md}{\y}\right]}{\SD\left[\elpd{\Md}{\y}\right]}
    &= \begin{dcases}
        \frac{\tau^2}{\sqrt{2} s_\star^2} \,, & \text{when } \beta_\Delta = 0 \,, \\
        - \infty & \text{otherwise,}
    \end{dcases}
  \\
  \lim_{n\rightarrow\infty} \frac{\E\left[\elpdHatErr{\Md}{\y}\right]}{\SD\left[\elpdHatErr{\Mfd}{\y}\right]}
    &= 0
  \\
    \lim_{n\rightarrow\infty} \skewness\left[\elpdHatErr{\Md}{\y}\right]
    &= \begin{dcases}
        - 2^{3/2}, & \text{when }  \beta_\Delta = 0 \\
        0 & \text{otherwise,}
    \end{dcases}
 \end{align}
\endgroup
\end{proposition}
\begin{proof}
\vspace{-6pt}
\renewcommand{\BlackBox}{} \renewcommand{\qedsymbol}{}
 See Appendices~\ref{app_sec_analytic_one_Rel_SD},\, \ref{app_sec_analytic_two_Rel_SD},\, \ref{app_sec_analytic_three_Rel_SD}, and \ref{app_sec_analytic_three_skew}.
\vspace{-6pt}
\end{proof}
When $\beta_\Delta=0$, the relative means of both $\elpdPlain$ and $\elpdHatPlain$ converge to the same non-zero value, that is the simpler, more parsimonious model performs better asymptotically. As a comparison, in the non-Bayesian linear regression setting with squared error inspected by \citet{shao1993linear}, both models have asymptotically equal predictive performance with all $\yobs$. Similarly, when $\beta_\Delta=0$, the skewness of the error converges to a non-zero value, which indicates that analysing the uncertainty will be problematic also with big data for models with very similar predictive performance (Scenario~1).

Even though we do not expect an underlying effect of a non-shared covariate to be precisely zero in practice, the analysed moments may still behave similarly even with large data size when the effect size is small enough. When $\beta_\Delta \neq 0$, the relative mean of both $|\elpdPlain|$ and $|\elpdHatPlain|$ grows infinitely, and the skewness of the error converges to zero; the more complex model performs better in general, and the problematic skewness hinders when more data is available. The relative mean of error converges to zero, regardless of $\beta_\Delta$. Hence, the approximation bias decreases with more data in any case. The example case and the behaviour of the moments are presented in more detail in Appendix~\ref{app_sec_one_cov}. 
Nevertheless, the case analysis shows that a simpler model can outperform a more complex one asymptotically. In addition, the skewness of the error can be problematic also with big data.
\section{Simulation Experiments}
\label{sec_experiments}

In this section, we present simulation results testing whether the analytic results for a simplified model empirically generalise for more commonly used models.
Similar to the theoretical analysis in Section~\ref{sec_analytic_case}, we analyse the finite sample properties of the estimator $\elpdHat{\Md}{\y}$, of the $\elpd{\Md}{\y}$, and of the error $\elpdHatErr{\Md}{\y}$ for the similarity of model performances (Scenario~1), model misspecification through the effect of an outlier observation (Scenario~2), and the effect of the sample size $n$ (Scenario~3). We also inspect the calibration of the uncertainty estimates. The source code for the experiments is available at \url{https://github.com/avehtari/loocv_uncertainty}.

First, we consider normal linear regression, but without conditioning on the design matrix $\+X$ and the model variance $\tau^2$.
We compare two nested linear regression models under data simulated from a linear regression model being $\p_\text{true}(\y)$. The data generating process follows the definition in~\eqref{eq_analytic_data_gen_process}, where $d=3$, $\+X_i = [1, X_{[i,2]}, X_{[i,3]}]$, $X_{[i,1]}, X_{[i,2]} \sim \N(0, 1)$ for $i = 1,2,\dots,n$, $\+\beta = [0, 1, \beta_\Delta]$, $\mu_\star = [\mu_{\star,0}, 0, \dots, 0]$, and, $\Sigma_\star = \eye$. The models \Ma{} and \Mb{} follow the definition in~\eqref{eq_analytic_case_models} with the difference that the residual variance $\tau^2$ is now unknown and the prior is noninformative uniform on ($\widehat{\+\beta}_{d_\Msk}$, $\log \tau^2$) (all the posteriors are proper). The model \Ma{} only includes intercept and one covariate, while the model \Mb{} includes one additional covariate.
The similarity of the models is varied by varying $\beta_\Delta$ in data generation.
The data size $n$ varies from $16$ to $1024$. Parameter $\mu_{\star,0}$ is used to scale the mean of one observation so that, when large enough, that observation becomes an outlier and the models become misspecified. Unless otherwise noted, $\mu_{\star,0} = 0$.
We generate 
2000 data sets from $\p_\text{true}(\y)$, and for each trial, 
we obtain pointwise LOO-CV estimates $\elpdHati{\Ma}{\yobs}{i}$ and $\elpdHati{\Mb}{\yobs}{i}$, which are used to form estimates $\elpdHat{\Md}{\yobs}$ and $\seHat{\Md}{\yobs}$ in particular. The respective target values $\elpd{\Ma}{\yobs}$ and $\elpd{\Mb}{\yobs}$ are obtained using an independent test set of 4000 data sets of the same size simulated from the same data generating process.

\para{Behaviour of the Sampling and the Error Distribution}
The moments of the sampling distribution of $\elpdHat{\Md}{\y}$, the distribution of the $\elpd{\Md}{\y}$, and the error distribution $\elpdHatErr{\Md}{\y}$ behave similarly in these simulated experiments and in the theoretical analysis conditional on the design matrix $\+X$ and known model variance $\tau$ (Section~\ref{sec_analytic_case}). In particular, when $\beta_\Delta=0$ and $n$ grow, LOO-CV is slightly more likely to pick the simpler model with a constant difference in the predictive performance, and the magnitude of the skewness does not fade away. With this experiment setting, however, the skewness of $\elpd{\Md}{\y}$ decreases when $\beta_\Delta$ grows, while in the experiments in Section~\ref{sec_analytic_case}, this skewness is similar with all $\beta_\Delta$. Figure~\ref{fig_moments_n_b} in Appendix~\ref{app_sec_additional_experiment_results} illustrates the behaviour of the moments in more detail.

\para{Negative Correlation and Bias}
\label{sec_exp_negative_correlation}
Figure~\ref{fig_joint} illustrates the joint distribution of $\elpdHat{\Md}{\y}$ and $\elpd{\Md}{\y}$ for various non-shared covariate effects $\beta_\Delta$ and data sizes $n$. The estimator and $\elpd{\Md}{\y}$ get negatively correlated when the model performances get more similar (Scenario~1). The effect is more noticeable with larger $n$. Similar to Figure~\ref{fig_joint}, Figure~\ref{fig_joint_out} in Appendix~\ref{app_sec_additional_experiment_results} illustrates the joint distribution of $\elpdHat{\Md}{\y}$ and $\elpd{\Md}{\y}$ when there is an outlier observation in the data set (Scenario~2).
\begin{figure}[tb!] \centering
 \includegraphics[width=0.73\figurecontrolwidth]{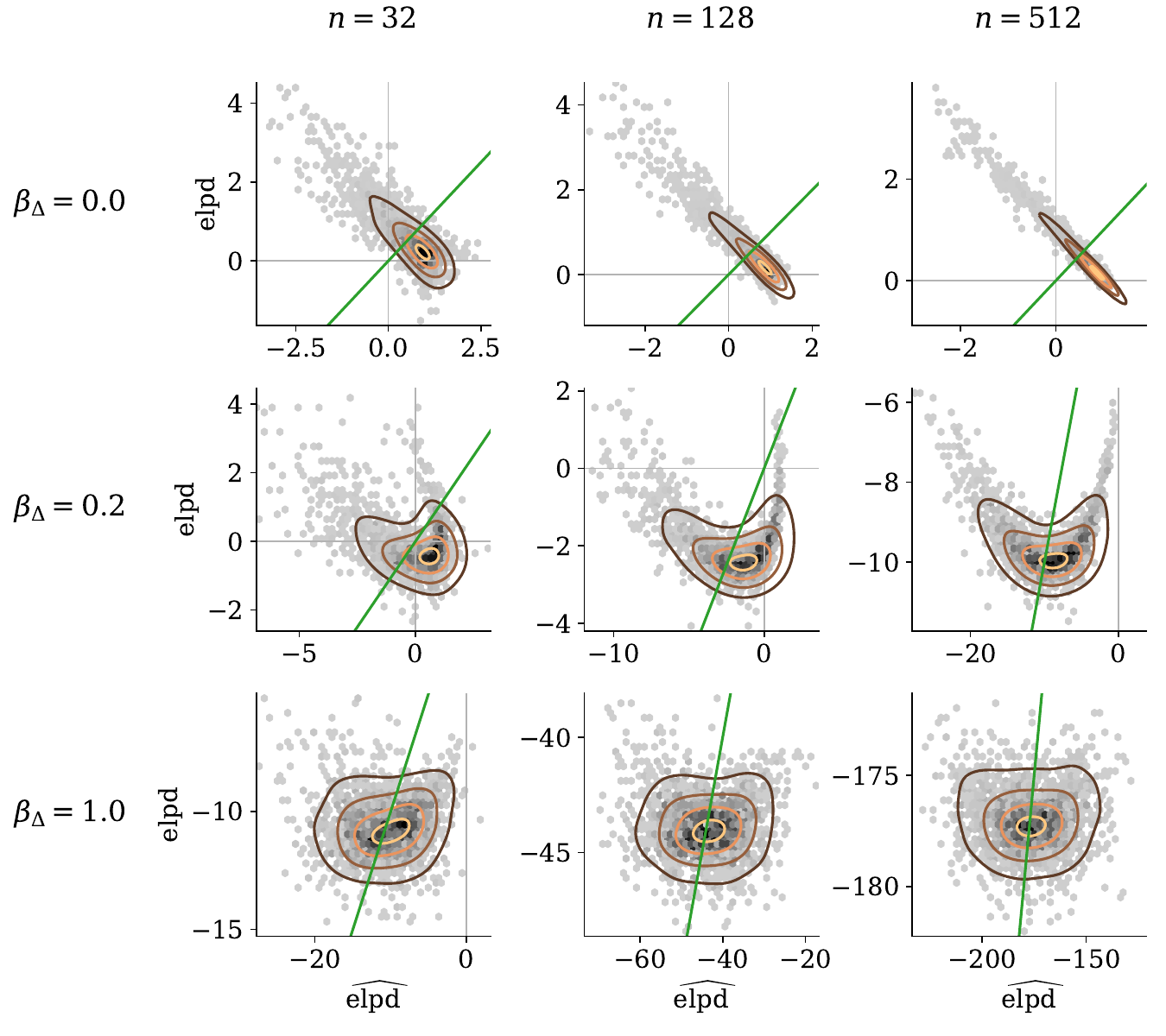}
 \vspace{-5pt}
 \caption{Illustration of the joint distribution of $\elpdHat{\Md}{\y}$ and $\elpd{\Md}{\y}$ for various data sizes $n$ and non-shared covariate effects $\beta_\Delta$. The green diagonal line indicates where $\elpdHat{\Md}{\y}=\elpd{\Md}{\y}$. The gray horizontal and vertical lines indicate $\elpdHat{\Md}{\y}=0$ and $\elpd{\Md}{\y}=0$.  Note that the axes ranges are different in each subplot, and $0$ can be outside the axes range. The problematic negative correlation occurs when $\beta_\Delta = 0$. In addition, while decreasing correlation, the nonlinear dependency in the transition from small to large $\beta_\Delta$ is problematic. In the ideal case, the distribution is centered on the green line, there is no correlation, and the distribution is close to normal.
 }\label{fig_joint}
\end{figure}
Figures~\ref{fig_err_n_b_both} and~\ref{fig_errdirection_n_b_both} in Appendix~\ref{app_sec_additional_experiment_results} show that when there is an outlier present in the data, the relative error's mean usually clearly deviates from zero, and the estimator is biased.

\para{Behaviour of the Uncertainty Estimates}
Due to the mismatch between the sampling and error distribution forms, estimated uncertainties based on the sampling distribution can be poorly calibrated.
The problem of underestimating the variance is illustrated in Figures~\ref{fig_var_ratioi_n_b} and~\ref{fig_var_ratioi_n_b_out} in Appendix~\ref{app_sec_additional_experiment_results}.
Figure~\ref{fig_calib_n_b_both} illustrates the calibration of the estimated uncertainties in different settings.
\begin{figure}[t!]
 \centering
 \includegraphics[width=0.88\figurecontrolwidth]{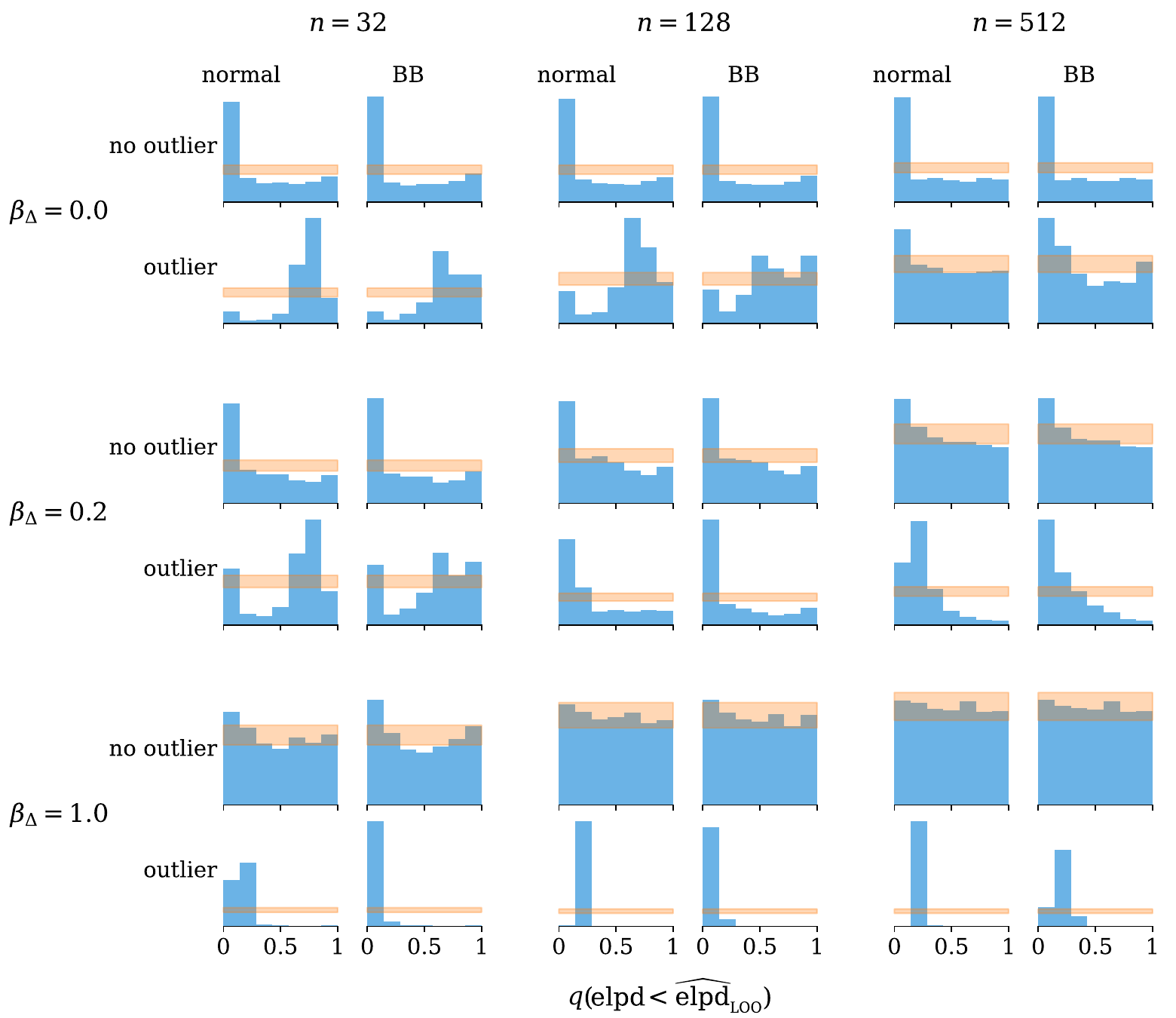}
 \caption{Calibration of the estimated uncertainty $\p(\elpd{\Md}{\y})$ for various data sizes $n$ and non-shared covariate effects $\beta_\Delta$. The histograms illustrate the PIT values $q \Bigl(\elpd{\Md}{\y} < \elpdHat{\Md}{\y}\Bigr)$ over simulated data sets $\y$, which would be uniform in a case of optimal calibration. The yellow shading indicates the range of 99 \% of the variation expected from uniformity. Two uncertainty estimators are presented: normal approximation and BB. The outlier observation has a deviated mean of 20 times the standard deviation of $\y_i$. The calibration is better when $\beta_\Delta$ is large or $n$ is big. The outlier makes the calibration worse, although with large $n$ and small $\beta_\Delta$, the calibration can be better, as outliers inflates the variance.}\label{fig_calib_n_b_both}
\end{figure}
Normal and BB approximations produce similar results. A small sample size (Scenario~3) and similarity in the predictive performance between the models (Scenario~1) can cause problems. Similarly, model misspecification through an outlier observation (Scenario~2) can worsen calibration. On the other hand, in the experiment with $n=512$ and $\beta_\Delta = 0$, the calibration is better with an outlier as the outlier inflates the variance. The skewness of the error has decreased more than the bias has increased. This effect is illustrated in more detail in Figures~\ref{fig_analytic_zscore_skew_mu_b} and \ref{fig_moments_n_b} in Appendix~\ref{app_sec_additional_experiment_results}.

\para{Additional simulations}
Appendix~\ref{subsec_additional_experiments} presents additional simulation results illustrating that the theoretical results generalise beyond the simplest case to models with more covariates, non-Gaussianity, hierarchy, and splines, and cases with fixed covariate values and $K$-fold-CV.

\section{Case studies}

We demonstrate the use of uncertainty quantification of the predictive performance difference with three real-data examples. We assume that the true data generating processes are more complex than the models used. We cover all three scenarios (very similar predictions, model misspecification, and small data) that can affect how well calibrated the normal approximation.

Inference was made using Markov chain Monte Carlo (MCMC) with 4 chains with 1000 warmup and 1000 sampling iterations. Convergence diagnostics \citep{Vehtari-Gelman-Simpson-etal:2021}, using the \texttt{posterior} package, \citep{Burkner-Gabry-Kay-etal:2024} indicated reliable posterior inference. For LOO-CV we used the \texttt{loo} package \citep{vehtari2022loopkg}, which uses fast PSIS-LOO \citep{Vehtari+Gelman+Gabry:2017_practical} for computation. 

\para{Primate milk}
\citet{mcelreath2020statistical} describes the primate milk data: \textit{``A popular hypothesis has it that primates with larger brains produce more energetic milk, so that brains can grow quickly... The question here is to what extent energy content of milk, measured here by kilocalories, is related to the percent of the brain mass that is neocortex... We'll end up needing female body mass as well, to see the masking that hides the relationships among the variables.''} The data include 17 different primate species. The target variable is the energy content of milk (kcal.per.g) and the covariates are the percent of the brain mass that is neocortex (neocortex) and the logarithm of female body mass (log(mass)). The covariates and target are centered and scaled to have unit variance.

We use the following four models,  fitted with the \texttt{rstanarm} package \citep{Goodrich-Gabry-Ali-etal:2024}, and using weakly informative $\normal(0, 1)$ priors for the coefficients and an exponential(1) prior for the residual scale:
\begin{align*}
  \mathrm{M}_1: \quad & \mathrm{kcal.per.g} \sim \normal(\alpha, \sigma) \\
  \mathrm{M}_2: \quad& \mathrm{kcal.per.g} \sim \normal(\alpha + \beta_1\times\mathrm{neocortex}, \sigma) \\
  \mathrm{M}_3: \quad & \mathrm{kcal.per.g} \sim \normal(\alpha + \beta_2\times\mathrm{log(mass)}, \sigma) \\
  \mathrm{M}_4: \quad & \mathrm{kcal.per.g} \sim \normal(\alpha + \beta_1\times\mathrm{neocortex} + \beta_2\times\mathrm{log(mass)}, \sigma).
\end{align*}

We compare models $\mathrm{M}_2,\mathrm{M}_3,\mathrm{M}_4$ to the intercept-only model $\mathrm{M}_1$:

\begin{center}
{\small
 \begin{tabular}{lccc}
   Model & $\elpdHat{{\mathrm{M}_1,\mathrm{M}_k}}{\y}$ & $\seHatPlain$ & $\hat{p}(\elpd{{\mathrm{M}_1,\mathrm{M}_k}}{\y}>0)$ \\
 \midrule
   $\mathrm{M}_1$ &  - &  - &  - \\
   $\mathrm{M}_2$ & -0.6 & 0.6 & 0.16 \\
   $\mathrm{M}_3$ &  0.3 & 1.2 & 0.60 \\
   $\mathrm{M}_4$ &  4.2 & 2.4 & 0.96 
 \end{tabular}\\
}
\end{center}

Based on model checking and the distribution of pointwise $\elpdHati{\Mk}{\yobs}{i}$, the models seem to be reasonably specified and we are fine with respect to Scenario 2 (model misspecification).
Models $\mathrm{M}_2$ and $\mathrm{M}_3$ have very small $\elpdHat{\Md}{\yobs}$ compared to model $\mathrm{M}_1$. The direct use of the normal approximation gives probabilities 0.16 and 0.6 that these models have better predictive performance than model $\mathrm{M}_1$. As $\elpdHat{\Md}{\yobs}$ is small (Scenario 1) and the number of observations is small (Scenario 3), we may assume  $\seHat{\Md}{\yobs}$ to be underestimated and the error distribution to be more skewed than normal. However, since $\elpdHat{\Md}{\yobs}$ is small, we can state that there is no practical or statistical difference in the predictive performance.

The direct use of $\elpdHat{{\mathrm{M}_4,\mathrm{M}_1}}{\yobs}$ and $\seHat{{\mathrm{M}_4,\mathrm{M}_1}}{\yobs}$ would give probability $0.96$ that model $\mathrm{M}_4$ has better predictive than model $\mathrm{M}_1$. This difference (4.2) is big enough that we are fine with respect to Scenario 1, but the number of observations is small (Scenario 3), and on expectation we may assume $\seHat{{\mathrm{M}_4,\mathrm{M}_1}}{\yobs}$ to be underestimated. If we multiply $\seHat{{\mathrm{M}_4,\mathrm{M}_1}}{\yobs}$ by 2 \citep[heuristic based on the limit of equations by][]{bengio_Grandvalet_2004} to make a more conservative estimate, the probability that model $\mathrm{M}_4$ has better predictive performance is bigger than 0.81. Considering we have only 17 observations, this is quite good. Collecting more data is, however, recommended.

As the predictive distribution includes the aleatoric uncertainty (modelled by the data model), there is often more uncertainty in the predictive performance model comparison than in the posterior distribution \citep[see, e.g., ][]{Wang-Gelman:2015}. In simple models, we can also look at the posterior for the quantities of interest. With model $\mathrm{M}_4$, $95\%$ central posterior intervals for $\beta_1$ and $\beta_2$ are $(1.1,3.7)$ and $(-0.12,-0.04)$ respectively, which indicates data have information about the parameters. The covariates neocortex and log(mass) are collinear, which causes correlation in the posterior of the coefficients, which could make the marginal posteriors overlap 0, even if the joint posterior does not, in which case, looking at the predictive performance is useful. In this case, although neocortex and log(mass) are collinear, they don't have useful information alone, and the useful predictive information is along the second principal component of their joint distribution, which explains why the models with only one of the covariates are not better than the intercept-only model.

\para{Sleep study}
\citet{Belenky-Wesensten-Thorne-etal:2003} collected data on the effect of chronic sleep restriction. We use a subset of data in the R package \texttt{lme4} \citep{Bates-Maechler-Bolker-etal:2015}. The data contains average reaction times (in milliseconds) for 18 subjects with sleep restricted to 3 hours per night for 7 consecutive nights (days 0 and 1 were adaptation and training and removed from this analysis).

The compared models are a linear model, a linear model with varying intercept for each subject, and a linear model with varying intercept and slope for each subject. All models use a normal data model. The models were fitted using \texttt{brms} \citep{Burkner:2017}, and the default \texttt{brms} priors; prior for the coefficient for Days is uniform, the prior for the varying intercept is normal with unknown scale having a half-normal prior, and the prior for the varying intercept and slope is bivariate normal with unknown scales having half-normal priors and correlation having LKJ prior \citep{Lewandowski-Kurowicka-Joe:2009}. 

Using \texttt{brms} formula notation, the compared models are
{\small
\begin{align*}
  \mathrm{M}_1: \quad & \mathrm{Reaction} \sim \mathrm{Days} \\
  \mathrm{M}_2: \quad & \mathrm{Reaction} \sim \mathrm{Days} + (1\,\mid\,\mathrm{Subject}) \\
  \mathrm{M}_3: \quad & \mathrm{Reaction} \sim \mathrm{Days} + (\mathrm{Days}\,\mid\, \mathrm{Subject}).
\end{align*}
\vspace{-\baselineskip}
}

Based on the study design, $\mathrm{M}_3$ is the appropriate model for the analysis, but comparing models is useful for assessing how much information the data has about the varying intercepts and slopes.  For a few LOO-folds with high Pareto-$\hat{k}$ diagnostic value \citep[$>0.7$, ][]{Vehtari+etal:2024:PSIS} we re-ran MCMC (with \texttt{reloo=TRUE} in \texttt{brms}).
\begin{center}
{\small
 \begin{tabular}{lccc}
   Model & $\elpdHat{{\mathrm{M}_3,\mathrm{M}_k}}{\y}$ & $\seHatPlain$ & $\hat{p}(\elpd{{\mathrm{M}_3,\mathrm{M}_k}}{\y}>0)$ \\
   \midrule
   $\mathrm{M}_3$ &  - &  - &  - \\
   $\mathrm{M}_2$ & -12.7 &  9.8 & 0.90 \\
   $\mathrm{M}_1$ & -77.8 & 20.9 & 0.9999
 \end{tabular}\\
}
\vspace{-1pt}
\end{center}
Model $\mathrm{M}_3$ is estimated to have better predictive performance, but only with 0.9 probability of having better performance than model $\mathrm{M}_2$. Model-checking reveals that two observations are clear outliers with respect to these models, making the normal approximation likely to be poorly calibrated (Scenario 3).

We also fitted models using a Student's $t$ model to create models $\mathrm{M}_{1t}$, $\mathrm{M}_{2t}$, and $\mathrm{M}_{3t}$. Based on model checking, there is no obvious model misspecification.  We first compare $\mathrm{M}_{3}$ and $\mathrm{M}_{3t}$ to see whether a Student's $t$ model is more appropriate.
\begin{center}
{\small
 \begin{tabular}{lccc}
   Model & $\elpdHat{{\mathrm{M}_{3t},\mathrm{M}_k}}{\y}$ & $\seHatPlain$ & $\hat{p}(\elpd{{\mathrm{M}_{3t},\mathrm{M}_k}}{\y}>0)$ \\
   \midrule
      $\mathrm{M}_{3t}$ &  - &  - &  - \\ 
   $\mathrm{M}_3$ & -41.7 &  13.4 & 0.999 \\
 \end{tabular}
}
\vspace{-1pt}
\end{center}
Although in this comparison $\mathrm{M}_3$ is misspecified, the better specified model $\mathrm{M}_{3t}$ shows much better predictive performance, and as we can expect $\seHatPlain$ to be inflated, the actual probability that $\mathrm{M}_{3t}$ is better than $\mathrm{M}_{3}$ is likely to be bigger than $0.999$. We then compare the three Student's $t$ models:

\begin{center}
{\small
 \begin{tabular}{lccc}
   Model & $\elpdHat{{\mathrm{M}_{3t},\mathrm{M}_k}}{\y}$ & $\seHatPlain$ & $\hat{p}(\elpd{{\mathrm{M}_{3t},\mathrm{M}_k}}{\y}>0)$ \\
  \midrule
   $\mathrm{M}_{3t}$ &  - &  - &  - \\ 
   $\mathrm{M}_{2t}$ & -45.4 &  8.5 & 1.0 \\
   $\mathrm{M}_{1t}$ & -119.1 & 15.9 & 1.0
 \end{tabular}\\
}
\vspace{-1pt}
\end{center}
The probability that model $\mathrm{M}_{3t}$ is better than models $\mathrm{M}_{1t}$ and $\mathrm{M}_{2t}$ is close to 1. The models appear sufficiently well specified, the number of observations is bigger than 100, and the differences are not small, so we can assume that the normal approximation is well calibrated.
In this case, the effect of days with sleep constrained to 3 hours is so big that the main conclusion stays the same with all the models. Still, for example, model $\mathrm{M}_{3t}$ does indicate higher variation between subjects than model $\mathrm{M}_{3}$. As $\mathrm{M}_{3t}$ passes the model checking and has higher predictive performance, we should continue looking at the posterior of model $\mathrm{M}_{3t}$. See Appendix~\ref{subsec_additional_experiments} for additional simulation results for a hierarchical normal model.

\para{Roaches}
\citet[][Chapter 8.3]{Gelman-Hill:2007} describe the roaches data as follows:
\textit{``the treatment and control were applied to 160 and 104 apartments, respectively, and the outcome measurement $y_i$ in each apartment $i$ was the number of roaches caught in a set of traps. Different apartments had traps for different numbers of days''}.
The goal is to estimate the efficacy of a pest management system at reducing the number of roaches.

The target is the number of roaches (y), and the covariates include the square root of the pre-treatment number of roaches (sqrt\_roach1), a treatment indicator variable (treatment), and a variable indicating whether the apartment is in a building restricted to elderly residents (senior). As the number of days for which the roach traps were used is not the same for all apartments, the offset argument includes the logarithm of the number of days the traps were used (log(exposure2)). The latent regression model presented with \texttt{brms} formula notation is: $$\mathrm{y} \sim \mathrm{sqrt\_roach1} + \mathrm{treatment} + \mathrm{senior} + \mathrm{offset(log(exposure2))}.$$
We fit the following models using the \texttt{brms} package. 
{\small
\begin{align*}
  \mathrm{M}_1: \quad & \text{Poisson} \\
  \mathrm{M}_2: \quad & \text{Negative-binomial} \\
  \mathrm{M}_3: \quad & \text{Zero-inflated negative-binomial}
\end{align*}
}
The zero-inflation is modelled using the same latent formula (with its own parameters). All coefficients have $\normal(0,1)$ priors and the negative-binomial shape parameter has the \texttt{brms} default prior, which is inverse-gamma$(.4, .3)$ \citep{Vehtari:2024}.
For the Poisson model we re-ran MCMC for all LOO-folds with high Pareto-$\hat{k}$ diagnostic value (>0.7) (with \texttt{reloo=TRUE} in \texttt{brms}), and for negative-binomial and zero-inflated negative-binomial we used moment matching \citep{Paananen+etal:2021:implicit} for a few LOO-folds with high Pareto-$\hat{k}$ diagnostic value (>0.7) (with \texttt{moment\_match=TRUE} in \texttt{brms}).
\begin{center}
{\small
 \begin{tabular}{lccc}
   Model & $\elpdHat{{\mathrm{M}_{3},\mathrm{M}_k}}{\y}$ & $\seHatPlain$ & $\hat{p}(\elpd{{\mathrm{M}_{3},\mathrm{M}_k}}{\y}>0)$ \\
 \midrule
   $\mathrm{M}_3$ %
   &  - &  - &  - \\ 
   $\mathrm{M}_2$ %
   & -23.0 &  6.9 & 0.9996 \\
   $\mathrm{M}_1$ %
   & -4633.2 & 684.9 & 1.0
 \end{tabular}\\
}
\vspace{-1pt}
\end{center}
The zero-inflated negative-binomial model ($\mathrm{M}_3$) is clearly the best. Based on model checking, the Poisson model ($\mathrm{M}_1$) is underdispersed which indicates Scenario 2, but the difference is so big that we can be certain that the zero-inflated negative-binomial model is better. As the number of observations is larger than 100, and the difference to model $\mathrm{M}_2$ is not small, we may assume the normal approximation is well calibrated.

As we had used an ad-hoc square root transformation of pre-treatment number of roaches, we fitted a model $\mathrm{M}_4$ replacing the latent linear term for the square root of pre-treatment number of roaches with a spline.

\begin{center}
{\small
\begin{tabular}{lccc}
   Model & $\elpdHat{{\mathrm{M}_{4},\mathrm{M}_k}}{\y}$ & $\seHatPlain$ & $\hat{p}(\elpd{{\mathrm{M}_{4},\mathrm{M}_k}}{\y}>0)$ \\
 \midrule
   $\mathrm{M}_4$ %
   &  - &  - &  - \\ 
   $\mathrm{M}_3$ %
   & -2.4 & 3.0 & 0.79
 \end{tabular}\\
}
\vspace{-1pt}
\end{center}
Model $\mathrm{M}_4$ (with spline) seems to be slightly better, but now the difference is so small that the normal approximation is likely to be not perfectly calibrated. As the difference is small, we can proceed with either model. See Appendix~\ref{subsec_additional_experiments} for additional simulation results for Poisson and spline models.

\section{Conclusions}
\label{sec_conclusions}
This paper is the first to thoroughly study the properties of uncertainty quantification in log-score LOO-CV predictive performance difference in the Bayesian setting.
Well-calibrated uncertainty quantification for the predictive performance difference can also be  used for computing the probability that one model has better predictive performance.
We analyse normal and Bayesian bootstrap approximations to quantify the uncertainty and inspect their properties in Bayesian (simple, hierarchical, latent, basis function) linear regression.
We show that problematic settings include models with similar predictions (Scenario~1), bad model misspecification with outliers in the data (Scenario~2), and small data (Scenario~3).

\para{Scenario 1: Models With Similar Predictions}
We show that the problematic skewness of the distribution of the approximation error occurs when models make similar predictions. This skewness does not necessarily disappear as $n$ grows. We show that considering the skewness of the sampling distribution is insufficient to improve the uncertainty estimate, as it has a weak connection to the skewness of the distribution of the estimators' error. We show that, in the problematic settings, both normal and BB approximations to the uncertainty are badly calibrated.

\para{Scenario 1 consequences} Given similar predictions (say $|\elpdHat{\Md}{\y}|<4$; see also \citealp{McLatchie+Vehtari:2024:selection_bias}), we are unlikely to lose in predictive performance, whichever model is selected, and we can use either model for predictions. In the case of nested models, we may prefer the bigger one, as we may get more information
by looking at the posterior of the additional terms.

\para{Scenario 2: Model Misspecification With Outliers}
Cross-validation has been advocated
when the true model is not included in the set of the compared models \citep{Bernardo+Smith:1994,Vehtari+Ojanen:2012}. Our results demonstrate that in the case of bad misspecification, there can be significant bias in the estimated predictive performance difference, and the estimated uncertainty can be miscalibrated.

\para{Scenario 2 consequences} Model checking, and possible refinement, should be considered before using cross-validation for comparing predictive performances. Scenario 2 should not arise when following the Bayesian workflow best practices \citep{Gelman+etal:2020:workflow}.

\para{Scenario 3: Small Data}
Small data (say $n<100$) makes estimating the uncertainty of the predictive performance difference less reliable and additional caution is needed. Obtaining more data is the best way to improve the reliability and reduce uncertainty.

\section{Discussion}
Here, we discuss connections to related methods and useful directions for future research.

\para{Other predictive methods}
\citet{Vehtari+Ojanen:2012} extensively review methods for assessing the predictive performance of Bayesian models. Cross-validation and widely applicable information criterion \citep[WAIC; ][]{Watanabe:2010d} are the only methods targeting the expected predictive performance in the sense of~\eqref{eq_elpd}. LOO-CV and WAIC use different computational approximations but are asymptotically equivalent \citep{Watanabe:2010d}, and thus we expect the uncertainty quantification results to hold for WAIC, too, as long as the computational approximation does not fail \citep[see][]{Vehtari+Gelman+Gabry:2017_practical}.

\para{Other scoring rules}
We may assume that other smooth, strictly proper scoring rules \citep{Gneiting+Raftery:2007} would behave similarly, but further research is justified.

\para{Other models}
We have focused on (simple, hierarchical, latent, basis function) linear models. We assume the results are similar to other models, including singular models, but we leave this for future research.  The potential approach to extend our results is to use singular learning theory by \citet{Watanabe:2009}, who used it to show the asymptotic equivalence of LOO-CV and WAIC and that the asymptotic behaviour of WAIC does not depend on what the functional form of the model is
  \citep{Watanabe:2009,Watanabe:2010d,Watanabe:2010a,Watanabe:2010b,Watanabe:2010c}.

\para{Leave-one-group-out cross-validation}
We used LOO-CV for a hierarchical model, which is a valid option when the focus is on analysing the data model or in predictions for new individuals in the existing groups. 
Alternatively, {leave-one-group-out} cross-validation can be used to simulate predictions for new groups \citep[see, e.g.,][]{Vehtari+Lampinen:2002,Merkle+Furr+Rabe_Hesketh:2019}.
If the joint log score is used to assess the performance of joint predictions for all observations in one group, we get only one log score per group. We assume that the number of groups is the decisive factor for the behaviour.

\para{Leave-future-out cross-validation}
In the case of time series, if the goal is to assess the predictive performance for the future (and not to time points between observations), we can use {leave-future-out} cross-validation \citep[see, e.g.][]{Burkner+Gabry+Vehtari:2020}. In this case, the pointwise log score values are not exchangeable, as the amount of data used to fit the posterior is different for each prediction, and the dependency structure between folds is different. For long timeseries, this is likely to have a minor effect.

\para{Comparison of multiple models}
When comparing a few models with LOO-CV, \citet{vehtari2022loopkg} recommend making pairwise comparisons to the model with the best predictive performance (approach used in \texttt{loo} R package since 2015). This approach reduces the number of comparisons to be one less than the number of models and provides a natural ordering for the comparisons. If the best model is clearly better than others based on the difference and the associated uncertainty, there is no need to examine the differences and uncertainties for the rest.

\para{Model selection}
Given well-calibrated uncertainty quantification of the
predictive performance difference, it is possible to compute
well-calibrated probability that one model has better predictive
performance than another model, as we have shown in this paper. We do
not suggest any fixed probability threshold for making model
selection, as the appropriate threshold depends on the context.
\citet{McLatchie+Vehtari:2024:selection_bias} show 1) what happens if
the model with the highest estimated predictive performance is
selected, 2) how by taking into account the uncertainty it is possible
to estimate the model selection induced bias and amount of overfitting
in the selection process, 3) if the model selection criterion is
changed to allow not to select a model, this bias and overfitting can
be avoided, and 4) in the case of two models the bias and overfitting
are negligible.
\citet{Liu+etal:2025:loose_model_selection} demonstrate the benefits
of using the uncertainty quantification by selecting the simplest
hierarchical model among those that are not significantly worse than
the model with the best predictive performance.
\citet{Riha+etal:2024:multiverse} propose to make a multiverse
analysis with all the models that have similar predictive performance
as the best model based on $\elpdHat{\Md}{\y}$ and
$\seHat{\Md}{\yobs}$, to better understand the effect of differences
in models.

\para{Model averaging} If one of the models is not clearly the best, and the aim is the best prediction (no need for model selection), model averaging can be used. \citet{Yao2018stacking} compare model weights using 1) LOO-CV differences, 2) LOO-CV differences plus related uncertainty handled with BB (normal approximation gets complicated with many models), and 3) LOO-CV based Bayesian stacking. \citet{Yao2018stacking} show that taking into account the LOO-CV uncertainty improves the model averaging with LOO weights and performs better than model selection with plain LOO-CV and marginal likelihood, but Bayesian stacking performs even better.
LOO-CV weights have an issue in that similar models get similar weights and dilute the weights of other models, making the interpretation of the weights more difficult. On the other hand, Bayesian stacking weights are optimised for predictive model averaging, and the interpretation of weights is also non-trivial \citep[see discussion and examples in][]{Yao+etal:2021:hierstacking}.

\para{Model selection and many similar models}
As the normal approximation for the predictive performance difference uncertainty of similar models can be miscalibrated,
\citet{McLatchie+Vehtari:2024:selection_bias} propose in the case of many models to examine the distribution of the performance estimates for all the models and use order statistics for estimating and correcting potential selection induced bias. They discuss the connection between $\elpdHat{\Md}{\y}$, $\seHat{\Md}{\yobs}$ and the LOO weights for model selection. They compare their proposed approach to the use of $\elpdHat{\Md}{\y}$ and $\seHat{\Md}{\yobs}$, which has similar performance and also avoids overfitting in model selection, even though $\seHat{\Md}{\yobs}$ is likely to be underestimated in the case of similar models.

\para{Projection predictive model selection}
Further stability in variable selection can be obtained using the projection predictive method \citep{Piironen+Vehtari:2016,piironen2020projective,McLatchie+etal:2025:projpred_workflow}, as it has lower variance than LOO-CV. One key part of the projection predictive method is the use of $\elpdHat{\Md}{\y}$ and $\seHat{\Md}{\yobs}$ to select the smallest projected model along the search path, which has similar predictive performance as the full reference model. The projective predictive method has been shown to outperform many other model selection methods \citep{Piironen+Vehtari:2016,piironen2020projective,McLatchie+etal:2025:projpred_workflow}.

\para{Additional practical advice and case studies}
Online CV-FAQ (\url{https://users.aalto.fi/~ave/CV-FAQ.html}) contains more practical advice and links to many case studies illustrating the use of predictive performance comparison in all three scenarios.



\subsection*{Acknowledgements}

We thank Noa Kallioinen, Daniel Simpson and anonymous reviewers for helpful comments and feedback.
We acknowledge the computational resources provided by the Aalto Science-IT project. This work was supported by the Academy of Finland grants (298742 and 313122) and Academy of Finland Flagship programme: Finnish Center for Artificial Intelligence FCAI.

\vskip 0.2in
\bibliography{references.bib}

\clearpage

\appendices




%
\section{Difference Between Estimating elpd and e-elpd}\label{app_sec_theory_elpd_vs_eelpd}
As discussed in the beginning of Section~\ref{sec_problem_setting}, depending on if the context of the model predictive performance comparison is in evaluating the models for the given data set or for the data generating mechanism in general, the measure of interest is either
\begin{align}
    \label{app_eq_elpd_2}
    \elpd{\Mk}{\yobs} &= \sum_{i=1}^n \int \p_\text{true}(\y_i) \log \p_\Mk(\y_i \mid \yobs) \diff \y_i\,,
\intertext{or its expectation over possible data sets}
    \label{app_eq_eelpd_2}
    \eelpd{\Mk} &= \E_{\y}\left[ \elpdC{\Mk}{\y} \right]
\end{align}
respectively. The uncertainty related to the $\elpdHatPlain$ estimator is different depending on if it is used to estimate $\elpd{\Md}{\yobs}$ or $\eelpd{\Md}$.
While otherwise focusing on analysing the nature of the uncertainty in the application-oriented context of $\elpd{\Md}{\yobs}$ measure, in this appendix, we formulate the uncertainties related to both measures and discuss their differences in more detail. The following analysis of the uncertainty generalises also for estimating $\eelpd{\Mk}$ or $\elpd{\Mk}{\yobs}$ for one model and other $K$-fold CV estimators.
\subsection{Estimating e-elpd}
When using $\elpdHat{\Md}{\yobs}$ to estimate $\eelpd{\Md}$, $\elpdHatC{\Md}{\y}$ is an estimator considering $(\yobs)$ as a random sample of the stochastic variable $\y$.
Any observed data set $\yobs$ can be used to estimate the same quantity $\eelpd{\Md}$. The uncertainty about the $\eelpd{\Md}$ given an estimate can be assessed by considering the error over possible data sets,
\begin{align}
 \elpdHatErrEEC{\Md}{\y}{\eelpdPlain} = \elpdHatC{\Md}{\y} - \eelpd{\Md}\,,
\end{align}
which corresponds to the estimator's sampling distribution $\elpdHatC{\Md}{\y}$ shifted by a constant.
\subsection{Estimating elpd}
When using $\elpdHat{\Md}{\yobs}$ to approximate $\elpd{\Md}{\yobs}$, however, $\yobs$ is given also in the approximated quantity. Each observed data set $\yobs$ can be used to approximate different quantities $\elpd{\Md}{\yobs}$. Here, the error is formulated as
\begin{align}\label{app_eq_loo_err_2}
 \elpdHatErrEEC{\Md}{\y}{\elpdPlain} = \elpdHatC{\Md}{\y} - \elpdC{\Md}{\y} \,.
\end{align}
Even though reflecting a different problem for each realisation of the data set, the associated uncertainty about one problem can be assessed by analysing 
the approximation error over possible data sets in a similar fashion as when estimating $\eelpd{\Md}$. However, here the variability of $\elpdHatErrC{\Md}{\y}$ depends both on $\elpdHatC{\Md}{\y}$ and $\elpdC{\Md}{\y}$.
\subsection{Error Distributions}
Assuming the observations $\y_i$, $i=1,2,\dots,n$ are independent, the expectation of the error distributions for both measures $\elpdPlain$ and $\eelpdPlain$ are the same, that is
\begin{align}
 \E\left[\elpdHatErrEEC{\Md}{\y}{\elpdPlain}\right] &= \E\left[\elpdHatErrEEC{\Md}{\y}{\eelpdPlain}\right]\,,
\end{align}
but they differ in variability. In particular, as demonstrated for example in Figure~\ref{fig_demonstration}, the correlation of $\elpdHatC{\Md}{\y}$ and $\elpdC{\Md}{\y}$ is generally small or negative and thus the variance,
\begin{align}
 \Var\left(\elpdHatErrEEC{\Md}{\y}{\elpdPlain}\right) &= \Var\left(\elpdHatC{\Md}{\y}\right) + \Var\left(\elpdC{\Md}{\y}\right)
 \nonumber \\
 &\quad- 2 \Cov\left(\elpdHatC{\Md}{\y} , \elpdC{\Md}{\y} \right)\,,
\end{align} 
is usually greater than
\begin{align} 
 \Var\left(\elpdHatErrEEC{\Md}{\y}{\eelpdPlain}\right)&= \Var\left(\elpdHatC{\Md}{\y}\right)\,.
\end{align}
Because of the differences in the error distributions, it is significant to consider the uncertainties separately for both measures $\elpdPlain$ and $\eelpdPlain$.
\subsection{Sampling Distributions}
When estimating $\eelpd{\Md}$, $\elpdHatC{\Md}{\y}$ is a random variable corresponding to the estimator's sampling distribution for the specific problem. However, when approximating $\elpd{\Md}{\yobs}$, $\elpdHatC{\Md}{\y}$ and $\elpdHatErrEE{\Md}{\y}{\elpdPlain}$ are stochastic variables reflecting the frequency properties of the approximation when applied for different problems. Nevertheless, we refer to $\elpdHatC{\Md}{\y}$ as an estimator and $\elpdHatC{\Md}{\y}$ as a sampling distribution also in the latter context. Note, however, that other assessments of the uncertainty of the estimator \elpdHatPlain{} for $\elpd{\Md}{\yobs}$ can be made. The related formulation of the target uncertainty about $\elpd{\Md}{\yobs}$ is discussed in more detail in Appendix~\ref{app_sec_theory_alternative}.

%
\section{Alternative Formulations of the Uncertainty}\label{app_sec_theory_alternative}
%
In Appendix~\ref{app_sec_theory_elpd_vs_eelpd}, we analyse and motivate the method applied in the paper and mention that other approaches can be made for assessing the uncertainty about $\elpd{\Md}{\yobs}$. This appendix discusses some of these and further motivates the applied method. Instead of analysing the error stochastically over possible data sets, it is also possible, for example, to find bounds or apply Bayesian inference to the error. 
As briefly discussed in Section~\ref{sec_intro_uncertainty}, also other formulations of the target uncertainty
\begin{align}
    \elpdHatUnkC{\Md}{\yobs} = \elpdHat{\Md}{\yobs} - \elpdHatErrC{\Md}{\y},
\end{align}
may satisfy the desired equality
\begin{gather}
    q \Bigr(\elpdHatUnkC{\Md}{\y}\Bigl) = \p\Bigr(\elpdC{\Md}{\y}\Bigl)
    \label{app_eq_calibration_approx_2}\,.
\end{gather}
For example, while not sensible as a target for the estimated uncertainty, assigning the Dirac delta function located at $\elpd{\Md}{\yobs}$ as a probability distribution for $\elpdHatUnkC{\Md}{\yobs}$ trivially satisfies Equation~\eqref{app_eq_calibration_approx_2}.
However, other approaches might also provide a feasible uncertainty estimator target.
In particular, these alternative formulations could be developed for specific problem settings.
\subsection{LOO-CV Estimate with Independent Test Data}
One possible general interpretation of the uncertainty could arise by considering $\elpdHatPlain$ as one possible realised estimation from the following estimator. Let
\begin{align} \label{app_eq_elpdhat_2yobs}
 \elpdHat{\Mk}{\tilde{y}^\mathrm{obs}, \yobs} &= \sum_{i=1}^n \log \p_\Msk\left(\tilde{y}^\mathrm{obs}_i \mid \yobs_{-i}\right)\,.
\end{align}
In this estimator, the data set $\tilde{y}^\mathrm{obs}$ is considered a random sample for estimating $\p_\text{true}(\y)$ and $\yobs$ is a given data set indicating the problem at hand in the $\elpd{\Mk}{\yobs}$, i.e.\ the training and test data sets are separated. Now $\elpdHat{\Mk}{\yobs} = \elpdHat{\Mk}{\yobs, \yobs}$ is one application of this estimator, where the same data set is re-used for both arguments. The uncertainty of the estimator $\elpdHat{\Md}{\tilde{y}^\mathrm{obs}, \yobs}$ can be formulated in the following way:
\begin{align}
 \elpdHatUnkC{\Md}{\tilde{y}^\mathrm{obs}, \yobs} &= \elpdHat{\Md}{\tilde{y}^\mathrm{obs}, \yobs} - \elpdHatErrEEC{\Md}{\y, \yobs}{\prime} \,,
 \intertext{where}
 \elpdHatErrEEC{\Md}{\y, \yobs}{\prime} &= \elpdHatC{\Md}{\y, \yobs} - \elpd{\Md}{\yobs}\,.
\end{align}
Similar to estimating $\eelpdPlain$, here the variability of the error is not affected by $\elpdC{\Md}{\y}$, unlike in the formulation
\begin{align}
 \elpdHatErrEEC{\Md}{\y}{\elpdPlain} = \elpdHatC{\Md}{\y} - \elpdC{\Md}{\y}\,.
\end{align}
\begin{figure}[tb!]
  \centering
  \includegraphics[width=0.9\figurecontrolwidth]{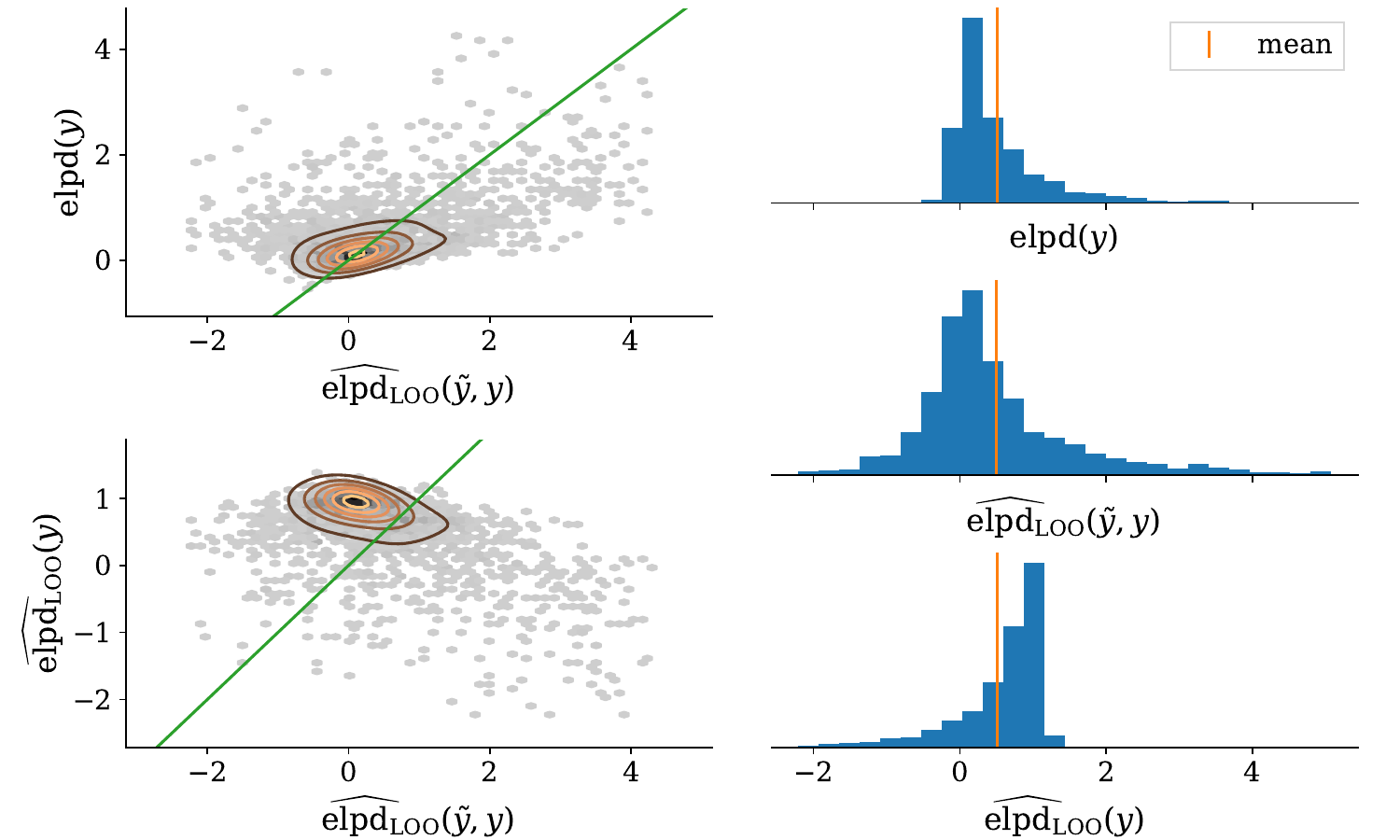}
  \caption{
   Comparison of $\elpdC{\Md}{\y}$ and the sampling distributions of $\elpdHatC{\Md}{\y}$ and $\elpdHatC{\Md}{\tilde{\y}, \y}$ for a selected problem setting, where $n=128$, $\beta_\Delta=0$, $r_\star=0$. In the joint distribution plots on the left column, kernel density estimation is shown with orange lines, and the green diagonal lines correspond to $y=x$. It can be seen from the figure that the sampling distributions of $\elpdHatC{\Md}{\y}$ and $\elpdHatC{\Md}{\tilde{\y}, \y}$ have different shapes. For brevity, model labels are omitted in the notation in the figure. 
  }\label{app_fig_loo_indep_pred}
\end{figure}
Even though being connected, using $\elpdHat{\Md}{\yobs, \yobs}$ as a proxy for the uncertainty in analysing the behaviour of the LOO-CV estimate would produce inaccurate results. As experimentally demonstrated in Figure~\ref{app_fig_loo_indep_pred}, the data sets' connection affects the estimator's related uncertainty. The behaviour of $\elpdHat{\Md}{\tilde{y}^\mathrm{obs}, \yobs}$ over possible data sets does not necessarily match with the behaviour of $\elpdHat{\Md}{\yobs}$. It can be seen from the figure that in the illustrated setting, the means of the distributions are close, but the variance and skewness do not match. Additionally, the figure compares the sampling distributions against the distribution of $\elpdC{\Md}{\y}$. It can be seen that $\elpdHatC{\Md}{\tilde{\y}, \y}$ has a distribution somewhat between $\elpdHatC{\Md}{\y}$ and $\elpdC{\Md}{\y}$. Indeed, although not feasible in practice, it is expected that $\elpdHatC{\Md}{\tilde{\y}, \y}$ would be a better estimator for $\elpdC{\Md}{\y}$.
%
%
\section{Analysing the Uncertainty Estimates}\label{app_sec_theory_analysis_of_the_estimates}
The uncertainty of a LOO-CV estimate is usually estimated using normal distribution or Bayesian bootstrap. In this appendix, we discuss these estimators in more detail.
\subsection{Normal Model for the Uncertainty}
As discussed in Section~\ref{sec_intro_uncertainty} in Equation~\eqref{eq_normalapprox_d}, a common approach for estimating the uncertainty in a LOO-CV estimate is to approximate it  with a normal distribution as
\begin{align}
    \elpdHatUnkHatC{\Md}{\yobs} &= \elpdHat{\Md}{\yobs} - \elpdHatErrHatC{\Md}{\yobs}
    \label{app_eq_elpdhat_err}\,,
\end{align}
where
\begin{align} 
    \elpdHatErrHatC{\Md}{\yobs} &\sim \N\left(0, \, \seHat{\Md}{\yobs} \right)
\end{align}
is an approximation to the distribution of the true error over the possible data sets, and $\seHat{\Md}{\yobs}$ is a naive estimator of the standard error of $\elpdHatC{\Md}{\y}$ defined by \citet{vehtari2022loopkg} as 
\begin{align*}
 \left(\seHat{\Md}{\yobs}\right)^2
    = \frac{n}{n-1}\sum_{i=1}^n \Bigg(\elpdHati{\Md}{\yobs}{i} - \frac{1}{n}\sum_{j=1}^n \elpdHati{\Md}{\yobs}{j} \Bigg)^2 \,.
    \numberthis \label{app_eq_naive_estimator}
\end{align*}
This estimator is motivated by the incorrect assumption that the terms $\elpdHatiC{\Md}{\y}{i}$ are independent. In reality, since each observation is a part of $n-1$ training sets, the variance $\Var\left(\elpdHatC{\Md}{\y}\right)$ depends on both the variance of each $\elpdHatiC{\Md}{\y}{i}$ and on the dependency between the different folds. 

In the following propositions~\ref{app_prop_loo_var} and~\ref{app_th_naive_var_expectation} and in the Corollary~\ref{app_th_naive_var_bias}, we present the associated bias with the naive variance estimator in the context of model comparison.
\begin{proposition}\label{app_prop_loo_var}
 Let $\Short_{\Msk,i} = \elpdHatiC{\Mk}{\y}{i}$ and $\Short_{\Msd,i} = \elpdHatiC{\Md}{\y}{i}$ and
 \begin{align*}
  \Var\left(\Short_{\Msd,i}\right) &= \sigma^2_\Msd
  &
  \Cov\left(\Short_{\Msd,i}, \Short_{\Msd,j} \right) &= \gamma_{\Msd}
  \\
  \Var\left(\Short_{\Msk,i}\right) &= \sigma^2_\Msk
  &
  \Cov\left(\Short_{\Msk,i}, \Short_{\Msk,j} \right) &= \gamma_{\Msk}
  \\
  \Cov\left(\Short_{\Msa,i}, \Short_{\Msb,i} \right) &= \rho_{\Msa\Msb} 
  &
  \Cov\left(\Short_{\Msa,i}, \Short_{\Msb,j} \right) &= \gamma_{\Msa\Msb}\,,
  \numberthis
 \end{align*}
 where $i \neq j$ and $\Mk \in \{\Ma,\Mb\}$. Now
 \begin{align*}
  \Var\left(\elpdHatC{\Md}{\y}\right)
  &= n \sigma^2_\Msd + n(n-1)\gamma_\Msd
  \\
  &= n \left(\sigma^2_\Msa + \sigma^2_\Msb - 2\rho_{\Msa\Msb} \right) + n(n-1)\left(\gamma_\Msa + \gamma_\Msb - 2 \gamma_{\Msa\Msb}\right)\,.
  \numberthis
 \end{align*}
\end{proposition}
\begin{proof}
 We have
 \begin{align*}
 &\Var\left(\elpdHatC{\Md}{\y}\right)
 \\
 &\qquad= \sum_{i=1}^n\sum_{j=1}^n \Cov\left( \Short_{\Msa,i} - \Short_{\Msb,i}, \Short_{\Msa,j} - \Short_{\Msb,j} \right)
 \\
 &\qquad= \sum_{i=1}^n\sum_{j=1}^n \Big(
    \Cov(\Short_{\Msa,i}, \Short_{\Msa,j}) + \Cov(\Short_{\Msb,i}, \Short_{\Msb,j})
    \\
    & \qquad\qquad\qquad -\Cov(\Short_{\Msa,i}, \Short_{\Msb,j}) - \Cov(\Short_{\Msb,i}, \Short_{\Msa,j})
 \Big)
 \\
 &\qquad= \sum_{i=1}^n \Big(
    \Var(\Short_{\Msa,i}) + \Short_{\Msb,i} - 2 \Cov(\Short_{\Msa,i},\Short_{\Msb,i})
 \Big)
 \\ & \qquad\qquad \qquad
 + \sum_{i=1}^n\sum_{j\neq i} \Big(
    \Cov(\Short_{\Msa,i}, \Short_{\Msa,j}) + \Cov(\Short_{\Msb,i}, \Short_{\Msb,j})
    -2\Cov(\Short_{\Msa,i}, \Short_{\Msb,j})
 \Big)
 \\
 &\qquad= n \left(\sigma^2_\Msa + \sigma^2_\Msb - 2\rho_{\Msa\Msb} \right)
    + n(n-1)\left(\gamma_\Msa + \gamma_\Msb - 2 \gamma_{\Msa\Msb}\right)\,.
 \numberthis
 \end{align*}
\end{proof}
\begin{proposition}\label{app_th_naive_var_expectation}
 Following the definitions in Proposition~\ref{app_prop_loo_var},
 the expectation of the variance estimator \seHatPlain{} in Equation~\eqref{app_eq_naive_estimator} is
 \begin{align*}
    \E\left[\seHatC{\Md}{\y}^2\right]
    &= n \sigma^2_\Msd - n \gamma_\Msd
    \\
    &= n \left(\sigma^2_\Msa + \sigma^2_\Msb - 2\rho_{\Msa\Msb} \right) - n\left(\gamma_\Msa + \gamma_\Msb - 2 \gamma_{\Msa\Msb}\right).
    \numberthis
 \end{align*}
\end{proposition}
\begin{proof}
 We have
 \begingroup
 \allowdisplaybreaks
 \begin{align}
    \E\left[\Short_{\Msd,i}^2\right] &= \E\left[\Short_{\Msd,i}\right]^2 + \Var\left(\Short_{\Msd,i}\right) \\
    \E\left[\Short_{\Msd,i} \Short_{\Msd,j}\right] &= \E\left[\Short_{\Msd,i}\right]\E\left[\Short_{\Msd,j}\right] + \Cov\left(\Short_{\Msd,i}, \Short_{\Msd,j}\right), \quad i \neq j.
 \end{align}
 Now
 \begin{align*}
    &\E\left[ \left(\seHat{\Md}{\yobs}\right)^2 \right]
    \\
    &\qquad= \E\left[ \frac{n}{n-1}\sum_{i=1}^n \left(\Short_{\Msd,i} - \frac{1}{n}\sum_{j=1}^n \Short_{\Msd,j} \right)^2 \right]
    \\
    &\qquad= \frac{n}{n-1}\sum_{i=1}^n \E\left[ \Short_{\Msd,i}^2 - \frac{2}{n} \Short_{\Msd,i} \sum_{j=1}^n \Short_{\Msd,j} + \left( \frac{1}{n}\sum_{j=1}^n \Short_{\Msd,j} \right)^2 \right]
    \\
    &\qquad= \frac{n}{n-1}\sum_{i=1}^n \left[
        \E\left[\Short_{\Msd,i}^2\right]
        - \frac{2}{n} \left( \E\left[\Short_{\Msd,i}^2\right] + \sum_{j \neq i } \E\left[\Short_{\Msd,i} \Short_{\Msd,j} \right] \right) \right.\\
        &\qquad \left. \hspace{25mm} + \frac{1}{n^2} \left( \sum_{j=1}^n \E\left[\Short_{\Msd,j}^2\right] + \sum_{j=1}^n \sum_{p \neq j } \E\left[\Short_{\Msd,j} \Short_{\Msd,p}\right] \right)
    \right]
    \\
    &\qquad=\frac{n}{n-1}\sum_{i=1}^n \Bigg[
    \E\left[\Short_{\Msd,i}\right]^2 +\Var\left(\Short_{\Msd,i}\right)
    \\ &\qquad  \hspace{25mm}
    - \frac{2}{n} \Big( \E\left[\Short_{\Msd,i}\right]^2 +\Var\left(\Short_{\Msd,i}\right) + (n-1)(\E\left[\Short_{\Msd,i}\right]^2
    \\ &\qquad  \hspace{40mm}
    + \Cov\left(\Short_{\Msd,i}, \Short_{\Msd,j}\right))  \Big)
    \\ &\qquad  \hspace{25mm}
    + \frac{1}{n^2} \Big( n (\E\left[\Short_{\Msd,i}\right]^2 +\Var\left(\Short_{\Msd,i}\right)) + n(n-1)(\E\left[\Short_{\Msd,i}\right]^2
    \\ &\qquad  \hspace{40mm}
    + \Cov\left(\Short_{\Msd,i}, \Short_{\Msd,j}\right)) \Big)
    \Bigg]\\
    &\qquad=\frac{n}{n-1}\sum_{i=1}^n \left[
    \left(1 - \frac{2n}{n} + \frac{n^2}{n^2}\right)\E\left[\Short_{\Msd,i}\right]^2
    + \left(1 - \frac{2}{n} + \frac{n}{n^2}\right)\Var\left(\Short_{\Msd,i}\right)  \right.
    \\
    &\qquad \left. \hspace{25mm} + \left(- \frac{2(n-1)}{n} + \frac{n(n-1)}{n^2}\right)\Cov\left(\Short_{\Msd,i}, \Short_{\Msd,j}\right)
    \right]
    \\
    &\qquad=\frac{n}{n-1}\sum_{i=1}^n \left[ \frac{n-1}{n}\Var\left(\Short_{\Msd,i}\right) - \frac{n-1}{n}\Cov\left(\Short_{\Msd,i}, \Short_{\Msd,j}\right) \right]
    \\
    &\qquad= n\Var\left(\Short_{\Msd,i}\right) - n\Cov\left(\Short_{\Msd,i}, \Short_{\Msd,j}\right)
    \\
    &\qquad= n \sigma^2_\Msd - n \gamma_\Msd\,,
 \numberthis
 \end{align*}
 and furthermore
 \begin{align*}
    \E\left[ \left(\seHat{\Md}{\yobs}\right)^2 \right]
    &= n\Var\left(\Short_{\Msa,i}-\Short_{\Msb,i}\right)
    \\
    &\quad- n\Cov\left(\Short_{\Msa,i}-\Short_{\Msb,i}, \Short_{\Msa,j}-\Short_{\Msb,j}\right)
    \\
    &= n \left(\Var(\Short_{\Msa,i}) + \Var(\Short_{\Msb,i}) - 2 \Cov(\Short_{\Msa,i}, \Short_{\Msb,i}) \right)
    \\
    &\quad - n\left(\Cov\left(\Short_{\Msa,i}, \Short_{\Msa,j} \right) + \Cov\left(\Short_{\Msb,i}, \Short_{\Msb,j} \right) - 2 \Cov\left(\Short_{\Msa,i}, \Short_{\Msb,j} \right)\right)
    \\
    &= n \left(\sigma^2_\Msa + \sigma^2_\Msb - 2\rho_{\Msa\Msb} \right) \\
 &\quad- n\left(\gamma_\Msa + \gamma_\Msb - 2 \gamma_{\Msa\Msb}\right)
 \,.
 \numberthis
 \end{align*}
 \endgroup 
\end{proof}
\begin{corollary}\label{app_th_naive_var_bias}
 Following the definitions in Proposition~\ref{app_prop_loo_var}, the estimator $\seHatC{\Md}{\y}^2$ defined in Equation~\eqref{app_eq_naive_estimator} for the variance $\Var\left(\elpdHatC{\Md}{\y}\right)$ has a bias of
 \begin{align}
   \E\left[\seHatC{\Md}{\y}^2\right] - \Var\left(\elpdHatC{\Md}{\y}\right) = -n^2 \gamma_\Msd = -n^2\left(\gamma_\Msa + \gamma_\Msb - 2 \gamma_{\Msa\Msb}\right)\,.
 \end{align}
\end{corollary}
\begin{proof}
 The $\elpdC{\Msd}{\y}$, i.e.\ the true variance $\Var\left(\elpdHatC{\Md}{\y}\right)$, is given in Proposition~\ref{app_prop_loo_var}. The expectation of the estimator $\seHatC{\Md}{\y}^2$ is given in Proposition~\ref{app_th_naive_var_expectation}. The resulting bias follows directly from these propositions.
\end{proof}
\subsection{Dirichlet Model for the Uncertainty}
%
As discussed in Section~\ref{sec_intro_uncertainty}, an alternative way to address the uncertainty is to use a Bayesian bootstrap procedure~\citep{Rubin:1981a,Vehtari+Lampinen:2002} to model $\p(\elpdHatUnkC{\Md}{\yobs})$. 
Compared to the normal approximation, while representing skewness, this method also has problems with higher moments and heavy-tailed distributions~\citep{Rubin:1981a}.
\subsection{Not Considering All the Terms in the Error}
As discussed in Section~\ref{sec_intro_problems_in_uncertainty}, in addition to possibly inaccurately approximating the variability in $\elpdHatC{\Md}{\y}$, the presented ways of estimating the uncertainty can be poor representations of the uncertainty about $\elpdHatUnkC{\Md}{\yobs}$ because they are based on estimating the sampling distribution, which can have only a weak connection to the error distribution.
As seen from the formulation of the error $\elpdHatErrC{\Md}{\y}$ presented in Equation~\eqref{app_eq_elpdhat_err}, an estimator based on the sampling distribution does not consider the effect of the term $\elpdC{\Md}{\y}$.
As demonstrated in figures~\ref{fig_err_vs_sampdist} and~\ref{fig_err_vs_sampdist_out} in Appendix~\ref{app_sec_additional_experiment_results}, while in well-behaved problem settings the variability of the sampling distribution $\elpdHatC{\Md}{\y}$ can match with the variability of the error $\elpdHatErrC{\Md}{\y}$, in problematic situations they do not match. As a comparison, when estimating $\eelpdPlain{}$ instead of $\elpdPlain{}$, the variance of the sampling distribution corresponds to the variance of the error distribution, as discussed in Appendix~\ref{app_sec_theory_elpd_vs_eelpd}, and estimating the sampling distribution is sufficient in estimating the uncertainty of the LOO-CV estimate.
%
%
\section{Normal Linear Regression Case Study}
\label{app_sec_norm_lin_reg_case_study}
%
In this appendix, we derive the analytic form for the approximation error 
\begin{align*}    
\elpdHatErrC{\Md}{\y} = \elpdHatC{\Md}{\y} - \elpdC{\Ma,\Mb}{\y}
\end{align*}
in a normal linear regression model comparison setting under known data generating mechanism. In addition, we derive the analytic forms for $\elpdC{\cdots}{\y}$ and $\elpdHatC{\cdots}{\y}$ for the individual models and the difference.

Consider the following data generation mechanism defined in Section~\ref{sec_analytic_case},
we compare two nested normal linear regression models \Mfa{} and \Mfb{}, both considering a subset of covariates. Let $X_{[\bcdot,\Msk]}$ and $\+\beta_{\Msk}$ denote the explanatory variable matrix and respective effect vector including only the covariates considered by model $\Mk \in \{\Mfa, \Mfb\}$. Correspondingly, let $X_{[\bcdot,-\Msk]}$ and $\+\beta_{-\Msk}$ denote the explanatory variable matrix and respective effect vector, including only the covariates not considered by model $\Mk$. If a model includes all the covariates, we define that $X_{[\bcdot,-\Msk]}$ is a column vector of length $n$ of zeroes and $\+\beta_{-\Msk}=0$. We assume that at least one covariate is included in one model but not in the other, so there is some difference in the models. Otherwise, $\elpdC{\Mfd}{\y}$ and $\elpdHatC{\Mfd}{\y}$ would be trivially always zero. The noise variance $\tau^2$ is fixed in both models, and $\widehat{\+\beta}_\Msk$ is the sole estimated unknown model parameter.
We apply uniform prior distribution for both models. Hence, we have the following forms for the likelihood, posterior distribution, and posterior predictive distribution for model $\Mk$
\citep[see e.g.][pp.\ 355--357]{bda_book}:
\begin{align}
    \+y | \widehat{\+\beta}_\Msk, \+X_{[\bcdot,\Msk]}, \tau
    & \sim
    \operatorname{N}\left(X_{[\bcdot,\Msk]} \widehat{\+\beta}_\Msk, \; \tau^2\eye\right),
    \label{app_eq_lin_reg_likelihood}
\\
    \widehat{\+\beta}_\Msk | \+y, \+X_{[\bcdot,\Msk]}, \tau
    & \sim
    \operatorname{N}\left((X_{[\bcdot,\Msk]}^\transp X_{[\bcdot,\Msk]})^{-1}X_{[\bcdot,\Msk]}^\transp \+y, \; (X_{[\bcdot,\Msk]}^\transp X_{[\bcdot,\Msk]})^{-1} \tau^2 \right),
\\
    \widetilde{y} | \+y, \+X_{[\bcdot,\Msk]}, \widetilde{\+x}, \tau
    & \sim
    \operatorname{N}\left(\widetilde{\+x} (X_{[\bcdot,\Msk]}^\transp X_{[\bcdot,\Msk]})^{-1}X_{[\bcdot,\Msk]}^\transp \+y, \; \left(1 + \widetilde{\+x} (X_{[\bcdot,\Msk]}^\transp X_{[\bcdot,\Msk]})^{-1} \widetilde{\+x}^\transp \right) \tau^2 \right),
    \label{app_eq_lin_reg_full_pred_dist}
\end{align}
where $\widetilde{y}, \widetilde{\+x}$ is a test observation with a scalar response variable and conformable explanatory variable row vector, respectively.

\subsection{Elpd} \label{app_sec_analytic_elpd}
In this section we find the analytic form for $\elpdC{\Mk}{\y}$ for model $\Mk \in \{\Mfa, \Mfb\}$.
We have
\begin{align}
    \elpdC{\Mk}{\y} &= \sum_{i=1}^n \int_{-\infty}^\infty \p_\text{true}(\tilde{y}_i) \log \p_\Msk(\tilde{y}_i|\+y) \diff \tilde{y}_i\,,
    \\
    \p_\text{true}(\tilde{y}_i) &= \N\left(\tilde{y}_i \middle| \widetilde{\mu}_{i}, \widetilde{\sigma}_{i}\right),
    \\
    \p_\Msk(\tilde{y}_i|\+y) &= \N\left(\tilde{y}_i \middle| \mu_{\Msk,\,i}, \sigma_{\Msk,\,i}\right),
\end{align}
where
\begin{align}
    \widetilde{\mu}_{i} &= \mu_{\star,\,i} + \+X_{[i,\bcdot]} \+\beta
    \\
    \widetilde{\sigma}_{i}^2 &= \sigma_{\star,\,i}^2
\end{align}
and, according to Equation~\eqref{app_eq_lin_reg_full_pred_dist},
\begin{align}
    \mu_{\Msk,\,i} &= \+X_{[i,\Msk]} \left(\+X_{[\bcdot,\Msk]}^\transp \+X_{[\bcdot,\Msk]}\right)^{-1}\+X_{[\bcdot,\Msk]}^\transp \+y
    \\
    \sigma_{\Msk,\,i}^2 &= \left(1 + \+X_{[i,\Msk]} \left(\+X_{[\bcdot,\Msk]}^\transp \+X_{[\bcdot,\Msk]}\right)^{-1} \+X_{[i,\Msk]}^\transp \right) \tau^2
\end{align}
for $i=1,2,\dots,n$.
The distributions can be formulated as
\begin{align*}
\p_\text{true}(\tilde{y}_i)
&= (2 \pi \widetilde{\sigma}_{i}^2)^{-1/2} \exp\left( -\frac{1}{2}\left( \frac{\tilde{y}_i - \widetilde{\mu}_{i}}{\widetilde{\sigma}_{i}} \right)^2 \right)
\\
    &= c \, \exp\left( -a \tilde{y}_i^2 + b \tilde{y}_i \right)\,,
    \numberthis
\end{align*}
where
\begin{align}
a &= \frac{1}{2 \widetilde{\sigma}_{i}^2} > 0\,,
&
b &= \frac{\widetilde{\mu}_{i}}{\widetilde{\sigma}_{i}^2}\,,
&
c &= \exp\left(-\frac{\widetilde{\mu}_{i}^2}{2\widetilde{\sigma}_{i}^2}-\frac{1}{2}\log\left(2 \pi \widetilde{\sigma}_{i}^2\right)\right)\,,
\end{align}
and
\begin{align*}
\log \p_\Msk(\tilde{y}_i|\+y)
&= -\frac{1}{2}\left( \frac{\tilde{y}_i - \mu_{\Msk,\,i}}{\sigma_{\Msk,\,i}} \right)^2 - \frac{1}{2}\log\left(2 \pi \sigma_{\Msk,\,i}^2 \right)
\\
    &= - p \tilde{y}_i^2 + q \tilde{y}_i + r\,,
    \numberthis
\end{align*}
where
\begin{align}
p &= \frac{1}{2 \sigma_{\Msk,\,i}^2} > 0\,,
&
q &= \frac{\mu_{\Msk,\,i}}{\sigma_{\Msk,\,i}^2}\,,
&
r &= -\frac{\mu_{\Msk,\,i}^2}{2 \sigma_{\Msk,\,i}^2} - \frac{1}{2}\log\left(2 \pi \sigma_{\Msk,\,i}^2 \right)\,.
\end{align}
Now
\begin{align*}
&
\int_{-\infty}^\infty \p_\text{true}(\tilde{y}_i) \log \p_\Msk(\tilde{y}_i|\+y) \diff \tilde{y}_i
\\
&\qquad
= -c p \int_{-\infty}^\infty
    \tilde{y}_i^2 \exp\left( -a \tilde{y}_i^2 + b \tilde{y}_i \right)
    \diff \tilde{y}_i
\\
&\qquad\quad
+ c q \int_{-\infty}^\infty
    \tilde{y}_i \exp\left( -a \tilde{y}_i^2 + b \tilde{y}_i \right)
    \diff \tilde{y}_i
\\
&\qquad\quad
+ c r \int_{-\infty}^\infty
    \exp\left( -a \tilde{y}_i^2 + b \tilde{y}_i \right)
    \diff \tilde{y}_i\,.
    \numberthis
\end{align*}
These integrals are
\begin{align}
\int_{-\infty}^\infty
    \tilde{y}_i^2 \exp\left( -a \tilde{y}_i^2 + b \tilde{y}_i \right)
    \diff \tilde{y}_i
    &= \frac{\sqrt{\pi}}{2a^{3/2}}\left( \frac{b^2}{2a} +1 \right) \exp\left(\frac{b^2}{4a}\right)
\intertext{\citep[][p. 360, Section~3.462, Eq~22.8]{table_of_integrals},}
\int_{-\infty}^\infty
    \tilde{y}_i \exp\left( -a \tilde{y}_i^2 + b \tilde{y}_i \right)
    \diff \tilde{y}_i
    &= \frac{\sqrt{\pi}b}{2a^{3/2}} \exp\left(\frac{b^2}{4a}\right)
\intertext{\citep[][p. 360, Section~3.462, Eq~22.8]{table_of_integrals}, and}
\int_{-\infty}^\infty
    \exp\left( -a \tilde{y}_i^2 + b \tilde{y}_i \right)
    \diff \tilde{y}_i
    &= \frac{\sqrt{\pi}}{a^{1/2}} \exp\left(\frac{b^2}{4a}\right)
\end{align}
\citep[][p. 333, Section~3.323, Eq~2.10]{table_of_integrals}.
Now we can simplify
\begin{align*}
&\int_{-\infty}^\infty \p_\text{true}(\tilde{y}_i) \log \p_\Msk(\tilde{y}_i|\+y) \diff \tilde{y}_i
\\
&\qquad
= \sqrt{\pi} \exp\left(\frac{b^2}{4a} + \log c \right) \left(
    -\frac{pb^2}{4a^{5/2}} - \frac{p}{2a^{3/2}} + \frac{qb}{2a^{3/2}} + \frac{r}{a^{1/2}}
\right)
\\
&\qquad
= \sqrt{\pi} \left(2 \pi \widetilde{\sigma}_{i}^2\right)^{-1/2} \left(
    -\sqrt{2}p\widetilde{\mu}_{i}^2\widetilde{\sigma}_{i} - \sqrt{2} p \widetilde{\sigma}_{i}^3 + \sqrt{2}q\widetilde{\mu}_{i}\widetilde{\sigma}_{i} + \sqrt{2} r \widetilde{\sigma}_{i}
\right)
\\
&\qquad
= - p\widetilde{\mu}_{i}^2 - p \widetilde{\sigma}_{i}^2 + q \widetilde{\mu}_{i} + r
\\
&\qquad
= \frac{-\widetilde{\mu}_{i}^2 - \widetilde{\sigma}_{i}^2 + 2 \widetilde{\mu}_{i}\mu_{\Msk,\,i} - \mu_{\Msk,\,i}^2}{2\sigma_{\Msk,\,i}^2}
- \frac{1}{2}\log\left(2 \pi \sigma_{\Msk,\,i}^2 \right)
\\
&\qquad
= - \frac{
    \left(\mu_{\Msk,\,i} - \widetilde{\mu}_{i}\right)^2 + \widetilde{\sigma}_{i}^2
    }{
    2 \sigma_{\Msk,\,i}^2
    }
    - \frac{1}{2} \log\left(2 \pi \sigma_{\Msk,\,i}^2 \right)\,.
    \numberthis
\end{align*}

Let $\+P_{\Msk}$ be the following orthogonal projection matrix for model \Mk{}:
\begin{align*}
    \+P_{\Msk} &= \+X_{[\bcdot,\Msk]} \left(\+X_{[\bcdot,\Msk]}^\transp \+X_{[\bcdot,\Msk]}\right)^{-1}\+X_{[\bcdot,\Msk]}^\transp
    \numberthis
\intertext{so that}
    \mu_{\Msk,\,i} &= \+P_{\Msk [i,\bcdot]} \+y
    \\ &= \+P_{\Msk [i,\bcdot]} \left( \+X \+\beta + \+\varepsilon \right)
    \\ &= \+P_{\Msk [i,\bcdot]} \+X \+\beta + \+P_{\Msk [i,\bcdot]} \+\varepsilon
    \numberthis
    \\
    \sigma_{\Msk,\,i}^2 &= \big(1 + \+P_{\Msk [i,i]} \big) \tau^2 \,.
    \numberthis
\end{align*}
Now we can write
\begin{align*}
\left(\mu_{\Msk,\,i} - \widetilde{\mu}_{i}\right)^2
&= \left(
        \+P_{\Msk [i,\bcdot]} \+\varepsilon + \+P_{\Msk [i,\bcdot]} \+X \+\beta - \+X_{[i,\bcdot]} \+\beta - \mu_{\star,\,i}
\right)^2
\\
&= \+\varepsilon^\transp \+P_{\Msk [i,\bcdot]}^\transp \+P_{\Msk [i,\bcdot]} \+\varepsilon
    \\&\quad
    + 2 \left(
        \+P_{\Msk [i,\bcdot]} \+X \+\beta - \+X_{[i,\bcdot]} \+\beta - \mu_{\star,\,i}
    \right) \+P_{\Msk [i,\bcdot]} \+\varepsilon
    \\&\quad
    + \left(
        \+P_{\Msk [i,\bcdot]} \+X \+\beta - \+X_{[i,\bcdot]} \+\beta - \mu_{\star,\,i}
    \right)^2
    \,.
    \numberthis
\end{align*}
The integral simplifies to
\begin{gather}
\int_{-\infty}^\infty \p_\text{true}(\tilde{y}_i) \log \p_\Msk(\tilde{y}_i|\+y) \diff \tilde{y}_i
=
\+\varepsilon^\transp \+A_{\Msk,i} \+\varepsilon + \+b_{\Msk,i}^\transp \+\varepsilon + c_{\Msk,i}\,,
\end{gather}
where
\begin{align}
\+A_{\Msk,i} &= - \frac{1}{2 \big(1 + \+P_{\Msk [i,i]} \big) \tau^2}
    \+P_{\Msk [i,\bcdot]}^\transp \+P_{\Msk [i,\bcdot]}
\\
\+b_{\Msk,i} &= - \frac{1}{\big(1 + \+P_{\Msk [i,i]} \big) \tau^2}
    \+P_{\Msk [i,\bcdot]}^\transp \left(
        \+P_{\Msk [i,\bcdot]} \+X \+\beta - \+X_{[i,\bcdot]} \+\beta - \mu_{\star,\,i}
    \right)
\\
c_{\Msk,i} &= - \frac{1}{2 \big(1 + \+P_{\Msk [i,i]} \big) \tau^2}
    \left(
        \left(
            \+P_{\Msk [i,\bcdot]} \+X \+\beta - \+X_{[i,\bcdot]} \+\beta - \mu_{\star,\,i}
        \right)^2
        + \sigma_{\star,\,i}^2
    \right)
    \nonumber\\&\qquad\quad
    - \frac{1}{2} \log\left(2 \pi \big(1 + \+P_{\Msk [i,i]} \big) \tau^2 \right)
    \,.
\end{align}
Let diagonal matrix
\begin{align}
\+D_{\Msk} &= \left(\left(\+P_{\Msk} \odot \eye\right) + \eye\right)^{-1}\,,
\intertext{where $\odot$ is the Hadamard (or element-wise) product, so that}
\left[\+D_{\Msk}\right]_{[i,i]} &= \left(\+P_{\Msk [i,i]} + 1\right)^{-1}
\nonumber
\\
&= \left(
        \+X_{\Msk [i,\bcdot]}
        (\+X_{\Msk}^\transp\+X_{\Msk})^{-1}
        \+X_{\Msk [i,\bcdot]}^\transp
        + 1 \right)^{-1}
\end{align}
for $i=1,2,\dots,n$.
Now $\elpdC{\Mk}{\y}$ can be written as
\begin{gather}
\elpdC{\Mk}{\y} = \sum_{i=1}^n \int_{-\infty}^\infty \p_\text{true}(\tilde{y}_i) \log \p_\Msk(\tilde{y}_i|\+y) \diff \tilde{y}_i
=
\+\varepsilon^\transp \+A_{\Msk} \+\varepsilon + \+b_{\Msk}^\transp \+\varepsilon + c_{\Msk}\,
\end{gather}
where
\begingroup
\allowdisplaybreaks
\begin{align*}
\+A_{\Msk} &= \sum_{i=1}^n \+A_{\Msk,i} \\
    &= - \frac{1}{2 \tau^2} \+P_{\Msk} \+D_{\Msk} \+P_{\Msk}
    \,,
    \numberthis
\\
\+b_{\Msk} &= \sum_{i=1}^n \+b_{\Msk,i}
\\
    &= - \frac{1}{\tau^2} \+P_{\Msk} \+D_{\Msk}\left( \+P_{\Msk} \+X \+\beta - \+X \+\beta - \+\mu_{\star} \right)
\\
    &= - \frac{1}{\tau^2} \Big(
        \+P_{\Msk} \+D_{\Msk} \left(\+P_{\Msk} - \eye \right) \+X \+\beta
        - \+P_{\Msk} \+D_{\Msk} \+\mu_{\star}
    \Big)
    \,,
    \numberthis
\\
c_{\Msk} &= \sum_{i=1}^n c_{\Msk,i}
\\
    &= - \frac{1}{2 \tau^2} \left(
    \Big( \left(\+P_{\Msk} - \eye \right) \+X \+\beta - \+\mu_{\star} \Big)^\transp
        \+D_{\Msk}
        \Big( \left(\+P_{\Msk} - \eye \right) \+X \+\beta - \+\mu_{\star} \Big)
    + \+\sigma_\star^\transp \+D_{\Msk} \+\sigma_\star
\right)
\\
 &\quad- \frac{n}{2} \log\left(2 \pi \tau^2 \right)  +\frac{1}{2} \log \prod_{i=1}^n \+D_{\Msk [i,i]}
\\
    &= - \frac{1}{2 \tau^2} \Bigg(
        \+\beta^\transp \+X^\transp \left(\+P_{\Msk} - \eye \right)^\transp \+D_{\Msk} \left(\+P_{\Msk} - \eye \right) \+X \+\beta
        \\&\qquad\qquad\quad
        -2 \+\beta^\transp \+X^\transp \left(\+P_{\Msk} - \eye \right)^\transp \+D_{\Msk} \+\mu_{\star}
        \\&\qquad\qquad\quad
        + \+\mu_{\star}^\transp \+D_{\Msk} \+\mu_{\star}
        + \+\sigma_\star^\transp \+D_{\Msk} \+\sigma_\star
    \Bigg)
\\
 &\quad- \frac{n}{2} \log\left(2 \pi \tau^2 \right)  +\frac{1}{2} \log \prod_{i=1}^n \+D_{\Msk [i,i]}
\,.
\numberthis
\end{align*}
\endgroup
Furthermore, we have
\begin{align*}
\left(\+P_{\Msk} - \eye \right) \+X \+\beta
&= \left(\+P_{\Msk} - \eye \right) \left( \+X_{[\bcdot,\Msk]} \+\beta_{\Msk} + \+X_{[\bcdot,-\Msk]} \+\beta_{-\Msk} \right)
\\
&= \+P_{\Msk} \+X_{[\bcdot,\Msk]} \+\beta_{\Msk}
    - \+X_{[\bcdot,\Msk]} \+\beta_{\Msk}
    + \left(\+P_{\Msk} - \eye \right) \+X_{[\bcdot,-\Msk]} \+\beta_{-\Msk}
\\
&= \+X_{[\bcdot,\Msk]} (\+X_{[\bcdot,\Msk]}^\transp \+X_{[\bcdot,\Msk]})^{-1}\+X_{[\bcdot,\Msk]}^\transp
    \+X_{[\bcdot,\Msk]} \+\beta_{\Msk}
    - \+X_{[\bcdot,\Msk]} \+\beta_{\Msk}
    + \left(\+P_{\Msk} - \eye \right) \+X_{[\bcdot,-\Msk]} \+\beta_{-\Msk}
\\
&= \+X_{[\bcdot,\Msk]} \+\beta_{\Msk}
    - \+X_{[\bcdot,\Msk]} \+\beta_{\Msk}
    + \left(\+P_{\Msk} - \eye \right) \+X_{[\bcdot,-\Msk]} \+\beta_{-\Msk}
\\
&= \left(\+P_{\Msk} - \eye \right) \+X_{[\bcdot,-\Msk]} \+\beta_{-\Msk}
\,.
\numberthis
\end{align*}
Now we can formulate $\elpdC{\Mk}{\y}$ and further $\elpdC{\Mfd}{\y}$ in the following sections.

\subsubsection{Elpd for One Model}
\label{app_sec_analytic_elpd_m}

In this section, we formulate $\elpdC{\Mk}{\y}$ for model $\Mk \in \{\Mfa, \Mfb\}$ in the problem setting defined in Appendix~\ref{app_sec_norm_lin_reg_case_study}.
Let $\+P_{\Msk}$, a function of $\+X_{[\bcdot,\Msk]}$, be the following orthogonal projection matrix:
\begin{align}
    \+P_{\Msk} &= \+X_{[\bcdot,\Msk]} \left(\+X_{[\bcdot,\Msk]}^\transp \+X_{[\bcdot,\Msk]}\right)^{-1}\+X_{[\bcdot,\Msk]}^\transp
\,.
\end{align}
Let diagonal matrix $\+D_{\Msk}$, a function of $\+X_{[\bcdot,\Msk]}$, be
\begin{align}
\+D_{\Msk} &= \left(\left(\+P_{\Msk} \odot \eye\right) + \eye\right)^{-1}\,,
\intertext{where $\odot$ is the Hadamard (or element-wise) product, so that}
\left[\+D_{\Msk}\right]_{[i,i]} &= \left(\+P_{\Msk [i,i]} + 1\right)^{-1}
= \left(
        \+X_{\Msk [i,\bcdot]}
        (\+X_{\Msk}^\transp\+X_{\Msk})^{-1}
        \+X_{\Msk [i,\bcdot]}^\transp
        + 1 \right)^{-1}
\end{align}
for $i=1,2,\dots,n$.
Let
\begin{align}
\label{app_eq_yhat_def}
\hat{\+y}_{-\Msk} = \+X_{[\bcdot,-\Msk]} \+\beta_{-\Msk}
\,.
\end{align}
Following the derivations in Appendix~\ref{app_sec_analytic_elpd}, we get the following quadratic form for $\elpdC{\Mk}{\y}$:
\begin{gather}
\label{app_eq_analytic_elpd_m}
\elpdC{\Mk}{\y}
= \+\varepsilon^\transp \+A_{\Msk} \+\varepsilon + \+b_{\Msk}^\transp \+\varepsilon + c_{\Msk}
\,,
\end{gather}
where
\begin{align}
\+A_{\Msk} &= \frac{1}{\tau^2} \+A_{\Msk,1}
\,,
\\
\+b_{\Msk} &= \frac{1}{\tau^2} \left(
    \+B_{\Msk,1} \hat{\+y}_{-\Msk}
    + \+B_{\Msk,2} \+\mu_{\star}
    \right)
\,,
\\
c_{\Msk} &= \frac{1}{\tau^2} \left(
        \hat{\+y}_{-\Msk}^\transp \+C_{\Msk,1} \hat{\+y}_{-\Msk}
        + \hat{\+y}_{-\Msk}^\transp \+C_{\Msk,2} \+\mu_{\star}
        + \+\mu_{\star}^\transp \+C_{\Msk,3} \+\mu_{\star}
        + \+\sigma_\star^\transp \+C_{\Msk,3} \+\sigma_\star
    \right)
    + c_{\Msk,4}
\,,
\end{align}
where each matrix $\+A_{\Msk,\bcdot}$\,, $\+B_{\Msk,\bcdot}$\,, and $\+C_{\Msk,\bcdot}$ and scalar $c_{\Msk,4}$ are functions of $\+X_{[\bcdot,\Msk]}$:
\begin{align}
\+A_{\Msk,1} &= - \frac{1}{2} \+P_{\Msk} \+D_{\Msk} \+P_{\Msk}
\,,
\\
\+B_{\Msk,1} &= - \+P_{\Msk} \+D_{\Msk} \left(\+P_{\Msk} - \eye \right)
\,,
\\
\+B_{\Msk,2} &= \+P_{\Msk} \+D_{\Msk}
\,,
\\
\+C_{\Msk,1} &= - \frac{1}{2} \left(\+P_{\Msk} - \eye \right) \+D_{\Msk} \left(\+P_{\Msk} - \eye \right)
\,,
\\
\+C_{\Msk,2} &= \left(\+P_{\Msk} - \eye \right) \+D_{\Msk}
\,,
\\
\+C_{\Msk,3} &= - \frac{1}{2} \+D_{\Msk}
\,,
\\
c_{\Msk,4} &= \frac{1}{2} \log \prod_{i=1}^n \+D_{\Msk [i,i]} - \frac{n}{2} \log\left(2 \pi \tau^2 \right)
\,.
\end{align}

\subsubsection{Elpd for the Difference}
\label{app_sec_analytic_elpd_d}

In this section, we formulate $\elpdC{\Mfd}{\y}$ in the problem setting defined in Appendix~\ref{app_sec_norm_lin_reg_case_study}.
Following the derivations in Appendix~\ref{app_sec_analytic_elpd_m} by applying Equation~\eqref{app_eq_analytic_elpd_m} for models \Mfa{} and \Mfb{}, we get the following quadratic form for the difference:
\begin{align}
\label{app_eq_analytic_elpd_d}
\elpdC{\Mfd}{\y}
= \+\varepsilon^\transp \+A_\Mfsd \+\varepsilon + \+b_\Mfsd^\transp \+\varepsilon + c_\Mfsd
\,,
\end{align}
where
\begin{align*}
\+A_\Mfsd
&= \frac{1}{\tau^2} \+A_{\Mfsd,1}
\,,
\numberthis
\\
\+b_\Mfsd
&= \frac{1}{\tau^2} \left(
    \+B_{\Mfsa,1} \hat{\+y}_{-\Mfsa}
    - \+B_{\Mfsb,1} \hat{\+y}_{-\Mfsb}
    + \+B_{\Mfsd,2} \+\mu_{\star}
\right)
\,,
\numberthis
\\
c_{\Mfsd} &= \frac{1}{\tau^2} \Bigg(
     \hat{\+y}_{-\Mfsa}^\transp \+C_{\Mfsa,1}  \hat{\+y}_{-\Mfsa}
    -  \hat{\+y}_{-\Mfsb}^\transp \+C_{\Mfsb,1}  \hat{\+y}_{-\Mfsb}
    \\&\qquad\quad
    +  \hat{\+y}_{-\Mfsa}^\transp \+C_{\Mfsa,2} \+\mu_{\star}
    -  \hat{\+y}_{-\Mfsb}^\transp \+C_{\Mfsb,2} \+\mu_{\star}
    \\&\qquad\quad
    + \+\mu_{\star}^\transp \+C_{\Mfsd,3} \+\mu_{\star}
    + \+\sigma_\star^\transp \+C_{\Mfsd,3} \+\sigma_\star
    \Bigg)
    + c_{\Mfsd,4}
\,.
\numberthis
\end{align*}
where matrices $\+A_{\Mfsd,1}$\,, $\+B_{\Mfsd,2}$\,, and $\+C_{\Mfsd,3}$ and scalar $c_{\Mfsd,4}$ are functions of $\+X$:
\begin{align}
\+A_{\Mfsd,1} 
    &= - \frac{1}{2} \left( \+P_{\Mfsa} \+D_{\Mfsa} \+P_{\Mfsa} - \+P_{\Mfsb} \+D_{\Mfsb} \+P_{\Mfsb} \right)
\,,
\\
\+B_{\Mfsd,2} 
    &= \+P_{\Mfsa} \+D_{\Mfsa} - \+P_{\Mfsb} \+D_{\Mfsb}
\,,
\\
\+C_{\Mfsd,3} 
    &= - \frac{1}{2} \left( \+D_{\Mfsa} - \+D_{\Mfsb}\right)
\,,
\\
c_{\Mfsd,4} 
    &= \frac{1}{2} \log \left( \prod_{i=1}^n \frac{\+D_{\Mfsa,[i,i]}}{\+D_{\Mfsb,[i,i]}} \right)
\,,
\end{align}
and matrices $\+B_{\Msk,1}$, $\+C_{\Msk,1}$, and $\+C_{\Msk,2}$, functions of $\+X_{[\bcdot,\Msk]}$, for $\Mk \in \{\Mfa, \Mfb\}$ are defined in Appendix~\ref{app_sec_analytic_elpd_m}:
\begin{align}
\+B_{\Msk,1} &= - \+P_{\Msk} \+D_{\Msk} \left(\+P_{\Msk} - \eye \right)
\,,
\\
\+C_{\Msk,1} &= - \frac{1}{2} \left(\+P_{\Msk} - \eye \right) \+D_{\Msk} \left(\+P_{\Msk} - \eye \right)
\,,
\\
\+C_{\Msk,2} &= \left(\+P_{\Msk} - \eye \right) \+D_{\Msk}
\,.
\end{align}
It can be seen that all these parameters do not depend on the shared covariate effects, that it is the effects $\beta_i$ that are included in both $\+\beta_\Mfsa$ and $\+\beta_\Mfsb$.


\subsection{LOO-CV Estimate}
\label{app_sec_analytic_loocv}

In this section, we present the analytic form for $\elpdHatC{\Mk}{\y}$ for model $\Mk \in \{\Mfa, \Mfb\}$.
Restating from the problem statement in the beginning of Appendix~\ref{app_sec_norm_lin_reg_case_study}, the likelihood for model $\Mk$ is formalised as
\begin{gather}
    \+y \Big| \widehat{\+\beta}_\Msk, \+X_{[\bcdot,\Msk]}, \tau^2 \sim \operatorname{N}\left(\+X_{[\bcdot,\Msk]} \widehat{\+\beta}_\Msk, \tau^2\eye\right).
\end{gather}
Analogous to the posterior predictive distribution for the full data as presented in Equation~\eqref{app_eq_lin_reg_full_pred_dist}, with uniform prior distribution, the LOO-CV posterior predictive distribution for observation $i$ follows a normal distribution
\begin{gather}
    y_i \Big| \+y_{-i}, \+X_{[-i,\Msk]}, \+X_{[i,\Msk]}, \tau^2 \sim \operatorname{N}\left(
        \widetilde{\mu}_{\Msk\,i}, \widetilde{\sigma}_{\Msk\,i}\right)^2,
\end{gather}
where
\begin{align}
    \widetilde{\mu}_{\Msk\,i}
    &= \+X_{[i,\Msk]}
        \left(\+X_{[-i,\Msk]}^\transp \+X_{[-i,\Msk]}\right)^{-1}
        \+X_{[-i,\Msk]}^\transp \+y_{-i}
\,,
\\
    \widetilde{\sigma}_{\Msk\,i}^2
    &= \left(1 + \+X_{[i,\Msk]}
        \left(\+X_{[-i,\Msk]}^\transp \+X_{[-i,\Msk]}\right)^{-1}
        \+X_{[i,\Msk]}^\transp \right) \tau^2
\,.
\end{align}
We have
\begin{gather}
    \+y_{-i} = \+X_{[-i,\bcdot]} \+\beta + \+\varepsilon_{-i} = \+X_{[-i,\Msk]}\+\beta_{\Msk} + \+X_{[-i,-\Msk]}\+\beta_{-\Msk} + \+\varepsilon_{-i}.
\end{gather}
Let vector
\begin{gather}
    \+v(\Mk,i) =
        \+X_{[\bcdot,\Msk]}
        (\+X_{[-i,\Msk]}^\transp\+X_{[-i,\Msk]})^{-1}
        \+X_{[i,\Msk]}^\transp.
\end{gather}
The predictive distribution parameters can be formulated as
\begin{align*}
    \widetilde{\mu}_{\Msk\,i} &= \+v(\Mk,i)_{-i}^\transp \+y_{-i}
    \\
        &= \+v(\Mk,i)_{-i}^\transp \+\varepsilon_{-i}
        + \+v(\Mk,i)_{-i}^\transp \+X_{[-i,\Msk]}\+\beta_{\Msk}
        + \+v(\Mk,i)_{-i}^\transp \+X_{[-i,-\Msk]} \+\beta_{-\Msk}
    \\
        &= \+v(\Mk,i)_{-i}^\transp \+\varepsilon_{-i}
        + \+X_{[i,\Msk]} (\+X_{[-i,\Msk]}^\transp\+X_{[-i,\Msk]})^{-1} \+X_{[-i,\Msk]}^\transp \+X_{[-i,\Msk]}\+\beta_{\Msk}
        + \+v(\Mk,i)_{-i}^\transp \+X_{[-i,-\Msk]} \+\beta_{-\Msk}
    \\
        &= \+v(\Mk,i)_{-i}^\transp \+\varepsilon_{-i}
        + \+X_{[i,\Msk]} \+\beta_{\Msk}
        + \+v(\Mk,i)_{-i}^\transp \+X_{[-i,-\Msk]} \+\beta_{-\Msk}
    \numberthis
\intertext{and}
    \widetilde{\sigma}_{\Msk\,i}^2 &= \left( v(\Mk,i)_{i} + 1\right)\tau^2.
    \numberthis
\end{align*}
Let vector $\+w(\Mk,i)$ denote $\+v(\Mk,i)$ where the $i$th element is replaced with $-1$:
\begin{gather}
w(\Mk,i)_j = \begin{dcases} -1, & \text{if } j=i \\ v(\Mk,i)_{j} & \text{if } j \neq i. \end{dcases}
\end{gather}
Now
\begin{align}
\+w(\Mk,i)^\transp \+\varepsilon &= \+v(\Mk,i)_{-i}^\transp \+\varepsilon_{-i} - \varepsilon_i
\,,
\\
\+w(\Mk,i)^\transp \+X_{[\bcdot,-\Msk]} &= \+v(\Mk,i)_{-i}^\transp \+X_{[-i,-\Msk]} - X_{[i,-\Msk]}.
\end{align}
The LOO-CV term for observation $i$ is
\begin{align}
    \elpdHatiC{\Mk}{\y}{i} &= \log \p(y_i |  \+y_{-i}, \+X_{[-i,\Msk]}, \+X_{[i,\Msk]}, \tau^2)
    \nonumber
    \\
    &= -\frac{1}{2 \widetilde{\sigma}_{\Msk\,i}^2} (y_i - \widetilde{\mu}_{\Msk\,i})^2
        - \frac{1}{2} \log(2\pi \widetilde{\sigma}_{\Msk\,i}^2).
\end{align}
As
\begin{align*}
    y_i - \widetilde{\mu}_{\Msk\,i}
        &= \+X_{[i,\bcdot]}\+\beta + \varepsilon_i
            - \+v(\Mk,i)_{-i}^\transp \+\varepsilon_{-i}
            - \+X_{[i,\Msk]} \+\beta_{\Msk}
            - \+v(\Mk,i)_{-i}^\transp \+X_{[-i,-\Msk]} \+\beta_{-\Msk}
    \\ \qquad
        &= \+X_{[i,\Msk]}\+\beta_{\Msk} + \+X_{[i,-\Msk]} \+\beta_{-\Msk}
            + \varepsilon_i
            - \+v(\Mk,i)_{-i}^\transp \+\varepsilon_{-i}
        \\&\quad
            - \+X_{[i,\Msk]} \+\beta_{\Msk}
            - \+v(\Mk,i)_{-i}^\transp \+X_{[-i,-\Msk]} \+\beta_{-\Msk}
    \\ \qquad
        &= - \left( \+v(\Mk,i)_{-i}^\transp \+\varepsilon_{-i} - \varepsilon_i \right)
            - \left( \+v(\Mk,i)_{-i}^\transp \+X_{[-i,-\Msk]} \+\beta_{-\Msk} - \+X_{[i,-\Msk]} \+\beta_{-\Msk} \right)
    \\ \qquad
        &= - \left( \+w(\Mk,i)^\transp \+\varepsilon
            + \+w(\Mk,i)^\transp \+X_{[\bcdot,-\Msk]} \+\beta_{-\Msk} \right),
    \numberthis
\end{align*}
we get
\begin{align*}
    \left( y_i - \widetilde{\mu}_{\Msk\,i} \right)^2
    &=   \+\varepsilon^\transp \+w(\Mk,i)
            \+w(\Mk,i)^\transp \+\varepsilon
    \\&\quad
        + 2 \+\beta_{-\Msk}^\transp \+X_{[\bcdot,-\Msk]}^\transp \+w(\Mk,i)
                \+w(\Mk,i)^\transp \+\varepsilon
    \\&\quad
        + \+\beta_{-\Msk}^\transp \+X_{[\bcdot,-\Msk]}^\transp \+w(\Mk,i)
                \+w(\Mk,i)^\transp \+X_{[\bcdot,-\Msk]} \+\beta_{-\Msk}
    \numberthis
\end{align*}
and
\begin{gather}
    \elpdHatiC{\Mk}{\y}{i}
    = \+\varepsilon^\transp \widetilde{\+A}_{\Msk\,i} \+\varepsilon
        + \widetilde{\+b}_{\Msk\,i}^\transp \+\varepsilon
        + \widetilde{c}_{\Msk\,i},
\end{gather}
where
\begin{align}
    \widetilde{\+A}_{\Msk\,i} &= -\frac{1}{2 \left(v(\Mk,i)_{i} +1\right)\tau^2}
        \+w(\Mk,i) \+w(\Mk,i)^\transp
    \,,
    \\
    \widetilde{\+b}_{\Msk\,i} &= -\frac{1}{\left(v(\Mk,i)_{i} +1\right)\tau^2}
        \+w(\Mk,i) \+w(\Mk,i)^\transp \+X_{[\bcdot,-\Msk]} \+\beta_{-\Msk}
    \,,
    \\
    \widetilde{c}_{\Msk\,i} &= -\frac{1}{2 \left(v(\Mk,i)_{i} +1\right)\tau^2}
            \+\beta_{-\Msk}^\transp \+X_{[\bcdot,-\Msk]}^\transp \+w(\Mk,i)
            \+w(\Mk,i)^\transp \+X_{[\bcdot,-\Msk]} \+\beta_{-\Msk}
    \nonumber\\&\quad
        - \frac{1}{2} \log\left(2\pi \left(v(\Mk,i)_{i}+1\right)\tau^2\right).
\end{align}
From this, by summing over all $i=1,2,\dots,n$, we get the LOO-CV approximation for model \Mk{}.
We present $\elpdHatC{\Mk}{\y}$ and further $\elpdHatC{\Mfd}{\y}$ in the following sections.
\subsubsection{LOO-CV Estimate for One Model}
\label{app_sec_analytic_loocv_m}
In this section, we formulate $\elpdHatC{\Mk}{\y}$ for model $\Mk \in \{\Mfa, \Mfb\}$ in the problem setting defined in Appendix~\ref{app_sec_norm_lin_reg_case_study}.
Let matrix $\widetilde{\+P}_{\Msk}$, a function of $\+X_{[\bcdot,\Msk]}$, have the following elements:
\begin{align}
\left[\widetilde{\+P}_{\Msk}\right]_{[i,j]} &=
\begin{dcases}
        -1
    , & \text{when } i = j,
    \\
        \+X_{[j,\Msk]}
        (\+X_{[-i,\Msk]}^\transp\+X_{[-i,\Msk]})^{-1}
        \+X_{[i,\Msk]}^\transp
    , & \text{when } i \neq j
    \,,
\end{dcases}
\intertext{and let diagonal matrix $\widetilde{\+D}_{\Msk}$, a function of $\+X_{[\bcdot,\Msk]}$, have the following elements:}
\left[\widetilde{\+D}_{\Msk}\right]_{[i,i]} &=
    \left(
        \+X_{[i,\Msk]}
        (\+X_{[-i,\Msk]}^\transp\+X_{[-i,\Msk]})^{-1}
        \+X_{[i,\Msk]}^\transp
        + 1
    \right)^{-1},
\end{align}
where $i,j = 1,2,\dots,n$.
Let
\begin{align}
\hat{\+y}_{-\Msk} &= \+X_{[\bcdot,-\Msk]} \+\beta_{-\Msk}
\,.
\end{align}
Following the derivations in Appendix~\ref{app_sec_analytic_loocv}, we obtain the following quadratic form for $\elpdHatC{\Mk}{\y}$:
\begin{gather}
    \label{app_eq_analytic_elpdhat_m}
    \elpdHatC{\Mk}{\y} =
        \+\varepsilon^\transp \widetilde{\+A}_{\Msk} \+\varepsilon
            + \widetilde{\+b}_{\Msk}^\transp \+\varepsilon
            + \widetilde{c}_{\Msk},
\end{gather}
where
\begin{align}
\widetilde{\+A}_{\Msk} 
    &= \frac{1}{\tau^2} \widetilde{\+A}_{\Msk,1}
    \,,
    \\
\widetilde{\+b}_{\Msk} 
    &= \frac{1}{\tau^2} \widetilde{\+B}_{\Msk,1} \hat{\+y}_{-\Msk}
    \,,
    \\
\widetilde{c}_{\Msk} 
    &= \frac{1}{\tau^2}
        \hat{\+y}_{-\Msk}^\transp \widetilde{\+C}_{\Msk,1} \hat{\+y}_{-\Msk}
        + \widetilde{c}_{\Msk,4}
    \,,
\end{align}
where matrices $\widetilde{\+A}_{\Msk,1}$, $\widetilde{\+B}_{\Msk,1}$, and $\widetilde{\+C}_{\Msk,1}$ and scalar $\widetilde{c}_{\Msk,4}$ are functions of $\+X_{[\bcdot,\Msk]}$:
\begin{align}
\widetilde{\+A}_{\Msk,1} &= -\frac{1}{2} \widetilde{\+P}_\Msk^\transp \widetilde{\+D}_{\Msk} \widetilde{\+P}_\Msk
\,,
\\
\widetilde{\+B}_{\Msk,1} &= - \widetilde{\+P}_\Msk^\transp \widetilde{\+D}_{\Msk} \widetilde{\+P}_\Msk
\,,
\\
\widetilde{\+C}_{\Msk,1} &= -\frac{1}{2} \widetilde{\+P}_\Msk^\transp \widetilde{\+D}_{\Msk} \widetilde{\+P}_\Msk
\,,
\\
\widetilde{c}_{\Msk,4} &= \frac{1}{2} \log\left( \prod_{i=1}^n  \widetilde{\+D}_{\Msk[i,i]}  \right) - \frac{n}{2} \log\left(2\pi\tau^2 \right)
\,.
\end{align}

%
\subsubsection{LOO-CV Estimate for the Difference}
\label{app_sec_analytic_loocv_d}
In this section, we formulate $\elpdHatC{\Mfd}{\y}$ in the problem setting defined in Appendix~\ref{app_sec_norm_lin_reg_case_study}.
Following the derivations in Appendix~\ref{app_sec_analytic_loocv_m} by applying Equation~\eqref{app_eq_analytic_elpdhat_m} for models \Mfa{} and \Mfb{}, we get the following quadratic form for the difference:
\begin{gather}
    \label{app_eq_analytic_elpdhat_d}
    \elpdHatC{\Mfd}{\y}
    = \+\varepsilon^\transp \widetilde{\+A}_{\Mfsd} \+\varepsilon + \widetilde{\+b}_{\Mfsd}^\transp \+\varepsilon + \widetilde{c}_{\Mfsd},
\end{gather}
where
\begin{align}
\widetilde{\+A}_{\Mfsd}
    &= \frac{1}{\tau^2} \widetilde{\+A}_{\Mfsd,1}
    \,,
\\
\widetilde{\+b}_{\Mfsd}
    &= \frac{1}{\tau^2} \left( \widetilde{\+B}_{\Mfsa,1} \hat{\+y}_{-\Mfsa} - \widetilde{\+B}_{\Mfsb,1} \hat{\+y}_{-\Mfsb} \right)
    \,,
\\
\widetilde{c}_{\Mfsd}
    &= \frac{1}{\tau^2}  \left(
        \hat{\+y}_{-\Mfsa}^\transp \widetilde{\+C}_{\Mfsa,1} \hat{\+y}_{-\Mfsa}
        - \hat{\+y}_{-\Mfsb}^\transp \widetilde{\+C}_{\Mfsb,1} \hat{\+y}_{-\Mfsb}
     \right)
    + \widetilde{c}_{\Mfsd,4}
\,,
\end{align}
where matrix $\widetilde{\+A}_{\Mfsd,1}$ and scalar $\widetilde{c}_{\Mfsd,4}$ are functions of $\+X$:
\begin{align}
\widetilde{\+A}_{\Mfsd,1} 
    &= -\frac{1}{2} \left( \widetilde{\+P}_{\Mfsa}^\transp \widetilde{\+D}_{\Mfsa} \widetilde{\+P}_{\Mfsa} - \widetilde{\+P}_{\Mfsb}^\transp \widetilde{\+D}_{\Mfsb} \widetilde{\+P}_{\Mfsb} \right)
    \,,
\\
\widetilde{c}_{\Mfsd,4} 
    &= \frac{1}{2} \log\left(\prod_{i=1}^n \frac{\widetilde{\+D}_{\Mfsa[i,i]}}{\widetilde{\+D}_{\Mfsb[i,i]}} \right)
\,,
\end{align}
and matrices $\widetilde{\+B}_{\Msk,1}$ and $\widetilde{\+C}_{\Msk,1}$, functions of $\+X_{[\bcdot,\Msk]}$, for $\Mk \in \{\Mfa, \Mfb\}$ are defined in Appendix~\ref{app_sec_analytic_loocv_m}:
\begin{align}
 \widetilde{\+B}_{\Msk,1} &= - \widetilde{\+P}_\Msk^\transp \widetilde{\+D}_{\Msk} \widetilde{\+P}_\Msk
 \,,
 \\
 \widetilde{\+C}_{\Msk,1} &= -\frac{1}{2} \widetilde{\+P}_\Msk^\transp \widetilde{\+D}_{\Msk} \widetilde{\+P}_\Msk
 \,.
\end{align}
It can be seen that all these parameters do not depend on the shared covariate effects, that it is the effects $\beta_i$ that are included in both $\+\beta_\Mfsa$ and $\+\beta_\Mfsb$.
\subsubsection{Additional Properties for the Parameters of the LOO-CV Estimate}
\label{app_sec_analytic_loocv_additional_properties}
In this section, we present some additional properties for the matrix parameters $\widetilde{\+P}_{\Msk}$ and $\widetilde{\+D}_{\Msk}$ for $\Mk \in \{\Mfa, \Mfb\}$ defined in Appendix~\ref{app_sec_analytic_loocv_m} and for $\widetilde{\+A}_{\Mfsd}$ defined in Appendix~\ref{app_sec_analytic_loocv_d}.
Trivially, product $\widetilde{\+P}_\Msk^\transp \widetilde{\+D}_{\Msk} \widetilde{\+P}_\Msk$ is symmetric.
Being a sum of two such matrices, it is clear that matrix $\widetilde{\+A}_{\Mfsd}$ is also symmetric.
Element $(i,j)$, $i,j=1,2,\dots,n$, of the product $\widetilde{\+P}_\Msk^\transp \widetilde{\+D}_{\Msk} \widetilde{\+P}_\Msk$ can be written as
\begin{gather}
 \label{app_eq_PtT_Pt_m}
 \left[\widetilde{\+P}_\Msk^\transp \widetilde{\+D}_{\Msk} \widetilde{\+P}_\Msk\right]_{[i,j]} =
    \begin{dcases}
    \sum_{p\neq\{i\}} \frac{v(\Mk,p)_i^2}{v(\Mk,p)_p +1}
        + \frac{1}{v(\Mk,i)_i + 1},
        &\text{when } i=j,
    \\
        \sum_{p\neq\{i,j\}}
            \frac{v(\Mk,p)_i v(\Mk,p)_j}{v(\Mk,p)_p +1}
        - \frac{v(\Mk,i)_j}{v(\Mk,i)_i +1}
        - \frac{v(\Mk,j)_i}{v(\Mk,j)_j +1},
        &\text{when } i \neq j,
    \end{dcases}
\end{gather}
where $v(\Mk,a)_b$ follows the definition in Appendix~\ref{app_sec_analytic_loocv}.
Sum of squares of each row in $\widetilde{\+D}_{\Msk}^{1/2}\widetilde{\+P}_{\Msk}$ sum up to 1:
\begin{align*}
\sum_{i=1}^n \left[\widetilde{\+D}_{\Msk}^{1/2} \widetilde{\+P}_{\Msk} \right]_{[i,j]}^2
&=
    \frac{
        \+X_{[j,\Msk]}
        \left(\+X_{[-i,\Msk]}^\transp\+X_{[-i,\Msk]}\right)^{-1}
            \+X_{[i,\Msk]}^\transp
        \left(
            \+X_{[j,\Msk]}
            \left(\+X_{[-i,\Msk]}^\transp\+X_{[-i,\Msk]}\right)^{-1}
            \+X_{[i,\Msk]}^\transp
        \right)^\transp
        +1
    }{
        \+X_{[i,\Msk]}
        \left(\+X_{[-i,\Msk]}^\transp\+X_{[-i,\Msk]}\right)^{-1}
        \+X_{[i,\Msk]}^\transp
        + 1
    }
\\
&=
    \frac{
        \+X_{[j,\Msk]}
        \left(\+X_{[-i,\Msk]}^\transp\+X_{[-i,\Msk]}\right)^{-1}
        \left(
            \+X_{[i,\Msk]}^\transp
            \+X_{[i,\Msk]}
        \right)
        \left(\+X_{[-i,\Msk]}^\transp\+X_{[-i,\Msk]}\right)^{-1}
        \+X_{[j,\Msk]}^\transp
        +1
    }{
        \+X_{[i,\Msk]}
        \left(\+X_{[-i,\Msk]}^\transp\+X_{[-i,\Msk]}\right)^{-1}
        \+X_{[i,\Msk]}^\transp
        + 1
    }
\\
&= 1
\,.
\numberthis
\end{align*}
As sum of squares of each row in $\widetilde{\+D}_{\Mfsa}^{1/2} \widetilde{\+P}_{\Mfsa}$ and $\widetilde{\+D}_{\Mfsb}^{1/2} \widetilde{\+P}_{\Mfsb}$ sum up to 1,
trace of $\widetilde{\+A}_{\Mfsd}$ equals to 0:
\begin{align*}
    \operatorname{tr}\left(\widetilde{\+A}_{\Mfsd}\right)
    &= -\frac{1}{2\tau^2} \left( \sum_{i=1}^n \sum_{j=1}^n \left[\widetilde{\+D}_{\Mfsa}^{1/2} \widetilde{\+P}_{\Mfsa}\right]_{[i,j]}^2 -\sum_{i=1}^n \sum_{j=1}^n \left[\widetilde{\+D}_{\Mfsb}^{1/2} \widetilde{\+P}_{\Mfsb}\right]_{[i,j]}^2 \right) \\
    &= -\frac{1}{2\tau^2} \left( n - n \right). \\
    &= 0
    \label{app_trace_a_zero}
    \numberthis
\end{align*}
From this, it can be concluded that the sum of eigenvalues of $\widetilde{\+A}_{\Mfsd}$ is zero and $\widetilde{\+A}_{\Mfsd}$ is indefinite matrix or zero matrix.


\subsection{LOO-CV Error}
\label{app_sec_analytic_error}

In this section, we formulate the error $\elpdHatErrC{\Mfd}{\y} = \elpdHatC{\Mfd}{\y} - \elpdC{\Mfd}{\y}$ in the problem setting defined in Appendix~\ref{app_sec_norm_lin_reg_case_study}.
Following the derivations in Appendix~\ref{app_sec_analytic_elpd_d} and~\ref{app_sec_analytic_loocv_d} by applying Equation~\eqref{app_eq_analytic_elpd_d} and~\eqref{app_eq_analytic_elpdhat_d}, we get the following quadratic form for the error:
\begin{align}
    \label{app_eq_analytic_elpdhat_error_d}
    \elpdHatErrC{\Mfd}{\y}
    &= \+\varepsilon^\transp \+A_\mathrm{err} \+\varepsilon + \+b_\mathrm{err}^\transp \+\varepsilon + c_\mathrm{err}
    \,,
\end{align}
where
\begin{align*}
    \+A_\mathrm{err}
    &=
    \frac{1}{\tau^2} \+A_{\mathrm{err},1}
\,,
\numberthis
\\
    \+b_\mathrm{err}
    &= \frac{1}{\tau^2} \Big(
        \+B_{\mathrm{err},\Mfsa,1} \hat{\+y}_{-\Mfsa}
        - \+B_{\mathrm{err},\Mfsb,1} \hat{\+y}_{-\Mfsb}
        - \+B_{\Mfsd,2} \+\mu_\star
   \Big)
\,,
\numberthis \label{app_eq_error_param_b_err}
\\
    c_\mathrm{err}
    &=
    \frac{1}{\tau^2} \Bigg(
        \hat{\+y}_{-\Mfsa}^\transp \+C_{\mathrm{err},\Mfsa,1} \hat{\+y}_{-\Mfsa}
        - \hat{\+y}_{-\Mfsb}^\transp \+C_{\mathrm{err},\Mfsb,1} \hat{\+y}_{-\Mfsb}
        \\&\qquad\quad
        -  \hat{\+y}_{-\Mfsa}^\transp \+C_{\Mfsa,2} \+\mu_{\star}
        +  \hat{\+y}_{-\Mfsb}^\transp \+C_{\Mfsb,2} \+\mu_{\star}
        \\&\qquad\quad
        - \+\mu_{\star}^\transp \+C_{\Mfsd,3} \+\mu_{\star}
        - \+\sigma_\star^\transp \+C_{\Mfsd,3} \+\sigma_\star
    \Bigg)
     + c_{\mathrm{err},4}
\,,
\numberthis \label{app_eq_error_param_c_err}
\end{align*}
where matrix $\+A_{\mathrm{err},1}$ and matrices $ \+B_{\mathrm{err},\Mk,1}$ and $\+C_{\mathrm{err},\Mk,1}$ for $\Mk \in \{\Mfa, \Mfb\}$ and scalar $c_{\mathrm{err},4}$ are functions of $\+X$:
\begin{align}
    \+A_{\mathrm{err},1}
    &= \frac{1}{2} \left(
        \+P_{\Mfsa} \+D_{\Mfsa} \+P_{\Mfsa}
        - \widetilde{\+P}_{\Mfsa}^\transp \widetilde{\+D}_{\Mfsa} \widetilde{\+P}_{\Mfsa}
        - \+P_{\Mfsb} \+D_{\Mfsb} \+P_{\Mfsb}
        + \widetilde{\+P}_{\Mfsb}^\transp \widetilde{\+D}_{\Mfsb} \widetilde{\+P}_{\Mfsb}
    \right)
\,,
\\
    \+B_{\mathrm{err},\Msk,1}
    &=
        \+P_{\Msk} \+D_{\Msk} \left(\+P_{\Msk} - \eye \right)
        - \widetilde{\+P}_\Msk^\transp \widetilde{\+D}_{\Msk} \widetilde{\+P}_\Msk
\,,
\\
    \+C_{\mathrm{err},\Msk,1}
    &= \frac{1}{2}
        \left(
            \left(\+P_{\Msk} - \eye \right) \+D_{\Msk} \left(\+P_{\Msk} - \eye \right)
            - \widetilde{\+P}_\Msk^\transp \widetilde{\+D}_{\Msk} \widetilde{\+P}_\Msk
        \right)
\,,
\\
    c_{\mathrm{err},4}
    &= \frac{1}{2} \log \left( \prod_{i=1}^n
        \frac{
            \+D_{\Mfsb,[i,i]} \widetilde{\+D}_{\Mfsa[i,i]}
        }{
            \+D_{\Mfsa,[i,i]} \widetilde{\+D}_{\Mfsb[i,i]}
        }
    \right)
\,,
\end{align}
and matrix $\+C_{\Msk,2}$ for $\Mk \in \{\Mfa, \Mfb\}$ and matrices $\+B_{\Mfsd,2}$ and $\+C_{\Mfsd,3}$, functions of $\+X_{[\bcdot,\Msk]}$, are defined in appendices~\ref{app_sec_analytic_elpd_m} and~\ref{app_sec_analytic_elpd_d} respectively:
\begin{align}
\+C_{\Msk,2} &= \left(\+P_{\Msk} - \eye \right) \+D_{\Msk}
\,,
\\
\+B_{\Mfsd,2} &= \+P_{\Mfsa} \+D_{\Mfsa} - \+P_{\Mfsb} \+D_{\Mfsb}
\,,
\\
\+C_{\Mfsd,3} &= - \frac{1}{2} \left( \+D_{\Mfsa} - \+D_{\Mfsb}\right)
\,.
\end{align}
It can be seen that all these parameters do not depend on the shared covariate effects, that it is the effects $\beta_i$ that are included in both $\+\beta_\Mfsa$ and $\+\beta_\Mfsb$.

\subsection{Reparametrisation as a Sum of Independent Variables} 
\label{app_sec_analytic_as_sum_of_chi2}

By adapting Jacobi's theorem, variables $\elpdC{\Mk}{\y}$, $\elpdC{\Mfd}{\y}$, $\elpdHatC{\Mk}{\y}$, $\elpdHatC{\Md}{\y}$, and $\elpdC{\Mfd}{\y} - \elpdHatC{\Md}{\y}$ for $\Mk \in \{ \Mfsa, \Mfsb\}$, which are all of a quadratic form on $\varepsilon$, can also be expressed as a sum of independent scaled non-central $\chi^2$ distributed random variables with degree one plus a constant.
Let $Z$ denote the variable at hand.
First we write the variable using normalised $\widetilde{\+\varepsilon}=\+\Sigma_\star^{-1/2}(\+\varepsilon - \+\mu_\star)$:
\begin{align}
    Z
    &= \+\varepsilon^\transp \+A \+\varepsilon + \+b^\transp \+\varepsilon + c
    \nonumber
    \\
    &= \widetilde{\+\varepsilon}^\transp \widetilde{\+A} \widetilde{\+\varepsilon} + \widetilde{\+b}^\transp \widetilde{\+\varepsilon} + \widetilde{c}\,,
\end{align}
where
\begin{align}
    \widetilde{\+A} &= \+\Sigma_\star^{1/2} \+A \+\Sigma_\star^{1/2}
    \\
    \widetilde{\+b} &= \+\Sigma_\star^{1/2} \+b + 2  \+\Sigma_\star^{1/2} \+A \+\mu_\star
    \\
    \widetilde{c} &= c + \+b^\transp \+\mu_\star + \+\mu_\star^\transp \+A \+\mu_\star
\,.
\end{align}
Eliminate the linear term $\widetilde{\+b}^\transp \+\varepsilon$ using transformed variable $\+z = \widetilde{\+\varepsilon} + \+r \sim \mathrm{N}\left(\+r ,\eye\right)$, where $\+r$ is any vector satisfying the linear system $2\widetilde{\+A}\+r = \widetilde{\+b}$:
\begin{align*}
    Z
    &= \widetilde{\+\varepsilon}^\transp \widetilde{\+A} \widetilde{\+\varepsilon} + \widetilde{\+b}^\transp \widetilde{\+\varepsilon} + \widetilde{c}
    \\
    &= \left(\+z-\+r\right)^\transp \widetilde{\+A} \left(\+z-\+r\right)
        +\widetilde{\+b}^\transp \left(\+z-\+r\right) + \widetilde{c}
    \\
    &= \+z^\transp\widetilde{\+A}\+z
        - 2 \+r^\transp\widetilde{\+A}\+z
        + \+r^\transp\widetilde{\+A}\+r
        + \widetilde{\+b}^\transp\+z
        - \widetilde{\+b}^\transp\+r
        + \widetilde{c}
    \\
    &= \+z^\transp\widetilde{\+A}\+z
        + (\widetilde{\+b} - 2 \widetilde{\+A}\+r)^\transp\+z
        + \+r^\transp\widetilde{\+A}\+r
        - 2\+r^\transp\widetilde{\+A}\+r
        + \widetilde{c}
    \\
    &= \+z^\transp\widetilde{\+A}\+z
        - \+r^\transp \widetilde{\+A} \widetilde{\+A}^+ \widetilde{\+A} \+r
        + \widetilde{c}
    \\
    &= \+z^\transp\widetilde{\+A}\+z
        - \frac{1}{4}\widetilde{\+b}^\transp \widetilde{\+A}^+ \widetilde{\+b}
        + \widetilde{c}
    \\
    &= \+z^\transp\widetilde{\+A}\+z + d,
    \numberthis
\end{align*}
where $d = \widetilde{c} - \frac{1}{4}\widetilde{\+b}^\transp \widetilde{\+A}^+ \widetilde{\+b}$ and
$\widetilde{\+A}^+$ is the Moore–Penrose inverse of $\widetilde{\+A}$ for which $\widetilde{\+A} \widetilde{\+A}^+ \widetilde{\+A} = \widetilde{\+A}$ in particular.
Let $\widetilde{\+A} = \+Q \+\Lambda \+Q^\transp$ be the spectral decomposition of matrix $\widetilde{\+A}$, where $\+Q$ is an orthogonal matrix and $\+\Lambda$ is a diagonal matrix containing the eigenvalues $\lambda_i, i=1,2,\dots,n$ of matrix $\widetilde{\+A}$. Consider the term $\+z^\transp\widetilde{\+A}\+z$. This can be reformatted to
\begin{gather}
    \+z^\transp\widetilde{\+A}\+z
    = \+z^\transp\+Q \+\Lambda \+Q^\transp\+z
    = \left(\+Q^\transp\+z\right)^\transp \+\Lambda \left(\+Q^\transp\+z\right).
\end{gather}
Let $\+g = \+Q^\transp \+z \sim \mathrm{N}\left(\+\mu_{g}, \+\Sigma_{g}\right)$, where
\begin{align}
    \+\mu_{g} &= \+Q^\transp \E[\+z] =  \+Q^\transp \+r,
\intertext{and}
\+\Sigma_{g} &= \+Q^\transp \Var[\+z] \+Q = \+Q^\transp \+Q = \eye \,.
\end{align}
Now the term $\+z^\transp\widetilde{\+A}\+z$ can be written as a sum of independent scaled non-central $\chi^2$ distributed random variables with degree one:
\begin{gather}
    \+z^\transp\widetilde{\+A}\+z
    = \+g^\transp \+\Lambda \+g
    = \sum_{i \in L_{\neq 0}}^n \lambda_i g_i^2,
\end{gather}
where $L_{\neq 0}$ is the set of indices for which the corresponding eigenvalue $\lambda_i$ is not zero, i.e.\ $L_{\neq 0} = \{i=1,2,\dots,n : \lambda_i \neq 0\}$.
Here, the distribution of each term $g_i, i \in L_{\neq 0}$ can be formulated unambiguously without $\+r$. We have
\begin{gather}
    2\widetilde{\+A}\+r = 2\+Q \+\Lambda \+Q^\transp \+r = \widetilde{\+b}
    \\
    \+\Lambda \+Q^\transp\+r = \frac{1}{2} \+Q^\transp \widetilde{\+b}.
\end{gather}
Now, for $i \in L_{\neq 0}$,
\begin{gather}
    \mu_{g,\,i} = \left[\+Q^\transp \+r\right]_i  = \frac{1}{2 \lambda_i} \left[\+Q^\transp \widetilde{\+b}\right]_i\,.
\end{gather}

\subsection{Moments of the Variables}
\label{app_sec_analytic_moments}

In this section, we present some moments of interest for the given variables of quadratic form on $\varepsilon$.
Let $Z$ denote such a variable:
\begin{align}
Z = \+\varepsilon^\transp \+A \+\varepsilon + \+b^\transp \+\varepsilon + c
\,.
\end{align}

A general form for the moments is presented in Theorem 3.2b3 by~\citet[][p. 54]{mathai_provost_quadratic}. Based on this general form, we formulate the mean, variance, and skewness.
The resulting moments can also be derived by considering the variables as a sum of independent scaled non-central $\chi^2$ distributed random variables as presented in Appendix~\ref{app_sec_analytic_as_sum_of_chi2}.

Let $\+\Sigma_\star^{1/2}\+A\+\Sigma_\star^{1/2} = \+Q \+\Lambda \+Q^\transp$ be the spectral decomposition of matrix $\+\Sigma_\star^{1/2}\+A\+\Sigma_\star^{1/2}$, where $\+Q$ is an orthogonal matrix and $\+\Lambda$ is a diagonal matrix containing the eigenvalues $\lambda_i, i=1,2,\dots,n$ of matrix $\+\Sigma_\star^{1/2}\+A\+\Sigma_\star^{1/2}$.
In particular, for this decomposition it holds that $\left(\+\Sigma_\star^{1/2}\+A\+\Sigma_\star^{1/2}\right)^k = \+Q \+\Lambda^k \+Q^\transp$.
Following the notation in the theorem, we have
\begin{gather}
g_\star^{(k)} = \begin{cases}
    \frac{1}{2}k!\sum_{j=1}^n (2 \lambda_j)^{k+1} + \frac{(k+1)!}{2}\sum_{j=1}^n \+b^{\star\,2}_j (2 \lambda_j)^{k-1}
    & \text{when } k \geq 1,
    \\
\frac{1}{2}\sum_{j=1}^n (2\lambda_j) + c + \+b^\transp \+\mu_\star + \+\mu_\star^\transp \+A \+\mu_\star
    & \text{when } k = 0,
\end{cases}
\end{gather}
where
\begin{gather}
\+b^\star = \+Q^\transp (\+\Sigma_\star^{1/2}\+b + 2 \+\Sigma_\star^{1/2} \+A \+\mu_\star).
\end{gather}
The moments of interest are
\begingroup
\allowdisplaybreaks
\begin{align*}
m_1 &= \E\left[Z\right]
= g_\star^0
\\
&= \sum_{j=1}^n \lambda_j + c + \+b^\transp \+\mu_\star + \+\mu_\star^\transp \+A \+\mu_\star
\\
&= \operatorname{tr}\left(\+\Sigma_\star^{1/2}\+A\+\Sigma_\star^{1/2}\right)
+ c
+ \+b^\transp \+\mu_\star
+ \+\mu_\star^\transp \+A \+\mu_\star
\label{app_eq_analytic_1_moment}
\numberthis
\\
\overline{m}_2 &= \Var\left[Z\right]
= g_\star^1
\\
&= 2 \sum_{j=1}^n \lambda_j^2 + \sum_{j=1}^n b^{\star\,2}_j
\\
&= 2 \operatorname{tr}\left(\left(\+\Sigma_\star^{1/2}\+A\+\Sigma_\star^{1/2}\right)^2\right) + \+b^{\star\,\transp} \+b^\star
\\
&= 2 \operatorname{tr}\left(\left(\+\Sigma_\star^{1/2}\+A\+\Sigma_\star^{1/2}\right)^2\right)
    + (\+\Sigma_\star^{1/2}\+b + 2 \+\Sigma_\star^{1/2} \+A \+\mu_\star)^\transp \+Q \+Q^\transp(\+\Sigma_\star^{1/2}\+b + 2 \+\Sigma_\star^{1/2} \+A \+\mu_\star)
\\
&= 2 \operatorname{tr}\left(\left(\+\Sigma_\star^{1/2}\+A\+\Sigma_\star^{1/2}\right)^2\right)
    + \+b^\transp \+\Sigma_\star \+b + 4 \+b^\transp \+\Sigma_\star \+A \+\mu_\star + 4 \+\mu_\star^\transp \+A \+\Sigma_\star \+A \+\mu_\star
\label{app_eq_analytic_2_moment}
\numberthis
\\
\overline{m}_3 &= \E\Big[\left(Z-\E\left[Z\right]\right)^3\Big]
= g_\star^1
\\
&= 8 \sum_{j=1}^n \lambda_j^3 + 6 \sum_{j=1}^n \+b^{\star\,2}_j \lambda_j
\\
&= 8 \operatorname{tr}\left(\left(\+\Sigma_\star^{1/2}\+A\+\Sigma_\star^{1/2}\right)^3\right)
    + 6 \+b^{\star\,\transp} \+\Lambda \+b^\star
\\
&= 8 \operatorname{tr}\left(\left(\+\Sigma_\star^{1/2}\+A\+\Sigma_\star^{1/2}\right)^3\right)
    + 6 (\+\Sigma_\star^{1/2}\+b + 2 \+\Sigma_\star^{1/2} \+A \+\mu_\star)^\transp \underbrace{\+Q \+\Lambda \+Q^\transp}_{=\+\Sigma_\star^{1/2}\+A\+\Sigma_\star^{1/2}}(\+\Sigma_\star^{1/2}\+b + 2 \+\Sigma_\star^{1/2} \+A \+\mu_\star)
\\
&= 8 \operatorname{tr}\left(\left(\+\Sigma_\star^{1/2}\+A\+\Sigma_\star^{1/2}\right)^3\right)
    + 6 \+b^\transp \+\Sigma_\star \+A \+\Sigma_\star \+b + 24 \+b^\transp \+\Sigma_\star \+A \+\Sigma_\star \+A \+\mu_\star + 24 \+\mu_\star^\transp \+A \+\Sigma_\star \+A \+\Sigma_\star \+A \+\mu_\star
\label{app_eq_analytic_3_moment}
\numberthis
\\
\widetilde{m}_3 &= \E\Big[\left(Z-\E\left[Z\right]\right)^3\Big] \Big/ \Big(\Var\left[Z\right]\Big)^{3/2}
= \overline{m}_3 \Big/ (\overline{m}_2)^{3/2}
\label{app_eq_analytic_3_moment_skew}
\,.
\numberthis
\end{align*}
\endgroup


\subsubsection{Effect of the Model Variance}
\label{app_sec_analytic_effect_of_model_variance}

We consider the effect of the model variance parameter $\tau$ to the moments defined in Appendix~\ref{app_sec_analytic_moments} for the error $\elpdHatErrC{\Mfd}{\y}$.
From the equations~\eqref{app_eq_analytic_1_moment}--\eqref{app_eq_analytic_3_moment} it can be directly seen that
\begin{align}
m_1 &= C_1 \tau^{-2} + C_2
\\
\overline{m}_2 &= C_3 \tau^{-4}
\label{app_eq_m_2_as_func_of_tau}
\\
\overline{m}_3 &= C_4 \tau^{-6}
\label{app_eq_m_3_as_func_of_tau}
\,,
\end{align}
where each $C_i$ denotes a different constant.
Furthermore, it follows from equations~\eqref{app_eq_m_2_as_func_of_tau} and~\eqref{app_eq_m_3_as_func_of_tau} that the skewness $\widetilde{m}_3 = \overline{m}_3 \Big/ (\overline{m}_2)^{3/2}$ does not depend on $\tau$.


\subsubsection{Effect of the Non-Shared Covariates' Effects}
\label{app_sec_analytic_effect_of_model_difference}

We further consider the moments defined in Appendix~\ref{app_sec_analytic_moments} for the error $\elpdHatErrC{\Mfd}{\y}$ when the difference of the models' performances grows via the difference in the effects of the non-shared covariates.
Let $\+\beta_\Delta$ denote the vector of effects of the non-shared covariates that are included either in model $\Mfa$ or $\Mfb$ but not in both of them, let $\+\beta_{-\Mfsa-\Mfsb}$ denote the vector of effects missing in both models and let $\+\beta_{\Msa-\Msb}$ for $(\Ma, \Mb) \in \{(\Mfa, \Mfb), (\Mfb, \Mfa)\}$ denote the vector of effects included in model $\Ma$ but not in $\Mb$.
Furthermore, let $\+X_{[\bcdot,\Delta]}$, $\+X_{[\bcdot,-\Mfsa-\Mfsb]}$, and $\+X_{[\bcdot,\Msa-\Msb]}$ denote the respective data.
In the following, we analyse the moments when the difference of the models is increased by increasing the magnitude in $\+\beta_\Delta$.
Consider a scaling of this vector $\+\beta_\Delta = \beta_r \+\beta_\text{rate} + \+\beta_\text{base}$, where $\beta_r$ is a scalar scaling factor and $\+\beta_\text{rate} \neq 0,  \+\beta_\text{base}$ are some effect growing rate vector and base effect vector respectively.
In the following, we consider the moments of interest as a function of $\beta_r$.

The matrix $\+A_\mathrm{err}$ does not depend on $\+\beta$ and is thus constant with respect to $\beta_r$.
The vector $\hat{\+y}_{-\Mfsa}$, involved in the formulation of the moments, can be expressed as
\begin{align*}
    \hat{\+y}_{-\Msa} &= \+X_{[\bcdot,-\Msa]} \+\beta_{-\Msa}
    \\
    &=\+X_{[\bcdot,\Msb-\Msa]} \+\beta_{\Msb-\Msa} + \+X_{[\bcdot,-\Msa-\Msb]}\+\beta_{-\Msa-\Msb}
    \\
    &\eqqcolon \hat{\+y}_{\Msb-\Msa} + \hat{\+y}_{-\Msa-\Msb}
     \numberthis
\end{align*}
 for $(\Ma, \Mb) \in \{(\Mfa, \Mfb), (\Mfb, \Mfa)\}$.
By utilising this, vector $\+b_\mathrm{err}$ defined in Equation~\eqref{app_eq_error_param_b_err} can be expressed as
\begin{align*}
\+b_\mathrm{err}
    &= \frac{1}{\tau^2} \Big(
        \+B_{\mathrm{err},\Mfsa,1} \hat{\+y}_{-\Mfsa}
        - \+B_{\mathrm{err},\Mfsb,1} \hat{\+y}_{-\Mfsb}
        - \+B_{\Mfsd,2} \+\mu_\star
   \Big)
    \\
    &= \frac{1}{\tau^2} \Bigl(
         \+B_{\mathrm{err},\Mfsa,1}  \hat{\+y}_{\Mfsb-\Mfsa} - \+B_{\mathrm{err},\Mfsb,1}  \hat{\+y}_{\Mfsa-\Mfsb}
        + \left(  \+B_{\mathrm{err},\Mfsa,1} - \+B_{\mathrm{err},\Mfsb,1} \right) \hat{\+y}_{-\Mfsa-\Mfsb}
        - \+B_{\Mfsd,2} \+\mu_\star
    \Bigr)
   \\
   &= \beta_r \+q_{\+b_\mathrm{err},1} + \+q_{\+b_\mathrm{err},0}
    \,, \numberthis
\intertext{where}
    \+q_{\+b_\mathrm{err},1} &= \frac{1}{\tau^2} \left(
        \+B_{\mathrm{err},\Mfsa,1} \+X_{[\bcdot,\Mfsb-\Mfsa]} \+\beta_{\text{rate},\Mfsb-\Mfsa}
        - \+B_{\mathrm{err},\Mfsb,1} \+X_{[\bcdot,\Mfsa-\Mfsb]} \+\beta_{\text{rate},\Mfsa-\Mfsb}
    \right)
\numberthis
\intertext{and}
    \+q_{\+b_\mathrm{err},0} &= \frac{1}{\tau^2} \Big(
        \+B_{\mathrm{err},\Mfsa,1} \+X_{[\bcdot,\Mfsb-\Mfsa]} \+\beta_{\text{base},\Mfsb-\Mfsa}
        - \+B_{\mathrm{err},\Mfsb,1} \+X_{[\bcdot,\Mfsa-\Mfsb]} \+\beta_{\text{base},\Mfsa-\Mfsb}
        \\ &\quad\qquad
        + \left(  \+B_{\mathrm{err},\Mfsa,1} - \+B_{\mathrm{err},\Mfsb,1} \right) \hat{\+y}_{-\Mfsa-\Mfsb}
        - \+B_{\Mfsd,2} \+\mu_\star
    \Big)
\,. \numberthis
\end{align*}
Scalar $c_\mathrm{err}$ defined in Equation~\eqref{app_eq_error_param_c_err} can be expressed as
\begin{align*}
c_\mathrm{err}
    &= \frac{1}{\tau^2} \Bigg(
        \hat{\+y}_{-\Mfsa}^\transp \+C_{\mathrm{err},\Mfsa,1} \hat{\+y}_{-\Mfsa}
        - \hat{\+y}_{-\Mfsb}^\transp \+C_{\mathrm{err},\Mfsb,1} \hat{\+y}_{-\Mfsb}
        \\&\qquad\quad
        -  \hat{\+y}_{-\Mfsa}^\transp \+C_{\Mfsa,2} \+\mu_{\star}
        +  \hat{\+y}_{-\Mfsb}^\transp \+C_{\Mfsb,2} \+\mu_{\star}
        \\&\qquad\quad
        - \+\mu_{\star}^\transp \+C_{\Mfsd,3} \+\mu_{\star}
        - \+\sigma_\star^\transp \+C_{\Mfsd,3} \+\sigma_\star
    \Bigg)
     + c_{\mathrm{err},4}
   \\
   &= \beta_r^2 q_{c_\mathrm{err},2} + \beta_r q_{c_\mathrm{err},1} + C_2
    \,, \numberthis
\intertext{where}
    q_{c_\mathrm{err},2} &=
        \frac{1}{\tau^2} \Big(
            \+\beta_{\text{rate},\Mfsb-\Mfsa}^\transp \+X_{[\bcdot,\Mfsb-\Mfsa]}^\transp \+C_{\mathrm{err},\Mfsa,1} \+X_{[\bcdot,\Mfsb-\Mfsa]} \+\beta_{\text{rate},\Mfsb-\Mfsa}
        \\ &\quad\qquad
            - \+\beta_{\text{rate},\Mfsa-\Mfsb}^\transp \+X_{[\bcdot,\Mfsa-\Mfsb]}^\transp \+C_{\mathrm{err},\Mfsb,1} \+X_{[\bcdot,\Mfsa-\Mfsb]} \+\beta_{\text{rate},\Mfsa-\Mfsb}
        \Big)
\,, \numberthis \\
    q_{c_\mathrm{err},1} &= \frac{1}{\tau^2} \bigg(
        \Big(
            2 \+\beta_{\text{base},\Mfsb-\Mfsa}^\transp \+X_{[\bcdot,\Mfsb-\Mfsa]}^\transp \+C_{\mathrm{err},\Mfsa,1}
            + 2 \hat{\+y}_{-\Mfsa-\Mfsb}^\transp \+C_{\mathrm{err},\Mfsa,1}
            - \+\mu_{\star}^\transp \+C_{\Mfsa,2}
        \Big) \+X_{[\bcdot,\Mfsb-\Mfsa]} \+\beta_{\text{rate},\Mfsb-\Mfsa}
        \\ &\quad\qquad
        - \Big(
            2 \+\beta_{\text{base},\Mfsa-\Mfsb}^\transp \+X_{[\bcdot,\Mfsa-\Mfsb]}^\transp \+C_{\mathrm{err},\Mfsb,1}
            + 2 \hat{\+y}_{-\Mfsb-\Mfsa}^\transp \+C_{\mathrm{err},\Mfsb,1}
            - \+\mu_{\star}^\transp \+C_{\Mfsb,2}
        \Big) \+X_{[\bcdot,\Mfsa-\Mfsb]} \+\beta_{\text{rate},\Mfsa-\Mfsb}
    \bigg)
\,, \numberthis \\
    q_{c_\mathrm{err},0} &= \frac{1}{\tau^2} \bigg(
        \left( \+X_{[\bcdot,\Mfsb-\Mfsa]} \+\beta_{\text{base},\Mfsb-\Mfsa} + \hat{\+y}_{-\Mfsa-\Mfsb} \right)^\transp \+C_{\mathrm{err},\Mfsa,1} \left( \+X_{[\bcdot,\Mfsb-\Mfsa]} \+\beta_{\text{base},\Mfsb-\Mfsa}  + \hat{\+y}_{-\Mfsa-\Mfsb} \right)
        \\ & \quad\qquad
        - \left( \+X_{[\bcdot,\Mfsa-\Mfsb]} \+\beta_{\text{base},\Mfsa-\Mfsb} + \hat{\+y}_{-\Mfsa-\Mfsb} \right)^\transp \+C_{\mathrm{err},\Mfsb,1} \left( \+X_{[\bcdot,\Mfsa-\Mfsb]} \+\beta_{\text{base},\Mfsa-\Mfsb} + \hat{\+y}_{-\Mfsa-\Mfsb} \right)
        \\ & \quad\qquad
        - \+\mu_{\star}^\transp\left(
            \+C_{\Mfsa,2} \+X_{[\bcdot,\Mfsb-\Mfsa]} \+\beta_{\text{base},\Mfsb-\Mfsa} - \+C_{\Mfsb,2} \+X_{[\bcdot,\Mfsa-\Mfsb]} \+\beta_{\text{base},\Mfsa-\Mfsb}
        \right)
        \\ & \quad\qquad
        - \+\mu_{\star}^\transp \+C_{\Mfsd,3} \+\mu_{\star}
        - \+\sigma_\star^\transp \+C_{\Mfsd,3} \+\sigma_\star
    \bigg) + c_{\mathrm{err},4}
\,. \numberthis
\end{align*}
From this it follows, that $m_1$, $\overline{m}_2$, and $\overline{m}_3$ presented in equations~\eqref{app_eq_analytic_1_moment}--\eqref{app_eq_analytic_3_moment} respectively are all of second degree as a function of $\beta_r$. Thus, the skewness
\begin{align}
\lim_{\beta_r \rightarrow \pm \infty} \widetilde{m}_3 = \lim_{\beta_r \rightarrow \pm \infty} \frac{\overline{m}_3}{(\overline{m}_2)^{3/2}} = 0 \,.
\end{align}

When $\+\beta_{\text{base}}=0$, there are no outliers in the data, and each covariate is included in either one of the models, we can further draw some conclusions when $|\beta_r|$ gets smaller so that the models gets closer in predictive performance.
In this situation $\+q_{\+b_\mathrm{err},0} = 0$ and the moments $\overline{m}_2$ and $\overline{m}_3$ have the following forms
\begin{align}
    \overline{m}_2 &= C_{2,2} \beta_r^2 + C_{2,0} \\
    \overline{m}_3 &= C_{3,2} \beta_r^2 + C_{3,0} \,,
\intertext{where}
    C_{2,2} &= \+q_{\+b_\mathrm{err},1}^\transp \+\Sigma_\star \+q_{\+b_\mathrm{err},1}
    \,, \\
    C_{2,0} &= 2 \operatorname{tr}\left(\left(\+\Sigma_\star^{1/2}\+A\+\Sigma_\star^{1/2}\right)^2\right)
    \,, \\
    C_{3,2} &= 6 \+q_{\+b_\mathrm{err},1}^\transp \+\Sigma_\star \+A \+\Sigma_\star \+q_{\+b_\mathrm{err},1}
    \,, \\
    C_{3,0} &= 8 \operatorname{tr}\left(\left(\+\Sigma_\star^{1/2}\+A\+\Sigma_\star^{1/2}\right)^3\right)
    \,.
\end{align}
Because $\Sigma_\star$ is positive definite $C_{2,2} > 0$.
Because trace corresponds to the sum of eigenvalues and eigenvalues of the second power of a matrix equal to the squared eigenvalues of the original, trace of a matrix to the second power is non-negative and here $C_{2,0} > 0$.
The skewness $\widetilde{m}_3$ continuous and symmetric with regards to $\beta_r$ and
\begin{gather*}
    \frac{\diff}{\diff \beta_r} \widetilde{m}_3 = \frac{\diff}{\diff \beta_r} \frac{C_{2,2} \beta_r^2 + C_{2,0}}{(C_{2,2} \beta_r^2 + C_{2,0})^{3/2}}
    = \frac{\beta_r \left(-C_{2,2}C_{3,2}\beta_r^2 +2C_{3,2}C_{2,0} -3 C_{2,2} C_{3,0} \right)}{(C_{2,2} \beta_r^2 + C_{2,0})^{5/2}}
\,. \numberthis
\end{gather*}
Solving for zero yields
\begin{align*}
    \beta_r &= 0
    \label{app_eq__skewness_beta_r_to_zero_root0}
\numberthis \\
\intertext{and if $2\frac{C_{2,0}}{C_{2,2}} - 3 \frac{C_{3,0}}{C_{3,2}} > 0$}
    \beta_r &= \pm \sqrt{2\frac{C_{2,0}}{C_{2,2}} - 3 \frac{C_{3,0}}{C_{3,2}}}
    \label{app_eq__skewness_beta_r_to_zero_rootpm}
\numberthis \,.
\end{align*}
From this it follows that the absolute skewness $|\widetilde{m}_3|$ has a maximum either at~\eqref{app_eq__skewness_beta_r_to_zero_root0} or at~\eqref{app_eq__skewness_beta_r_to_zero_rootpm} or in all of them.


\subsubsection{Effect of Outliers}
\label{app_sec_analytic_effect_of_outliers}

We consider the effect of outliers through parameter $\+\mu_\star$ to the moments defined in Appendix~\ref{app_sec_analytic_moments} for the error $\elpdHatErrC{\Mfd}{\y}$.
The effect of $\+\mu_\star$ depends on the explanatory variable $\+X$ and the covariate effect vector $\+\beta$.
Let us restate the moments $m_1$, $\overline{m}_2$, and $\overline{m}_3$ as a quadratic form on $\+\mu_\star$:
\begin{align}
m_1 &= \+\mu_\star^\transp \+Q_{m_1} \+\mu_\star + \+q_{m_1}^\transp \+\mu_\star + C_1
\,, \\
\overline{m}_2 &= \+\mu_\star^\transp \+Q_{\overline{m}_2} \+\mu_\star + \+q_{\overline{m}_2}^\transp \+\mu_\star + C_2
\,, \\
\overline{m}_3 &= \+\mu_\star^\transp \+Q_{\overline{m}_3} \+\mu_\star + \+q_{\overline{m}_3}^\transp \+\mu_\star + C_3
\,,
\end{align}
where
\begin{align}
\+Q_{m_1} &= \frac{1}{\tau^2} \left( \+A_{\mathrm{err},1} - \+B_{\Mfsd,2} - \+C_{\Mfsd,3} \right)
\,,\\
\+q_{m_1} &= \frac{1}{\tau^2} \left(
       \left( \+B_{\mathrm{err},\Mfsa,1} - \+C_{\Mfsa,2} \right)  \hat{\+y}_{-\Mfsa}
        - \left( \+B_{\mathrm{err},\Mfsb,1} - \+C_{\Mfsb,2} \right) \hat{\+y}_{-\Mfsb}
    \right)
\,,\\
\+Q_{\overline{m}_2} &= \frac{1}{\tau^4} \left( 2 \+A_{\mathrm{err},1} - \+B_{\Mfsd,2} \right)^\transp \+\Sigma_\star \left( 2 \+A_{\mathrm{err},1} - \+B_{\Mfsd,2} \right)
\,,\\
\+q_{\overline{m}_2} &= \frac{2}{\tau^4}
       \left( 2 \+A_{\mathrm{err},1} - \+B_{\Mfsd,2} \right)^\transp \+\Sigma_\star
       \left( \+B_{\mathrm{err},\Mfsa,1} \hat{\+y}_{-\Mfsa} - \+B_{\mathrm{err},\Mfsb,1} \hat{\+y}_{-\Mfsb} \right)
\,,\\
\+Q_{\overline{m}_3} &= \frac{6}{\tau^6} \left( 2 \+A_{\mathrm{err},1} - \+B_{\Mfsd,2} \right)^\transp \+\Sigma_\star \+A_{\mathrm{err},1} \+\Sigma_\star \left( 2 \+A_{\mathrm{err},1} - \+B_{\Mfsd,2} \right)
\,,\\
\+q_{\overline{m}_3} &= \frac{12}{\tau^6}
       \left( 2 \+A_{\mathrm{err},1} - \+B_{\Mfsd,2} \right)^\transp \+\Sigma_\star \+A_{\mathrm{err},1} \+\Sigma_\star
       \left( \+B_{\mathrm{err},\Mfsa,1} \hat{\+y}_{-\Mfsa} - \+B_{\mathrm{err},\Mfsb,1} \hat{\+y}_{-\Mfsb} \right)
\,,
\end{align}
and $C_1$, $C_2$, and $C_3$ are some constants.
Consider the moments as a function of a scalar scaling factor $\mu_{\star,r}$, where $\+\mu_\star = \mu_{\star,r} \+\mu_{\star,\text{rate}} + \+\mu_{\star,\text{base}}$, where $\+\mu_{\star,\text{rate}} \neq 0$, and $\+\mu_{\star,\text{base}}$ are some growing rate vector and base vector respectively.
Depending on $\+X$, $\+\beta$, $\+\mu_{\star,\text{rate}}$, and $\+\mu_{\star,\text{base}}$, the first moment $m_1$ can be of first or second degree or constant.
Because $\+x^\transp\+Q_{\overline{m}_3}\+x = 0 \Leftrightarrow \+x^\transp\+q_{\overline{m}_3} = 0 \Leftrightarrow \+x^\transp\+Q_{\overline{m}_2}\+x = 0 \Leftrightarrow \+x^\transp\+q_{\overline{m}_2} = 0, \forall \+x \in \mathbb{R}^n$, moments $\overline{m}_2$ and $\overline{m}_3$ are both either constants or of second degree.
Thus, if not constant, the skewness
\begin{align}
\lim_{\mu_{\star,r} \rightarrow \pm \infty} \widetilde{m}_3 = \lim_{\mu_{\star,r} \rightarrow \pm \infty} \frac{\overline{m}_3}{(\overline{m}_2)^{3/2}} = 0 \,.
\end{align}


\subsubsection{Effect of Residual Variance}
\label{app_sec_analytic_effect_of_data_variance}

Next we analyse the moments defined in Appendix~\ref{app_sec_analytic_moments} for the error $\elpdHatErrC{\Mfd}{\y}$ with respect to the data residual variance $\+\Sigma_\star$ by formulating it as $\+\Sigma_\star = \sigma_\star^2 \eye_{n}$.
Now
\begin{align}
m_1 &= \operatorname{tr}\left(\+A_\mathrm{err}\right) \sigma_\star^4 + C_1
\\
\overline{m}_2 &= 2 \operatorname{tr}\left(\+A_\mathrm{err}^2\right) \sigma_\star^4 + C_2 \sigma_\star^2
\label{app_eq_m_2_as_func_of_sigma_star}
\\
\overline{m}_3 &= 8 \operatorname{tr}\left(\+A_\mathrm{err}^3\right) \sigma_\star^6 + C_3 \sigma_\star^4
\label{app_eq_m_3_as_func_of_sigma_star}
\,,
\end{align}
where each $C_i$ denotes a different constant.
Combining  equations~\eqref{app_eq_m_2_as_func_of_sigma_star} and~\eqref{app_eq_m_3_as_func_of_sigma_star}, we get
\begin{align*}
    \lim_{\sigma_\star \rightarrow \infty} \widetilde{m}_3 &=
    \frac{
        \lim_{\sigma_\star \rightarrow \infty} \sigma_\star^{-6} \overline{m}_3
    }{
        \left(
            \lim_{\sigma_\star \rightarrow \infty} \sigma_\star^{-4} \overline{m}_2
        \right)^{3/2}
    }
    = \frac{
        8 \operatorname{tr}\left(\+A_\mathrm{err}^3\right)
    }{
        \left(
            2 \operatorname{tr}\left(\+A_\mathrm{err}^2\right)
        \right)^{3/2}
    }
    = 2^{3/2} \frac{
        \operatorname{tr}\left(\+A_\mathrm{err}^3\right)
    }{
        \operatorname{tr}\left(\+A_\mathrm{err}^2\right)^{3/2}
    }
    \,, \numberthis
\end{align*}
that is, the skewness converges into a constant determined by the explanatory variable matrix $\+X$ when the data variance grows.


\subsubsection{Graphical Illustration of the Moments for an Example Case}
\label{app_sec_analytic_graph}

The behaviour of the moments of the estimator $\elpdHatC{\Mfd}{\y}$, the estimand, $\elpdC{\Mfd}{\y}$, and the error $\elpdHatErrC{\Mfd}{\y}$ for an example problem setting are illustrated in Figure~\ref{fig_analytic_zscore_skew_n_b}. Figure~\ref{fig_analytic_zscore_skew_n_b_tot} illustrates the same problem unconditional on the design matrix $\+X$, so that the design matrix is also random in the data generating mechanism.
The total mean, variance, and skewness are estimated from the simulated $\+X$s, and the resulting uncertainty is estimated using Bayesian bootstrap.
The example case has an intercept and two covariates.
Model \Mb{} ignores one covariate with true effect $\beta_\Delta$ while model \Mb{} considers them all. Here $\+\mu_\star = 0$ so that no outliers are present in the data. The data residual variance is fixed at $\Sigma_\star = \eye{n}$. The model variance is also fixed at $\tau=1$. The illustrated moments of interest are the mean relative to the standard deviation, $m_1 \big/ \sqrt{\overline{m}_2}$, and the skewness $\widetilde{m}_3 = \overline{m}_3 \Big/ (\overline{m}_2)^{3/2}$.
When compared to the analysis with conditional to $\+X$ in Section~\ref{sec_experiments}, the most notable difference can be observed in the behaviour of $\elpdC{\Mfd}{\y}$; with conditionalised design matrix $\+X$, the skewness is high with all effects $\beta_\Delta$, whereas with unconditionalised $\+X$, the skewness decreases when $\beta_\Delta$ grows.

\begin{figure}[t]
\centering
\includegraphics[width=0.8\figurecontrolwidth]{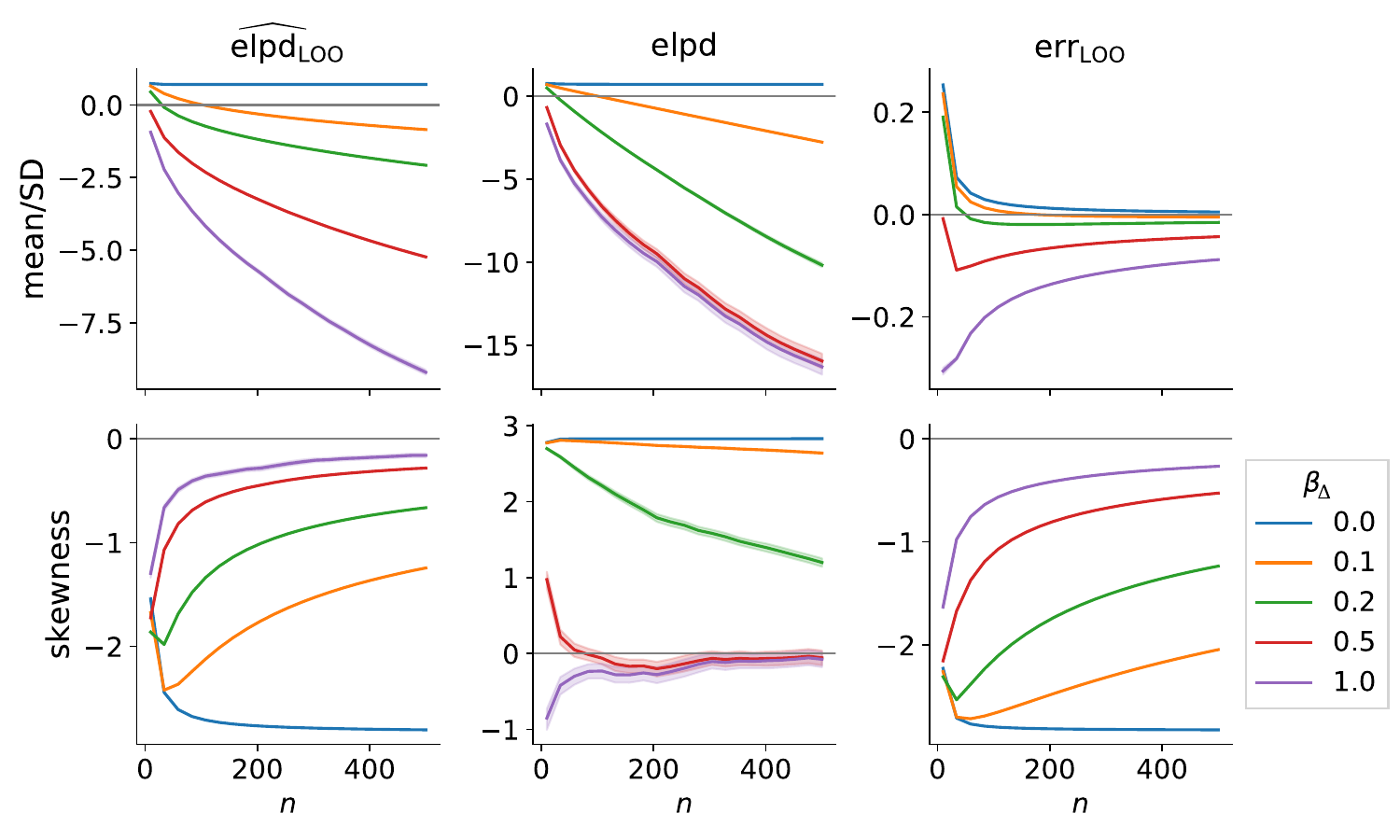}
\caption{%
Illustration of the mean relative to the standard deviation and skewness for $\elpdHatC{\Mfd}{\y}$, $\elpdC{\Mfd}{\y}$, and for the error $\elpdHatErrC{\Mfd}{\y}$ as a function of the data size $n$.
The data consists of an intercept and two covariates following standard normal distribution. One of the covariates with true effect $\beta_\Delta$ is considered only in model $\Mb$.
The solid lines correspond to the median over a Bayesian bootstrap sample of size 2000 from 2000 simulated $\+X$s.
Although wide enough to be visible only in some lines in the middle column, a shaded area around the lines illustrates the 95 \% confidence interval.
}
\label{fig_analytic_zscore_skew_n_b_tot}
\end{figure}


\subsection{One Covariate Case}
\label{app_sec_one_cov}

Let us inspect the behaviour of the moments $m_1$, $\overline{m}_{2}$, and $\tilde{m}_{3}$ of the LOO-CV error formulated in Appendix~\ref{app_sec_analytic_error} in a nested example case, where a null model is compared to a model with one covariate.
Consider that $n$ is even, $n \geq 4$, and $d=2$ so that $\+X$ is two-dimensional.
One column in $\+X$ corresponds to the intercept, being full of $1$s, and the other column corresponds to the covariate, consisting of half $1$s and $-1$s in any order.
Model \Mfa{} only considers the intercept column, and model \Mfb{} considers both the intercept and the sole covariate column.

In addition, we set the data generating mechanism parameters $\+\Sigma_\star$ and $\mu_{\star,\,i}$ to the following form, in which the observations are independent. There is one outlier observation with some index $i_\text{out}$ for which $x_{i_\text{out}} = 1$:
\begin{align}
    \+\Sigma_\star &= s_\star^2 \, \eye_{n},
    \\
    \mu_{\star,\,i} &= \begin{cases} m_\star & \text{when } i = i_\text{out}\,, \\ 0 & \text{otherwise}. \end{cases}
\end{align}
Let $\ones_{n}$ and $0_{n}$ denote a vector of ones and zeroes of length $n$, respectively.
Let vector $\+x \in \mathbb{R}^n$ denote the covariate column in $\+X$. Considering the half $1$s half $-1$s structure of $\+x$ yields
\begin{align}
    x_i^2 &= 1,
    \\
    \+x^\transp\+x &= n,
    \\
    \ones_{n}^\transp\+x &= 0,
    \\
    \+x_{-i}^\transp\+x_{-i} &= n-1,
    \\
    \ones_{n-1}^\transp\+x_{-i} &= - x_i,
    \\
    (x_i x_j + 1)^2 &= 2(x_i x_j + 1),
    \\
    \operatorname{diag}(\+x \+x^\transp) &= \ones_{n},
\end{align}
for all $i,j=1,2,\dots,n$.
Let $\beta_1$ denote the true covariate effect in vector $\+\beta$. As can be seen from the equations of the parameters of the LOO-CV error, the effects in $\+\beta$, which both models consider, do not affect the outcome. In this problem setting, the intercept coefficient is one such an effect.
The parameters $\hat{\+y}_{-\Mfsa}$ and $\hat{\+y}_{-\Mfsb}$ defined in equation~\eqref{app_eq_yhat_def}, which are involved in the formulation of the LOO-CV error, simplifies to
\begin{align}
    \hat{\+y}_{-\Mfsa} &= \+X_{[\bcdot,-\Mfsa]} \+\beta_{-\Mfsa} = \beta_1 \+x
    \\
    \hat{\+y}_{-\Mfsb} &= \+X_{[\bcdot,-\Mfsb]} \+\beta_{-\Mfsb} = 0_{n}
\,.
\end{align}


\subsubsection{Elpd}

In this section, we derive a simplified analytic form for $\elpdC{\Mfd}{\y}$ presented in Appendix~\ref{app_sec_analytic_elpd_d} and for some moments of interest in the one covariate case defined in Appendix~\ref{app_sec_one_cov}. First, we derive the parameters $\+A_\Mfsd$, $\+b_\Mfsd$, and $c_\Mfsd$ defined in Appendix~\ref{app_sec_analytic_elpd_d} and then we use them to derive the respective moments of interest defined in Appendix~\ref{app_sec_analytic_moments}.

\paragraph{Parameters}
\label{app_sec_one_cov_elpd_params}

Following the notation in Appendix~\ref{app_sec_analytic_elpd}, in the one covariate case defined in Appendix~\ref{app_sec_one_cov}, let us find simplified form for the matrices $\+P_\Msk$, $\+D_\Msk$, and for the required products for the LOO-CV error parameters
\begin{align*}
&\+P_\Msk \+D_\Msk, \\
&\+P_\Msk \+D_\Msk \+P_\Msk, \\
&\+P_\Msk \+D_\Msk \left( \+P_\Msk - \eye \right), \\
&\left( \+P_\Msk - \eye \right) \+D_\Msk \left( \+P_\Msk - \eye \right), \\
&\left( \+P_\Msk - \eye \right) \+D_\Msk, \\
&\+D_\Mfsa - \+D_\Mfsb
\numberthis
\end{align*}
for $\Mk \in \{ \Mfsa, \Mfsb\}$, presented in Appendix~\ref{app_sec_analytic_error}.
For the model \Mfa{} we have
\begin{align*}
\+P_\Mfsa
    &= \+X_{[\bcdot,\Mfsa]}
        \left(\+X_{[\bcdot,\Mfsa]}^\transp\+X_{[\bcdot,\Mfsa]}\right)^{-1}
        \+X_{[\bcdot,\Mfsa]}^\transp
\\
    &= \ones_{n} \left(\ones_{n}^\transp \ones_{n}\right)^{-1} \ones_{n}^\transp
\\
    &= \frac{1}{n} \ones_{n} \ones_{n}^\transp
    \numberthis
\end{align*}
and
\begin{align*}
\+D_\Mfsa
    &= \left(\left(\+P_{\Mfsa} \odot \eye_{n}\right) + \eye_{n}\right)^{-1}
\\
    &= \frac{n}{n+1}\eye_{n}
\,.
\numberthis
\end{align*}
Now we get
\begin{align*}
\+P_\Mfsa \+D_\Mfsa
    &= \frac{n}{n+1} \frac{1}{n} \ones_{n} \ones_{n}^\transp \eye_{n}
    \\
    &= \frac{1}{n+1} \ones_{n} \ones_{n}^\transp
\,,
\numberthis
\\
\+P_\Mfsa \+D_\Mfsa \+P_\Mfsa
    &=  \frac{n}{n+1} \underbrace{\+P_\Mfsa \eye_{n} \+P_\Mfsa}_{=\+P_\Mfsa}
    \\
    &= \frac{1}{n+1} \ones_{n} \ones_{n}^\transp
\,,
\numberthis
\\
\+P_\Mfsa \+D_\Mfsa \left( \+P_\Mfsa - \eye \right)
    &= \+P_\Mfsa \+D_\Mfsa \+P_\Mfsa - \+P_\Mfsa \+D_\Mfsa
    \\
    &= 0
\,, \numberthis \\
\left( \+P_\Mfsa - \eye \right) \+D_\Mfsa \left( \+P_\Mfsa - \eye \right)
    &= \+P_\Mfsa \+D_\Mfsa \+P_\Mfsa - \+P_\Mfsa \+D_\Mfsa - \+D_\Mfsa \+P_\Mfsa + \+D_\Mfsa
    \\
    &= \frac{n}{n+1}\eye_{n} - \frac{1}{n+1} \ones_{n} \ones_{n}^\transp
\,, \numberthis \\
\left( \+P_\Mfsa - \eye \right) \+D_\Mfsa
    &= \+D_\Mfsa \+P_\Mfsa - \+D_\Mfsa
    \\
    &= - \frac{n}{n+1}\eye_{n} + \frac{1}{n+1} \ones_{n} \ones_{n}^\transp
\,.
\numberthis
\end{align*}
For model \Mfb{} we have
\begin{align*}
\+P_\Mfsb
&= \+X_{[\bcdot,\Mfsb]}
    \left((\+X_{[\bcdot,\Mfsb]}^\transp\+X_{[\bcdot,\Mfsb]}\right)^{-1}
    \+X_{[\bcdot,\Mfsb]}^\transp,
\\
&=  \begin{bmatrix} \ones_n & \+x \end{bmatrix}
    \left(
        \begin{bmatrix} \ones_{n} & \+x \end{bmatrix}^\transp
        \begin{bmatrix} \ones_{n} & \+x \end{bmatrix}
    \right)^{-1}
    \begin{bmatrix} \ones_n & \+x \end{bmatrix}^\transp
\\
&=  \begin{bmatrix} \ones_n & \+x \end{bmatrix}
    \begin{bmatrix}
        \ones_{n}^\transp\ones_{n} & \ones_{n}^\transp\+x \\
        \ones_{n}^\transp\+x & \+x^\transp\+x
    \end{bmatrix}^{-1}
    \begin{bmatrix} \ones_n & \+x \end{bmatrix}^\transp
\\
&=  \begin{bmatrix} \ones_n & \+x \end{bmatrix}
    \begin{bmatrix}
        n & 0 \\
        0 & n
    \end{bmatrix}^{-1}
    \begin{bmatrix} \ones_n & \+x \end{bmatrix}^\transp
\\
&=  \frac{1}{n^2}\begin{bmatrix} \ones_n & \+x \end{bmatrix}
    \begin{bmatrix}
        n & 0 \\
        0 & n
    \end{bmatrix}
    \begin{bmatrix} \ones_n & \+x \end{bmatrix}^\transp
\\
&= \frac{1}{n} \left(\ones_{n} \ones_{n}^\transp + \+x \+x^\transp \right)
\numberthis
\end{align*}
and
\begin{align*}
\+D_\Mfsb
    &= \left(\left(\+P_{\Mfsb} \odot \eye_{n}\right) + \eye_{n}\right)^{-1}
\\
    &= \left(\frac{1}{n} \left(\eye_{n} + \eye_{n} \right) + \eye_{n}\right)^{-1}
\\
    &= \frac{n}{n+2}\eye_{n}
\,.
\numberthis
\end{align*}
Now we get
\begin{align*}
\+P_\Mfsb \+D_\Mfsb
    &= \frac{1}{n} \frac{n}{n+2} \left(\ones_{n} \ones_{n}^\transp + \+x \+x^\transp \right) \eye_{n}
    \\
    &= \frac{1}{n+2} \left(\ones_{n} \ones_{n}^\transp + \+x \+x^\transp \right)
\,,
\numberthis
\\
\+P_\Mfsb \+D_\Mfsb \+P_\Mfsb
    &=  \frac{n}{n+2} \underbrace{\+P_\Mfsb \eye_{n} \+P_\Mfsb}_{=\+P_\Mfsb}
    \\
    &= \frac{1}{n+2} \left(\ones_{n} \ones_{n}^\transp + \+x \+x^\transp \right)
\,,
\numberthis
\\
\+P_\Mfsb \+D_\Mfsb \left( \+P_\Mfsb - \eye \right)
    &= \+P_\Mfsb \+D_\Mfsb \+P_\Mfsb - \+P_\Mfsb \+D_\Mfsb
    \\
    &= 0
\,, \numberthis \\
\left( \+P_\Mfsb - \eye \right) \+D_\Mfsb \left( \+P_\Mfsb - \eye \right)
    &= \+P_\Mfsb \+D_\Mfsb \+P_\Mfsb - \+P_\Mfsb \+D_\Mfsb - \+D_\Mfsb \+P_\Mfsb + \+D_\Mfsb
    \\
    &= \frac{n}{n+2}\eye_{n} - \frac{1}{n+2} \left(\ones_{n} \ones_{n}^\transp + \+x \+x^\transp \right)
\,, \numberthis \\
\left( \+P_\Mfsb - \eye \right) \+D_\Mfsb
    &= \+D_\Mfsb \+P_\Mfsb - \+D_\Mfsb
    \\
    &= - \frac{n}{n+2}\eye_{n} + \frac{1}{n+2} \left(\ones_{n} \ones_{n}^\transp + \+x \+x^\transp \right)
\,.
\numberthis
\end{align*}
Furthermore, we get
\begin{align}
    \+D_\Mfsa - \+D_\Mfsb &= \frac{n}{n+1}\eye_{n} - \frac{n}{n+2}\eye_{n} = \frac{n}{(n+1)(n+2)} \eye_{n}
    \,.
\end{align}
Moreover, we get
\begin{align*}
\+B_{\Mfsa,1} &= - \+P_{\Mfsa} \+D_{\Mfsa} \left(\+P_{\Mfsa} - \eye \right)
\\
    &= 0
    \numberthis
\,,
\\
\+B_{\Mfsb,1} &= - \+P_{\Mfsb} \+D_{\Mfsb} \left(\+P_{\Mfsb} - \eye \right)
\\
    &= 0
    \numberthis
\,,
\\
\+C_{\Mfsa,1} &= - \frac{1}{2} \left(\+P_{\Mfsa} - \eye \right) \+D_{\Mfsa} \left(\+P_{\Mfsa} - \eye \right)
\\
    &= -\frac{n}{2(n+1)}\eye_{n} + \frac{1}{2(n+1)} \ones_{n} \ones_{n}^\transp
\numberthis
\,,
\\
\+C_{\Mfsb,1} &= - \frac{1}{2} \left(\+P_{\Mfsb} - \eye \right) \+D_{\Mfsb} \left(\+P_{\Mfsb} - \eye \right)
\\
    &= -\frac{n}{2(n+2)} \eye_{n} + \frac{1}{2(n+2)} \left(\ones_{n} \ones_{n}^\transp + \+x \+x^\transp \right)
\numberthis
\,,
\\
\+C_{\Mfsa,2} &= \left(\+P_{\Mfsa} - \eye \right) \+D_{\Mfsa}
\\
    &= - \frac{n}{n+1}\eye_{n} + \frac{1}{n+1} \ones_{n} \ones_{n}^\transp
\numberthis
\\
\+C_{\Mfsb,2} &= \left(\+P_{\Mfsb} - \eye \right) \+D_{\Mfsb}
\\
    &= - \frac{n}{n+2}\eye_{n} + \frac{1}{n+2} \left(\ones_{n} \ones_{n}^\transp + \+x \+x^\transp \right)
\numberthis
\,,
\end{align*}
and

\begin{align*}
\+A_{\Mfsd,1} 
    &= - \frac{1}{2} \left( \+P_{\Mfsa} \+D_{\Mfsa} \+P_{\Mfsa} - \+P_{\Mfsb} \+D_{\Mfsb} \+P_{\Mfsb} \right)
\\
    &= - \frac{1}{2} \left( \frac{1}{n+1} \ones_{n} \ones_{n}^\transp - \frac{1}{n+2} \left(\ones_{n} \ones_{n}^\transp + \+x \+x^\transp \right) \right)
\\
    &= - \frac{1}{2(n+1)(n+2)} \ones_{n} \ones_{n}^\transp + \frac{1}{2(n+2)} \+x \+x^\transp
    \numberthis
\,,
\\
\+B_{\Mfsd,2} &= \+P_{\Mfsa} \+D_{\Mfsa} - \+P_{\Mfsb} \+D_{\Mfsb}
    \\
    &= \frac{1}{n+1} \ones_{n} \ones_{n}^\transp
        - \frac{1}{n+2} \left(\ones_{n} \ones_{n}^\transp + \+x \+x^\transp \right)
    \\
    &= \frac{1}{(n+2)(n+1)} \ones_{n} \ones_{n}^\transp - \frac{1}{n+2} \+x \+x^\transp
\,,
\numberthis
\\
\+C_{\Mfsd,3} &= - \frac{1}{2} \left( \+D_{\Mfsa} - \+D_{\Mfsb}\right)
    \\
    &= - \frac{n}{2(n+1)(n+2)} \eye_{n}
\numberthis
\,,
\\
c_{\Mfsd,4} 
    &= \frac{1}{2} \log \left( \prod_{i=1}^n \frac{\+D_{\Mfsa,[i,i]}}{\+D_{\Mfsb,[i,i]}} \right)
\\
    &= \frac{1}{2} \log \left( \prod_{i=1}^n \frac{\frac{n}{n+1}}{\frac{n}{n+2}} \right)
\\
    &= \frac{n}{2} \log \frac{n+2}{n+1}
\numberthis
\,.
\end{align*}
Now we get
\begin{align*}
\+A_\Mfsd
&= \frac{1}{\tau^2} \+A_{\Mfsd,1}
\\
    &= \frac{1}{\tau^2} \left( - \frac{1}{2(n+1)(n+2)} \ones_{n} \ones_{n}^\transp + \frac{1}{2(n+2)} \+x \+x^\transp \right)
\,,
\numberthis
\\
\+b_\Mfsd
&= \frac{1}{\tau^2} \left(
    \+B_{\Mfsa,1} \hat{\+y}_{-\Mfsa}
    - \+B_{\Mfsb,1} \hat{\+y}_{-\Mfsb}
    + \+B_{\Mfsd,2} \+\mu_{\star}
\right)
\\
    &= \frac{1}{\tau^2} m_\star \left(
         \frac{1}{(n+2)(n+1)} \ones_{n} - \frac{1}{n+2} \+x
    \right)
\,,
\numberthis
\\
c_{\Mfsd} &= \frac{1}{\tau^2} \Bigg(
     \hat{\+y}_{-\Mfsa}^\transp \+C_{\Mfsa,1}  \hat{\+y}_{-\Mfsa}
    -  \hat{\+y}_{-\Mfsb}^\transp \+C_{\Mfsb,1}  \hat{\+y}_{-\Mfsb}
    \\&\qquad\quad
    +  \hat{\+y}_{-\Mfsa}^\transp \+C_{\Mfsa,2} \+\mu_{\star}
    -  \hat{\+y}_{-\Mfsb}^\transp \+C_{\Mfsb,2} \+\mu_{\star}
    \\&\qquad\quad
    + \+\mu_{\star}^\transp \+C_{\Mfsd,3} \+\mu_{\star}
    + \+\sigma_\star^\transp \+C_{\Mfsd,3} \+\sigma_\star
    \Bigg)
    + c_{\Mfsd,4}
\\
    &= \frac{1}{\tau^2} \Bigg(
     \beta_1^2\+x^\transp \left(-\frac{n}{2(n+1)}\eye_{n} + \frac{1}{2(n+1)} \ones_{n} \ones_{n}^\transp \right)  \+x
    \\&\qquad\quad
    + \beta_1 \+x^\transp \left( - \frac{n}{n+1}\eye_{n} + \frac{1}{n+1} \ones_{n} \ones_{n}^\transp \right) \+\mu_{\star}
    \\&\qquad\quad
    - \frac{n}{2(n+1)(n+2)} \left(
        \+\mu_{\star}^\transp  \eye_{n} \+\mu_{\star}
        + \+\sigma_\star^\transp \eye_{n} \+\sigma_\star
    \right)
    \Bigg)
    \\&\quad
    + \frac{n}{2} \log \frac{n+2}{n+1}
\\
    &= \frac{1}{\tau^2} \Bigg(
     -\beta_1^2 \frac{n^2}{2(n+1)}
    - \beta_1 m_{\star} \frac{n}{n+1}
    - \frac{n}{2(n+1)(n+2)} \left(
        m_{\star}^2
        + n s_{\star}^2
    \right)
    \Bigg)
    \\&\quad
    + \frac{n}{2} \log \frac{n+2}{n+1}
\,.
\numberthis
\end{align*}

\paragraph{First Moment}
\label{app_sec_analytic_one_cov_elpd_1_moment}

In this section, we formulate the first raw moment $m_1$ in Equation~\eqref{app_eq_analytic_1_moment} for $\elpdC{\Mfd}{\y}$ in the one covariate case defined in Appendix~\ref{app_sec_one_cov}.
The trace of $\+\Sigma_\star^{1/2}\+A_{\Mfsd}\+\Sigma_\star^{1/2} = s_\star^2\+A_{\Mfsd}$ simplifies to
\begin{align*}
\operatorname{tr}\left(\+\Sigma_\star^{1/2}\+A_{\Mfsd}\+\Sigma_\star^{1/2}\right)
    &= \frac{1}{\tau^2} s_\star^2 n \left(
        - \frac{1}{2(n+1)(n+2)} + \frac{1}{2(n+2)}
    \right)
    \\
    &= \frac{1}{\tau^2} s_\star^2 \frac{n^2}{2(n+1)(n+2)}
    \,.
    \numberthis
\end{align*}
Furthermore
\begin{align*}
\+b_{\Mfsd}^\transp \+\mu_\star
    &= \frac{1}{\tau^2} m_\star \left(
         \frac{1}{(n+2)(n+1)} m_\star - \frac{1}{n+2} m_\star
    \right)
    \\
    &= -\frac{1}{\tau^2} m_\star^2 \frac{n}{(n+2)(n+1)}
    \numberthis
\end{align*}
and
\begin{align*}
\+\mu_\star^\transp \+A_{\Mfsd} \+\mu_\star
    &= \frac{1}{\tau^2} \left( - \frac{1}{2(n+1)(n+2)} m_\star^2 + \frac{1}{2(n+2)} m_\star^2 \right)
    \\
    &= \frac{1}{\tau^2} m_\star^2 \frac{n}{2(n+2)(n+1)}
    \,.
    \numberthis
\end{align*}
Now Equation~\eqref{app_eq_analytic_1_moment} simplifies to
\begin{align*}
m_1
    &= \operatorname{tr}\left(\+\Sigma_\star^{1/2}\+A_{\Mfsd}\+\Sigma_\star^{1/2}\right)
        + c_{\Mfsd}
        + \+b_{\Mfsd}^\transp \+\mu_\star
        + \+\mu_\star^\transp \+A_{\Mfsd} \+\mu_\star
    \\
    & = \frac{1}{\tau^2} \left(
        P_{1,1}(n) \beta_1^2
        + Q_{1,0}(n) \beta_1 m_\star
        + R_{1,-1}(n) m_\star^2
    \right)
    + F_{1}(n)
    \,,
    \numberthis
\intertext{where}
    P_{1,1}(n) &= - \frac{n^2}{2(n+1)}
    \numberthis\\
    Q_{1,0}(n) &= - \frac{n}{n+1}
    \numberthis\\
    R_{1,-1}(n) &= - \frac{n}{(n+2)(n+1)}
    \numberthis\\
    F_{1}(n) &= \frac{n}{2} \log \frac{n+2}{n+1}
    \,,
    \numberthis
\end{align*}
where the first subscript indicates the corresponding order of the moment, and for the rational functions $P_{1,1}$, $Q_{1,0}$, and $R_{1,-1}$, the second subscript indicates the degree of the rational as a difference between the degrees of the numerator and the denominator.
It can be seen that $m_1$ does not depend on $s_\star$.

\paragraph{Second Moment}
\label{app_sec_analytic_one_cov_elpd_2_moment}

In this section, we formulate the second moment $\overline{m}_2$ about the mean in Equation~\eqref{app_eq_analytic_2_moment} for $\elpdC{\Mfd}{\y}$ in the one covariate case defined in Appendix~\ref{app_sec_one_cov}.
The second power of $\+A_{\Mfsd}$ is
\begin{align*}
\+A_{\Mfsd}^2
    &= \frac{1}{\tau^4} \Bigg(
        \frac{n}{4(n+1)^2(n+2)^2}\ones_{n}\ones_{n}^\transp
        + \frac{n}{4(n+2)^2}\+x\+x^\transp
    \Bigg)
    \,.
    \numberthis
\end{align*}
The trace in Equation~\eqref{app_eq_analytic_2_moment} simplifies to
\begin{align*}
\operatorname{tr}\left(\left(\+\Sigma_\star^{1/2}\+A_{\Mfsd}\+\Sigma_\star^{1/2}\right)^2\right)
    &= \frac{1}{\tau^4} s_\star^4 n \bigg(
        \frac{n}{4(n+1)^2(n+2)^2}
        + \frac{n}{4(n+2)^2}
    \bigg)
    \\
    &= \frac{1}{\tau^4} s_\star^4 \frac{n^2 (n^2 + 2 n + 2)}{4 (n + 1)^2 (n + 2)^2}
    \,.
    \numberthis
\end{align*}
Furthermore
\begin{align*}
\+b_{\Mfsd}^\transp \+b_{\Mfsd}
    &= \frac{1}{\tau^4} m_\star^2 \frac{n (n^2 + 2 n + 2)}{(n + 1)^2 (n + 2)^2}
\,,
\numberthis
\\
\+b_{\Mfsd}^\transp \+A_{\Mfsd} \+\mu_\star
    &= - \frac{1}{\tau^4} m_\star^2 \frac{n (n^2 + 2 n + 2)}{2(n + 1)^2 (n + 2)^2}
\,,
\numberthis
\intertext{and}
\+\mu_\star^\transp \+A_{\Mfsd}^2 \+\mu_\star
    &= \frac{1}{\tau^4} m_\star^2 \frac{n (n^2 + 2 n + 2)}{4 (n + 1)^2 (n + 2)^2}
\,.
\numberthis
\end{align*}
Now Equation~\eqref{app_eq_analytic_2_moment} simplifies to
\begin{align*}
\overline{m}_2
    &= 2 \operatorname{tr}\left(\left(\+\Sigma_\star^{1/2}\+A_{\Mfsd}\+\Sigma_\star^{1/2}\right)^2\right)
    + \+b_{\Mfsd}^\transp \+\Sigma_\star \+b_{\Mfsd}
    \\ &\quad
    + 4 \+b_{\Mfsd}^\transp \+\Sigma_\star \+A_{\Mfsd} \+\mu_\star
    + 4 \+\mu_\star^\transp \+A_{\Mfsd} \+\Sigma_\star \+A_{\Mfsd} \+\mu_\star
    \\
    & = \frac{1}{\tau^4} S_{2,0}(n) s_\star^4
    \,,
    \numberthis
\intertext{%
where
}
    S_{2,0}(n) &= \frac{n^2 (n^2 + 2 n + 2)}{2 (n + 1)^2 (n + 2)^2}
    \,,
    \numberthis
\end{align*}
and the first subscript in the rational function $S_{2,0}$ indicates the corresponding order of the moment, and the second subscript indicates the degree of the rational as a difference between the degrees of the numerator and the denominator.
It can be seen that $\overline{m}_2$ does not depend on $\beta_1$ and $m_\star$.

\paragraph{Mean Relative to the Standard Deviation}
\label{app_sec_analytic_one_Rel_SD}

In this section, we formulate the ratio of mean and standard deviation $m_1 \big/ \sqrt{\overline{m}_2}$ for $\elpdC{\Mfd}{\y}$ in the one covariate case defined in Appendix~\ref{app_sec_one_cov}.
Combining results from appendices~\ref{app_sec_analytic_one_cov_elpd_1_moment} and~\ref{app_sec_analytic_one_cov_elpd_2_moment}, we get
\begin{align*}
    \frac{m_1}{\sqrt{\overline{m}_2}}
    & =
    \frac{
        P_{1,1}(n) \beta_1^2
        + Q_{1,0}(n) \beta_1 m_\star
        + R_{1,-1}(n) m_\star^2
        + \tau^2 F_{1}(n)
    }{\sqrt{ S_{2,0}(n) s_\star^4 }}
\,,
\numberthis
\intertext{where}
    P_{1,1}(n) &= - \frac{n^2}{2(n+1)}
    \numberthis\\
    Q_{1,0}(n) &= - \frac{n}{n+1}
    \numberthis\\
    R_{1,-1}(n) &= - \frac{n}{(n+2)(n+1)}
    \numberthis\\
    F_{1}(n) &= \frac{n}{2} \log \frac{n+2}{n+1}
    \numberthis\\
    S_{2,0}(n) &= \frac{n^2 (n^2 + 2 n + 2)}{2 (n + 1)^2 (n + 2)^2}
\,,
\numberthis
\end{align*}
where the first subscript in the rational functions $P_{1,1}$, $Q_{1,0}$, $R_{1,-1}$, and $S_{2,0}$ indicates the corresponding order of the associated moment. The second subscript indicates the degree of the rational as a difference between the degrees of the numerator and the denominator.

Let us inspect the behaviour of $m_1 \big/ \sqrt{\overline{m}_2}$ when $n \rightarrow \infty$.
We have
\begin{align}
    \lim_{n\rightarrow\infty} P_{1,1}(n) &= - \infty
    \,,
    \\
    \lim_{n\rightarrow\infty} S_{2,0}(n) &= \frac{1}{2}
\intertext{and}
    \lim_{n\rightarrow\infty} F_{1}(n) &= \frac{1}{2}
\,.
\end{align}
Thus we get
\begin{align*}
    \lim_{n\rightarrow\infty} \frac{m_1}{\sqrt{\overline{m}_2}}
    & =
    \frac{
        \lim_{n\rightarrow\infty} \Big(
        P_{1,1}(n) \beta_1^2
        + Q_{1,0}(n) \beta_1 m_\star
        + R_{1,-1}(n) m_\star^2
        + \tau^2 F_{1}(n)
        \Big)
    }{\sqrt{ \lim_{n\rightarrow\infty} S_{2,0}(n) s_\star^4 }}
\\
    &= \begin{dcases}
        \frac{\tau^2}{\sqrt{2} s_\star^2} & \text{when } \beta_1 = 0 \,, \\
        - \infty & \text{otherwise} \,.
    \end{dcases}
    \numberthis
    \label{app_eq_onecov_ninf_elpd}
\end{align*}


\subsubsection{LOO-CV}

In this section, we derive a simplified analytic form for $\elpdHatC{\Mfd}{\y}$ presented in Appendix~\ref{app_sec_analytic_loocv_d} and for some moments of interest in the one covariate case defined in Appendix~\ref{app_sec_one_cov}. First we derive the parameters $\widetilde{\+A}_\Mfsd$, $\widetilde{\+b}_\Mfsd$, and $\widetilde{c}_\Mfsd$ defined in Appendix~\ref{app_sec_analytic_loocv_d} and then we use them to derive the respective moments of interest defined in Appendix~\ref{app_sec_analytic_moments}.

\paragraph{Parameters}
\label{app_sec_one_cov_loo_params}

Following the notation in Appendix~\ref{app_sec_analytic_loocv}, in the one covariate case defined in Appendix~\ref{app_sec_one_cov}, let us find simplified form for matrix $\widetilde{\+D}_{\Msk}$ and $\widetilde{\+P}_\Msk^\transp \widetilde{\+D}_{\Msk} \widetilde{\+P}_\Msk$ for $\Mk \in \{ \Mfsa, \Mfsb\}$.
For the model \Mfa{} we have
\begin{align*}
\+v(\Mfsa,i)
    &= \+X_{[\bcdot,\Mfsa]}
        \left(\+X_{[-i,\Mfsa]}^\transp\+X_{[-i,\Mfsa]}\right)^{-1}
        X_{[i,\Mfsa]}^\transp
\\
    &= \ones_{n} \left(\ones_{n-1}^\transp\ones_{n-1}\right)^{-1}
\\
    &= \frac{1}{n-1} \ones_{n},
    \numberthis
\end{align*}
for $i=1,2,\dots,n$.
From this, we get
\begin{align*}
\widetilde{\+D}_{\Mfsa[i,i]}
    &= \left( v(\Mfsa,i)_i + 1 \right)^{-1} \\
    &= \left( \frac{1}{n-1} + 1 \right)^{-1} \\
    &= \frac{n-1}{n}
    \numberthis
\intertext{and further}
\widetilde{\+D}_{\Mfsa} &= \frac{n-1}{n} \eye_{n}
\numberthis
\,.
\end{align*}
According to Equation~\eqref{app_eq_PtT_Pt_m}, for the diagonal elements of $\widetilde{\+P}_\Msk^\transp \widetilde{\+D}_{\Msk} \widetilde{\+P}_\Msk$ we get
\begin{align*}
\left[\widetilde{\+P}_{\Mfsa}^\transp \widetilde{\+D}_{\Mfsa} \widetilde{\+P}_{\Mfsa}\right]_{[i,i]}
    &= \sum_{p\neq\{i\}} \frac{\left(\frac{1}{n-1}\right)^2}{\frac{1}{n-1} +1}
        + \frac{1}{\frac{1}{n-1} + 1}
    \\
    &= \frac{n-1}{n (n-1)^2} \sum_{p\neq\{i\}} 1 + \frac{n-1}{n}
    \\
    &= \frac{1}{n} + \frac{n-1}{n}
    \\
    &= 1
\,,
\numberthis
\intertext{and for the off-diagonal elements we get}
    \left[\widetilde{\+P}_{\Mfsa}^\transp \widetilde{\+D}_{\Mfsa} \widetilde{\+P}_{\Mfsa}\right]_{[i,j]}
    &= \sum_{p\neq\{i,j\}}
            \frac{\frac{1}{n-1} \frac{1}{n-1}}{\frac{1}{n-1} +1}
        - \frac{\frac{1}{n-1}}{\frac{1}{n-1} +1}
        - \frac{\frac{1}{n-1}}{\frac{1}{n-1} +1}
    \\
    &= \frac{n-1}{n (n-1)^2} \sum_{p\neq\{i,j\}} 1 - 2\frac{1}{n}
    \\
    &= \frac{n-2 - 2(n-1)}{n (n-1)}
    \\
    &= - \frac{1}{n-1}
\,,
\numberthis
\end{align*}
where $i,j=1,2,\dots,n, \; i \neq j$.
For the model \Mfb{} we have
\begin{align*}
\+v(\Mfsb,i)
&= \+X_{[\bcdot,\Mfsb]}
    \left(\+X_{[-i,\Mfsb]}^\transp\+X_{[-i,\Mfsb]}\right)^{-1}
    \+X_{[i,\Mfsb]}^\transp,
\\
&=  \begin{bmatrix} \ones_n & \+x \end{bmatrix}
    \left(
        \begin{bmatrix} \ones_{n-1} & \+x_{-i} \end{bmatrix}^\transp
        \begin{bmatrix} \ones_{n-1} & \+x_{-i} \end{bmatrix}
    \right)^{-1}
    \begin{bmatrix} 1 & x_i \end{bmatrix}^\transp
\\
&=  \begin{bmatrix} \ones_n & \+x \end{bmatrix}
    \begin{bmatrix}
        n-1 & \ones_{n-1}^\transp\+x_{-i} \\
        \ones_{n-1}^\transp\+x_{-i} & \+x_{-i}^\transp\+x_{-i}
    \end{bmatrix}^{-1}
    \begin{bmatrix} 1 & x_i \end{bmatrix}^\transp
\\
&=  \frac{
    \begin{bmatrix} \ones_n & \+x \end{bmatrix}
    \begin{bmatrix}
        \+x_{-i}^\transp\+x_{-i} & -\ones_{n-1}^\transp\+x_{-i} \\
        -\ones_{n-1}^\transp\+x_{-i} & n-1
    \end{bmatrix}
    \begin{bmatrix} 1 & x_i \end{bmatrix}^\transp
}{
    (n-1)\+x_{-i}^\transp\+x_{-i}
    - \left(\ones_{n-1}^\transp\+x_{-i}\right)^2
},
\numberthis
\\
v(\Mfsb,i)_j &= \frac{
        \+x_{-i}^\transp\+x_{-i}
        - (x_i + x_j) \ones_{n-1}^\transp\+x_{-i}
        + (n-1) x_i x_j
    }{
        (n-1)\+x_{-i}^\transp\+x_{-i}
        - \left(\ones_{n-1}^\transp\+x_{-i}\right)^2
    },
\numberthis
\end{align*}
for all $i,j=1,2,\dots,n$.
Now we can write
\begin{equation}
 v(\Mfsb,i)_j
 = \frac{n-1 + x_i (x_i + x_j) + (n-1) x_i x_j}{ (n-1)^2 - x_i^2 }
 = \frac{x_i x_j + 1}{n-2}
\end{equation}
for which $v(\Mfsb,i)_i = \frac{2}{n-2}$ in particular. From this we get
\begin{align*}
\widetilde{\+D}_{\Mfsb[i,i]}
    &= \left( v(\Mfsb,i)_i + 1 \right)^{-1} \\
    &= \left( \frac{2}{n-2} + 1 \right)^{-1} \\
    &= \frac{n-2}{n}
    \numberthis
\intertext{and further}
\widetilde{\+D}_{\Mfsb} &= \frac{n-2}{n} \eye_{n}
\,.
\numberthis
\end{align*}
According to Equation~\eqref{app_eq_PtT_Pt_m}, for the diagonal elements of $\widetilde{\+P}_{\Mfsb}^\transp \widetilde{\+D}_{\Mfsb} \widetilde{\+P}_{\Mfsb}$ we get
\begin{align*}
\left[\widetilde{\+P}_{\Mfsb}^\transp \widetilde{\+D}_{\Mfsb} \widetilde{\+P}_{\Mfsb}\right]_{[i,i]}
    &= \sum_{p\neq\{i\}} \frac{\left(\frac{x_i x_p + 1}{n-2}\right)^2}{\frac{2}{n-2} +1}
        + \frac{1}{\frac{2}{n-2} + 1}
    \\
    &= \frac{n-2}{n} \left( \frac{1}{(n-2)^2} \sum_{p\neq\{i\}} 2(x_i x_p + 1) + 1 \right)
    \\
    &= \frac{n-2}{n} \left( \frac{2}{(n-2)^2} \left( x_i \sum_{p\neq\{i\}} x_p + n - 1\right) + 1 \right)
    \\
    &= \frac{n-2}{n} \left( \frac{2}{(n-2)^2} \left( -x_i^2 + n - 1\right) + 1 \right)
    \\
    & = \frac{n-2}{n} \left( \frac{2}{(n-2)^2} (n - 2) + 1 \right)
    \\
    & = \frac{n-2}{n} \frac{n}{n-2}
    \\
    &= 1
\,,
\numberthis
\end{align*}
and for the off-diagonal elements, we get
\begin{align*}
\left[\widetilde{\+P}_{\Mfsb}^\transp \widetilde{\+D}_{\Mfsb} \widetilde{\+P}_{\Mfsb}\right]_{[i,j]}
    &= \sum_{p\neq\{i,j\}}
            \frac{\frac{x_p x_i + 1}{n-2} \frac{x_p x_j + 1}{n-2}}{\frac{2}{n-2} +1}
        - \frac{\frac{x_i x_j + 1}{n-2}}{\frac{2}{n-2} +1}
        - \frac{\frac{x_i x_j + 1}{n-2}}{\frac{2}{n-2} +1}
    \\
    &= \frac{n-2}{n(n-2)^2} \sum_{p\neq\{i,j\}} \left( x_p^2 x_i x_j + x_p (x_i+x_j) + 1 \right) - 2\frac{n-2}{n(n-2)}(x_i x_j + 1)
    \\
    &= \left( \frac{1}{n(n-2)} \sum_{p\neq\{i,j\}} x_p^2  - \frac{2}{n}\right) x_i x_j
        + \frac{1}{n(n-2)} (x_i + x_j) \sum_{p\neq\{i,j\}} x_p
        \\&\quad
        + \frac{1}{n(n-2)} \sum_{p\neq\{i,j\}} 1 - \frac{2}{n}
    \\
    &= -\frac{1}{n} x_i x_j
        + \frac{1}{n(n-2)} (x_i + x_j) \sum_{p\neq\{i,j\}} x_p
        - \frac{1}{n}
\,,
\numberthis
\end{align*}
where $i,j=1,2,\dots,n, \; i \neq j$. When $x_i = x_j$, we have
$x_i x_j = 1$
and
$(x_i + x_j)\sum_{p\neq\{i,j\}} x_p = (2 x_i)(-2x_i)=-4$
and
\begin{align*}
\left[\widetilde{\+P}_{\Mfsb}^\transp \widetilde{\+D}_{\Mfsb} \widetilde{\+P}_{\Mfsb}\right]_{[i,j]}
    &= - \frac{1}{n}
        + \frac{1}{n(n-2)} (-4)
        - \frac{1}{n}
    \\
    &= -\frac{2}{n-2}
\,,
\numberthis
\intertext{%
and when $x_i \neq x_j$, we have $x_i x_j = -1$ and $(x_i + x_j)\sum_{p\neq\{i,j\}} x_p = 0 \cdot 0 = 0$ and}
\left[\widetilde{\+P}_{\Mfsb}^\transp \widetilde{\+D}_{\Mfsb} \widetilde{\+P}_{\Mfsb}\right]_{[i,j]}
    &= \frac{1}{n}
        + \frac{1}{n(n-2)} 0 \cdot 0
        - \frac{1}{n}
    \\
    &= 0
\,.
\numberthis
\end{align*}
Now we can summarise for both models \Mfa{} and \Mfb{} that
\begin{align}
\left[\widetilde{\+P}_{\Mfsa}^\transp \widetilde{\+D}_{\Mfsa} \widetilde{\+P}_{\Mfsa}\right]_{[i,j]}
    &= \begin{dcases}
        1   & \text{when } i = j \,, \\
        - \frac{1}{n-1} & \text{when } i \neq j \,,
    \end{dcases}
\\
\left[\widetilde{\+P}_{\Mfsb}^\transp \widetilde{\+D}_{\Mfsb} \widetilde{\+P}_{\Mfsb}\right]_{[i,j]}
    &= \begin{dcases}
        1   & \text{when } i = j \,, \\
        -\frac{2}{n-2}  & \text{when } i \neq j, \text{ and } x_i = x_j \,,\\
        0 & \text{when } i \neq j, \text{ and } x_i \neq x_j \,,
    \end{dcases}
\end{align}
and further, simplify
\begin{align}
\widetilde{\+P}_{\Mfsa}^\transp \widetilde{\+D}_{\Mfsa} \widetilde{\+P}_{\Mfsa}
    &= \frac{n}{n-1} \eye_{n} - \frac{1}{n-1} \ones_{n} \ones_{n}^\transp
\,,
\\
\widetilde{\+P}_{\Mfsb}^\transp \widetilde{\+D}_{\Mfsb} \widetilde{\+P}_{\Mfsb}
    &= \frac{n}{n-2}\eye_{n} - \frac{1}{n-2}\left(\ones_{n} \ones_{n}^\transp + \+x \+x^\transp \right)
\,.
\end{align}

Now we get
\begin{align*}
\widetilde{\+B}_{\Mfsa,1} &= - \widetilde{\+P}_\Mfsa^\transp \widetilde{\+D}_{\Mfsa} \widetilde{\+P}_\Mfsa
\\
    &= - \frac{n}{n-1} \eye_{n} + \frac{1}{n-1} \ones_{n} \ones_{n}^\transp
    \numberthis
\,,
\\
\widetilde{\+B}_{\Mfsb,1} &= - \widetilde{\+P}_\Mfsb^\transp \widetilde{\+D}_{\Mfsb} \widetilde{\+P}_\Mfsb
\\
    &= - \frac{n}{n-2}\eye_{n} + \frac{1}{n-2}\left(\ones_{n} \ones_{n}^\transp + \+x \+x^\transp \right)
    \numberthis
\,,
\\
\widetilde{\+C}_{\Mfsa,1} &= -\frac{1}{2} \widetilde{\+P}_\Mfsa^\transp \widetilde{\+D}_{\Mfsa} \widetilde{\+P}_\Mfsa
\\
    &= - \frac{n}{2(n-1)} \eye_{n} + \frac{1}{2(n-1)} \ones_{n} \ones_{n}^\transp
    \numberthis
\,,
\\
\widetilde{\+C}_{\Mfsb,1} &= -\frac{1}{2} \widetilde{\+P}_\Mfsb^\transp \widetilde{\+D}_{\Mfsb} \widetilde{\+P}_\Mfsb
\\
    &= - \frac{n}{2(n-2)}\eye_{n} + \frac{1}{2(n-2)}\left(\ones_{n} \ones_{n}^\transp + \+x \+x^\transp \right)
    \numberthis
\,,
\intertext{and}
\widetilde{\+A}_{\Mfsd,1} 
    &= -\frac{1}{2} \left(
        \widetilde{\+P}_{\Mfsa}^\transp \widetilde{\+D}_{\Mfsa} \widetilde{\+P}_{\Mfsa}
        - \widetilde{\+P}_{\Mfsb}^\transp \widetilde{\+D}_{\Mfsb} \widetilde{\+P}_{\Mfsb}
    \right)
    \,,
\\
&=  \frac{n}{2(n-2)(n-1)} \eye_{n} - \frac{1}{2(n-2)(n-1)} \ones_{n} \ones_{n}^\transp - \frac{1}{2(n-2)} \+x \+x^\transp
    \,, \numberthis
\\
\widetilde{c}_{\Mfsd,4} 
    &= \frac{1}{2} \log\left(\prod_{i=1}^n \frac{\widetilde{\+D}_{\Mfsa[i,i]}}{\widetilde{\+D}_{\Mfsb[i,i]}} \right)
\\
    &= \frac{1}{2} \log\left(\prod_{i=1}^n \frac{\frac{n-1}{n}}{\frac{n-2}{n}} \right)
\\
    &= \frac{n}{2} \log\frac{n-1}{n-2}
\numberthis
\,.
\end{align*}
Finally, we get the desired parameters
\begin{align}
\widetilde{\+A}_{\Mfsd}
    &= \frac{1}{\tau^2} \widetilde{\+A}_{\Mfsd,1}
\\
    &= \frac{1}{\tau^2} \left(
        \frac{n}{2(n-2)(n-1)} \eye_{n} - \frac{1}{2(n-2)(n-1)} \ones_{n} \ones_{n}^\transp - \frac{1}{2(n-2)} \+x \+x^\transp
    \right)
    \,, \numberthis
\\
\widetilde{\+b}_{\Mfsd}
    &= \frac{1}{\tau^2} \left( \widetilde{\+B}_{\Mfsa,1} \hat{\+y}_{-\Mfsa} - \widetilde{\+B}_{\Mfsb,1} \hat{\+y}_{-\Mfsb} \right)
\\
    &= - \frac{1}{\tau^2} \beta_1 \frac{n}{n-1} \+x
    \,, \numberthis
\\
\widetilde{c}_{\Mfsd}
    &= \frac{1}{\tau^2}  \left(
        \hat{\+y}_{-\Mfsa}^\transp \widetilde{\+C}_{\Mfsa,1} \hat{\+y}_{-\Mfsa}
        - \hat{\+y}_{-\Mfsb}^\transp \widetilde{\+C}_{\Mfsb,1} \hat{\+y}_{-\Mfsb}
     \right)
    + \widetilde{c}_{\Mfsd,4}
\\
    &= -\frac{1}{\tau^2} \beta_1^2 \frac{n^2}{2(n-1)} + \frac{n}{2} \log\frac{n-1}{n-2}
\,. \numberthis
\end{align}

\paragraph{First Moment}
\label{app_sec_analytic_one_cov_loocv_1_moment}

In this section, we formulate the first raw moment $m_1$, defined in a general setting in Equation~\eqref{app_eq_analytic_1_moment}, for $\elpdHatC{\Mfd}{\y}$ in the one covariate case defined in Appendix~\ref{app_sec_one_cov}.
The trace of $\+\Sigma_\star^{1/2}\widetilde{\+A}_{\Mfsd}\+\Sigma_\star^{1/2} = s_\star^2\widetilde{\+A}_{\Mfsd}$ simplifies to
\begin{align*}
\operatorname{tr}\left(\+\Sigma_\star^{1/2}\widetilde{\+A}_{\Mfsd}\+\Sigma_\star^{1/2}\right)
    &= \frac{1}{\tau^2} s_\star^2 n \left(
        \frac{n}{2(n-2)(n-1)} - \frac{1}{2(n-2)(n-1)} - \frac{1}{2(n-2)}
    \right)
    \\
    &= 0
    \numberthis
\end{align*}
as was also shown to hold in a general case in Appendix~\ref{app_sec_analytic_loocv_additional_properties}.
Furthermore,
\begin{align*}
\widetilde{\+b}_{\Mfsd}^\transp \+\mu_\star
    &= - \frac{1}{\tau^2} \beta_1 m_\star \frac{n}{n-1}
    \numberthis
\end{align*}
and
\begin{align*}
\+\mu_\star^\transp \widetilde{\+A}_{\Mfsd} \+\mu_\star
    &= \frac{1}{\tau^2} \left(
        \frac{n}{2(n-2)(n-1)} m_\star^2 - \frac{1}{2(n-2)(n-1)} m_\star^2 - \frac{1}{2(n-2)} m_\star^2
    \right)
    \\
    &= 0
    \,.
    \numberthis
\end{align*}
Now Equation~\eqref{app_eq_analytic_1_moment} simplifies to
\begin{align*}
m_1
    &= \operatorname{tr}\left(\+\Sigma_\star^{1/2}\widetilde{\+A}_{\Mfsd}\+\Sigma_\star^{1/2}\right)
        + \widetilde{c}_{\Mfsd}
        + \widetilde{\+b}^\transp \+\mu_\star
        + \+\mu_\star^\transp \widetilde{\+A}_{\Mfsd} \+\mu_\star
    \\
    & = \frac{1}{\tau^2} \left(
        P_{1,1}(n) \beta_1^2
        + Q_{1,0}(n) \beta_1 m_\star
    \right)
    + F_{1}(n)
    \,,
    \numberthis
\intertext{where}
    P_{1,1}(n) &= - \frac{n^2}{2(n-1)}
    \numberthis\\
    Q_{1,0}(n) &= - \frac{n}{n-1}
    \numberthis\\
    F_{1}(n) &= \frac{n}{2} \log \frac{n-1}{n-2}
    \,,
    \numberthis
\end{align*}
where the first subscript indicates the corresponding order of the moment, and for the rational functions $P_{1,1}$ and $Q_{1,0}$, the second subscript indicates the degree of the rational as a difference between the degrees of the numerator and the denominator.
It can be seen that $m_1$ does not depend on $s_\star$.

\paragraph{Second Moment}
\label{app_sec_analytic_one_cov_loocv_2_moment}

In this section, we formulate the second moment $\overline{m}_2$ about the mean in Equation~\eqref{app_eq_analytic_2_moment} for $\elpdHatC{\Mfd}{\y}$ in the one covariate case defined in Appendix~\ref{app_sec_one_cov}.
The second power of $\widetilde{\+A}_{\Mfsd}$ is
\begin{align*}
\widetilde{\+A}_{\Mfsd}^2
    &= \frac{1}{\tau^4} \Bigg(
        \frac{n^2}{4(n-2)^2(n-1)^2} \eye_n
        - \frac{n}{4(n-2)^2(n-1)^2} \ones_{n}\ones_{n}^\transp
        + \frac{n(n-3)}{4(n-2)^2(n-1)}\+x\+x^\transp
    \Bigg)
    \,.
    \numberthis
\end{align*}
The trace in Equation~\eqref{app_eq_analytic_2_moment} simplifies to
\begin{align*}
&\operatorname{tr}\left(\left(\+\Sigma_\star^{1/2}\widetilde{\+A}_{\Mfsd}\+\Sigma_\star^{1/2}\right)^2\right)
\\
    &\qquad= \frac{1}{\tau^4} s_\star^4 n \bigg(
        \frac{n^2}{4(n-2)^2(n-1)^2}
        - \frac{n}{4(n-2)^2(n-1)^2}
        + \frac{n(n-3)}{4(n-2)^2(n-1)}
    \bigg)
    \\
    &\qquad= \frac{1}{\tau^4} s_\star^4 \frac{n^2}{4 (n-2) (n-1)}
    \,.
    \numberthis
\end{align*}
Furthermore
\begin{align*}
\widetilde{\+b}_{\Mfsd}^\transp \widetilde{\+b}_{\Mfsd}
    &= \frac{1}{\tau^4} \beta_1^2 \frac{n^3}{(n-1)^2}
\,,
\numberthis
\\
\widetilde{\+b}_{\Mfsd}^\transp \widetilde{\+A}_{\Mfsd} \+\mu_\star
    &= \frac{1}{\tau^4} \beta_1 m_\star \frac{n^2}{2(n-1)^2}
\,,
\numberthis
\intertext{and}
\+\mu_\star^\transp \widetilde{\+A}_{\Mfsd}^2 \+\mu_\star
    &= \frac{1}{\tau^4} m_\star^2 \frac{n}{4 (n-2) (n-1)}
\,.
\numberthis
\end{align*}
Now Equation~\eqref{app_eq_analytic_2_moment} simplifies to
\begin{align*}
\overline{m}_2
    &= 2 \operatorname{tr}\left(\left(\+\Sigma_\star^{1/2}\widetilde{\+A}_{\Mfsd}\+\Sigma_\star^{1/2}\right)^2\right)
    + \widetilde{\+b}_{\Mfsd}^\transp \+\Sigma_\star \widetilde{\+b}_{\Mfsd}
    \\ &\quad
    + 4 \widetilde{\+b}_{\Mfsd}^\transp \+\Sigma_\star \widetilde{\+A}_{\Mfsd} \+\mu_\star
    + 4 \+\mu_\star^\transp \widetilde{\+A}_{\Mfsd} \+\Sigma_\star \widetilde{\+A}_{\Mfsd} \+\mu_\star
    \\
    & = \frac{1}{\tau^4} \left(
        P_{2,1}(n) \beta_1^2 s_\star^2
        + Q_{2,0}(n) \beta_1 m_\star s_\star^2
        + R_{2,-1}(n) m_\star^2 s_\star^2
        + S_{2,0}(n) s_\star^4
    \right)
    \,,
    \numberthis
\intertext{where}
    P_{2,1}(n) &= \frac{n^3}{(n-1)^2}
    \numberthis\\
    Q_{2,0}(n) &= \frac{2n^2}{(n-1)^2}
    \numberthis\\
    R_{2,-1}(n) &= \frac{n}{(n - 2) (n - 1)}
    \numberthis\\
    S_{2,0}(n) &= \frac{n^2}{2 (n - 2) (n - 1)}
    \,,
    \numberthis
\end{align*}
where the first subscript in the rational functions $P_{2,1}$, $Q_{2,0}$, $R_{2,-1}$, and $S_{2,0}$ indicates the corresponding order of the moment. The second subscript indicates the degree of the rational as a difference between the degrees of the numerator and the denominator.
\paragraph{Mean Relative to the Standard Deviation}
\label{app_sec_analytic_two_Rel_SD}
In this section, we formulate the ratio of mean and standard deviation $m_1 \big/ \sqrt{\overline{m}_2}$ for $\elpdHatC{\Mfd}{\y}$ in the one covariate case defined in Appendix~\ref{app_sec_one_cov}.
Combining results from appendices~\ref{app_sec_analytic_one_cov_loocv_1_moment} and~\ref{app_sec_analytic_one_cov_loocv_2_moment}, we get
\begin{align*}
    \frac{m_1}{\sqrt{\overline{m}_2}}
    & =
    \frac{
        P_{1,1}(n) \beta_1^2
        + Q_{1,0}(n) \beta_1 m_\star
        + \tau^2 F_{1}(n)
    }{\sqrt{
        P_{2,1}(n) \beta_1^2 s_\star^2
        + Q_{2,0}(n) \beta_1 m_\star s_\star^2
        + R_{2,-1}(n) m_\star^2 s_\star^2
        + S_{2,0}(n) s_\star^4
    }}
\,,
\numberthis
\intertext{where}
    P_{1,1}(n) &= - \frac{n^2}{2(n-1)}
    \numberthis\\
    Q_{1,0}(n) &= - \frac{n}{n-1}
    \numberthis\\
    F_{1}(n) &= \frac{n}{2} \log \frac{n-1}{n-2}
    \numberthis\\
    P_{2,1}(n) &= \frac{n^3}{(n-1)^2}
    \numberthis\\
    Q_{2,0}(n) &= \frac{2n^2}{(n-1)^2}
    \numberthis\\
    R_{2,-1}(n) &= \frac{n}{(n - 2) (n - 1)}
    \numberthis\\
    S_{2,0}(n) &= \frac{n^2}{2 (n - 2) (n - 1)}
    \,,
    \numberthis
\end{align*}
where the first subscript in the rational functions $P_{1,1}$, $Q_{1,0}$, $R_{1,-1}$, and $S_{2,0}$ indicates the corresponding order of the associated moment. The second subscript indicates the degree of the rational as a difference between the degrees of the numerator and the denominator.

Let us inspect the behaviour of $m_1 \big/ \sqrt{\overline{m}_2}$ when $n \rightarrow \infty$.
When $\beta_1 \neq 0$ we get
\begin{align*}
    \lim_{n\rightarrow\infty} \frac{m_1}{\sqrt{\overline{m}_2}}
    & =
    \frac{ \lim_{n\rightarrow\infty} n^{-1/2} \Big(
        P_{1,1}(n) \beta_1^2
        + Q_{1,0}(n) \beta_1 m_\star
        + \tau^2 F_{1}(n)
    \Big)
    }{\sqrt{ \lim_{n\rightarrow\infty} n^{-1} \Big(
        P_{2,1}(n) \beta_1^2 s_\star^2
        + Q_{2,0}(n) \beta_1 m_\star s_\star^2
        + R_{2,-1}(n) m_\star^2 s_\star^2
        + S_{2,0}(n) s_\star^4
    \Big)
    }}
\\
    &= \frac{
        \lim_{n\rightarrow\infty} n^{-1/2} P_{2,1}(n) \beta_1^2 s_\star^2
    }{\sqrt{\beta_1^2 s_\star^2}}
\\
    &= - \infty
\,.
\numberthis
\end{align*}
Otherwise, when $\beta_1 = 0$, we get
\begin{align*}
    \lim_{n\rightarrow\infty} \frac{m_1}{\sqrt{\overline{m}_2}}
    & =
    \frac{ \lim_{n\rightarrow\infty}
        \tau^2 F_{1}(n)
    }{\sqrt{ \lim_{n\rightarrow\infty} \Big(
        R_{2,-1}(n) m_\star^2 s_\star^2
        + S_{2,0}(n) s_\star^4
    \Big)
    }}
\\
    &= \frac{
        \tau^2 \lim_{n\rightarrow\infty} F_{1}(n)
    }{\sqrt{ s_\star^4 \lim_{n\rightarrow\infty} S_{2,0}(n)}}
\\
    &= \frac{
        \tau^2 \frac{1}{2}
    }{\sqrt{ s_\star^4 \frac{1}{2}}}
\\
    &= \frac{\tau^2}{\sqrt{2}s_\star^2}
\,.
\numberthis
\end{align*}
Now we can summarise
\begin{align*}
    \lim_{n\rightarrow\infty} \frac{m_1}{\sqrt{\overline{m}_2}}
    & = \begin{dcases}
        \frac{\tau^2}{\sqrt{2}s_\star^2} \,, &\text{when } \beta_1 = 0 \\
        - \infty & \text{otherwise.}
    \end{dcases}
\numberthis
\label{app_eq_onecov_ninf_loocv}
\end{align*}
This limit matches with the limit of the ratio of the mean and standard deviation of $\elpdC{\Mfd}{\y}$ in Equation~\eqref{app_eq_onecov_ninf_elpd}.


\subsubsection{LOO-CV Error}

In this section, we derive a simplified analytic form for the LOO-CV error presented in Appendix~\ref{app_sec_analytic_error} and for some moments of interest in the one covariate case defined in Appendix~\ref{app_sec_one_cov}. First, we derive the parameters $\+A_\mathrm{err}$, $\+b_\mathrm{err}$, and $c_\mathrm{err}$ defined in Appendix~\ref{app_sec_analytic_error} and then we use them to derive the respective moments of interest defined in Appendix~\ref{app_sec_analytic_moments}.

\paragraph{Parameters}

Using the results from Appendix~\ref{app_sec_one_cov_elpd_params} and~\ref{app_sec_one_cov_loo_params}, we can derive simplified forms for the parameters for the LOO-CV error presented in Appendix~\ref{app_sec_analytic_error} in the one covariate case defined in Appendix~\ref{app_sec_one_cov}:
\begingroup
\allowdisplaybreaks
\begin{align*}
    \+A_{\mathrm{err},1}
    &= \frac{1}{2} \left(
        \+P_{\Mfsa} \+D_{\Mfsa} \+P_{\Mfsa}
        - \widetilde{\+P}_{\Mfsa}^\transp \widetilde{\+D}_{\Mfsa} \widetilde{\+P}_{\Mfsa}
        - \+P_{\Mfsb} \+D_{\Mfsb} \+P_{\Mfsb}
        + \widetilde{\+P}_{\Mfsb}^\transp \widetilde{\+D}_{\Mfsb} \widetilde{\+P}_{\Mfsb}
    \right)
    \\
    &= \frac{1}{2} \bigg(
        \frac{1}{n+1} \ones_{n} \ones_{n}^\transp
        - \frac{n}{n-1} \eye_{n} + \frac{1}{n-1} \ones_{n} \ones_{n}^\transp
        \\ &\qquad\quad
        - \frac{1}{n+2} \left(\ones_{n} \ones_{n}^\transp + \+x \+x^\transp \right)
        + \frac{n}{n-2}\eye_{n} - \frac{1}{n-2}\left(\ones_{n} \ones_{n}^\transp + \+x \+x^\transp \right)
    \bigg)
    \\
    &= \frac{n}{2(n-1)(n-2)} \eye_{n}
        \\ &\quad
        - \frac{3 n}{(n+2)(n+1)(n-1)(n-2)} \ones_{n} \ones_{n}^\transp
        \\ &\quad
        - \frac{n}{(n+2)(n-2)} \+x \+x^\transp
\,,
\numberthis
\\
    \+B_{\mathrm{err},\Mfsa,1}
    &= \+P_{\Mfsa} \+D_{\Mfsa} \left(\+P_{\Mfsa} - \eye \right)
        - \widetilde{\+P}_\Mfsa^\transp \widetilde{\+D}_{\Mfsa} \widetilde{\+P}_\Mfsa
    \\
    &= - \frac{n}{n-1} \eye_{n} + \frac{1}{n-1} \ones_{n} \ones_{n}^\transp
\,,
\numberthis
\\
    \+B_{\mathrm{err},\Mfsb,1}
    &= \+P_{\Mfsb} \+D_{\Mfsb} \left(\+P_{\Mfsb} - \eye \right)
        - \widetilde{\+P}_\Mfsb^\transp \widetilde{\+D}_{\Mfsb} \widetilde{\+P}_\Mfsb
    \\
    &= - \frac{n}{n-2}\eye_{n} +\frac{1}{n-2}\left(\ones_{n} \ones_{n}^\transp + \+x \+x^\transp \right)
\,,
\numberthis
\\
    \+C_{\mathrm{err},\Mfsa,1}
    &= \frac{1}{2}
        \left(
            \left(\+P_{\Mfsa} - \eye \right) \+D_{\Mfsa} \left(\+P_{\Mfsa} - \eye \right)
            - \widetilde{\+P}_\Mfsa^\transp \widetilde{\+D}_{\Mfsa} \widetilde{\+P}_\Mfsa
        \right)
    \\
    &= \frac{1}{2}
        \left(
            \frac{n}{n+1}\eye_{n} - \frac{1}{n+1} \ones_{n} \ones_{n}^\transp
            - \frac{n}{n-1} \eye_{n} + \frac{1}{n-1} \ones_{n} \ones_{n}^\transp
        \right)
    \\
    &= -\frac{n}{(n+1)(n-1)} \eye_{n} + \frac{1}{(n+1)(n-1)} \ones_{n} \ones_{n}^\transp
\,,
\numberthis
\\
    \+C_{\mathrm{err},\Mfsb,1}
    &= \frac{1}{2}
        \left(
            \left(\+P_{\Mfsb} - \eye \right) \+D_{\Mfsb} \left(\+P_{\Mfsb} - \eye \right)
            - \widetilde{\+P}_\Mfsb^\transp \widetilde{\+D}_{\Mfsb} \widetilde{\+P}_\Mfsb
        \right)
    \\
    &= \frac{1}{2}
        \left(
            \frac{n}{n+2}\eye_{n} - \frac{1}{n+2} \left(\ones_{n} \ones_{n}^\transp + \+x \+x^\transp \right)
            - \frac{n}{n-2}\eye_{n} +\frac{1}{n-2}\left(\ones_{n} \ones_{n}^\transp + \+x \+x^\transp \right)
        \right)
    \\
    &= -\frac{2n}{(n+2)(n-2)} \eye_{n} +\frac{2}{(n+2)(n-2)}\left(\ones_{n} \ones_{n}^\transp + \+x \+x^\transp \right)
\,,
\numberthis
\\
    c_{\mathrm{err},4}
    &= \frac{1}{2} \log \left( \prod_{i=1}^n
        \frac{
            \+D_{\Mfsb,[i,i]} \widetilde{\+D}_{\Mfsa[i,i]}
        }{
            \+D_{\Mfsa,[i,i]} \widetilde{\+D}_{\Mfsb[i,i]}
        }
    \right)
    \\
    &= \frac{1}{2} \log \left( \prod_{i=1}^n
        \frac{
            \frac{n}{n+2} \frac{n-1}{n}
        }{
            \frac{n}{n+1} \frac{n-2}{n}
        }
    \right)
    \\
    &= \frac{n}{2} \log \frac{ (n+1)(n-1) }{  (n+2)(n-2) }
\,,
\numberthis
\\
\+C_{\Mfsa,2} &= \left(\+P_{\Mfsa} - \eye \right) \+D_{\Mfsa}
    \\
    &= - \frac{n}{n+1}\eye_{n} + \frac{1}{n+1} \ones_{n} \ones_{n}^\transp
\,,
\numberthis
\\
\+C_{\Mfsb,2} &= \left(\+P_{\Mfsb} - \eye \right) \+D_{\Mfsb}
    \\
    &= - \frac{n}{n+2}\eye_{n} + \frac{1}{n+2} \left(\ones_{n} \ones_{n}^\transp + \+x \+x^\transp \right)
\,,
\numberthis
\\
\+B_{\Mfsd,2} &= \+P_{\Mfsa} \+D_{\Mfsa} - \+P_{\Mfsb} \+D_{\Mfsb}
    \\
    &= \frac{1}{n+1} \ones_{n} \ones_{n}^\transp
        - \frac{1}{n+2} \left(\ones_{n} \ones_{n}^\transp + \+x \+x^\transp \right)
    \\
    &= \frac{1}{(n+2)(n+1)} \ones_{n} \ones_{n}^\transp - \frac{1}{n+2} \+x \+x^\transp
\,,
\numberthis
\\
\+C_{\Mfsd,3} &= - \frac{1}{2} \left( \+D_{\Mfsa} - \+D_{\Mfsb}\right)
    \\
    &= - \frac{n}{2(n+1)(n+2)} \eye_{n}
\,.
\numberthis
\end{align*}
\endgroup
Furthermore, we get
\begingroup
\allowdisplaybreaks
\begin{align*}
     \+A_\mathrm{err} &= \frac{1}{\tau^2} \+A_{\mathrm{err},1}
     \\
     &= \frac{1}{\tau^2} \Bigg(
        + \frac{n}{2(n-1)(n-2)} \eye_{n}
        \\ &\qquad\quad
        - \frac{3 n}{(n+2)(n+1)(n-1)(n-2)} \ones_{n} \ones_{n}^\transp
        \\ &\qquad\quad
        - \frac{n}{(n+2)(n-2)} \+x \+x^\transp
    \Bigg)
\,,
\numberthis
\\
    \+b_\mathrm{err} &= \frac{1}{\tau^2} \Big(
        \+B_{\mathrm{err},\Mfsa,1} \hat{\+y}_{-\Mfsa}
        - \+B_{\mathrm{err},\Mfsb,1} \hat{\+y}_{-\Mfsb}
        - \+B_{\Mfsd,2} \+\mu_\star
   \Big)
   \\
   &= \frac{1}{\tau^2} \left(
        \beta_1 \left( - \frac{n}{n-1} \+x + \frac{1}{n-1} \ones_{n} \ones_{n}^\transp \+x \right)
        - \frac{1}{(n+2)(n+1)} \ones_{n} \ones_{n}^\transp \+\mu_\star + \frac{1}{n+2} \+x \+x^\transp \+\mu_\star
   \right)
   \\
   &= \frac{1}{\tau^2} \left(
       - \beta_1 \frac{n}{n-1} \+x
       - \frac{1}{(n+2)(n+1)} \ones_{n} \ones_{n}^\transp \+\mu_\star +\frac{1}{n+2} \+x \+x^\transp \+\mu_\star
   \right)
\,,
\numberthis
\\
    c_\mathrm{err} &=
    \frac{1}{\tau^2} \Bigg(
        \hat{\+y}_{-\Mfsa}^\transp \+C_{\mathrm{err},\Mfsa,1} \hat{\+y}_{-\Mfsa}
        - \hat{\+y}_{-\Mfsb}^\transp \+C_{\mathrm{err},\Mfsb,1} \hat{\+y}_{-\Mfsb}
        \\&\qquad\quad
        -  \hat{\+y}_{-\Mfsa}^\transp \+C_{\Mfsa,2} \+\mu_{\star}
        +  \hat{\+y}_{-\Mfsb}^\transp \+C_{\Mfsb,2} \+\mu_{\star}
        \\&\qquad\quad
        - \+\mu_{\star}^\transp \+C_{\Mfsd,3} \+\mu_{\star}
        - \+\sigma_\star^\transp \+C_{\Mfsd,3} \+\sigma_\star
    \Bigg)
        + c_{\mathrm{err},4}
    \\
    &= \frac{1}{\tau^2} \Bigg(
        \beta_1^2 \+x^\transp \left(-\frac{n}{(n+1)(n-1)} \eye_{n} + \frac{1}{(n+1)(n-1)} \ones_{n} \ones_{n}^\transp \right) \+x
        \\&\qquad\quad
        + \beta_1 \+x^\transp \left( \frac{n}{n+1}\eye_{n} - \frac{1}{n+1} \ones_{n} \ones_{n}^\transp \right) \+\mu_{\star}
        \\&\qquad\quad
        + \frac{n}{2(n+1)(n+2)} \left(
            \+\mu_{\star}^\transp \+\mu_{\star}
            + \+\sigma_\star^\transp \+\sigma_\star
        \right)
    \Bigg)
    \\&\quad
        + \frac{n}{2} \log \frac{ (n+1)(n-1) }{ (n+2)(n-2) }
    \\
    &= \frac{1}{\tau^2} \left(
        - \beta_1^2 \frac{n^2}{(n+1)(n-1)}
        + \beta_1 \frac{n}{n+1} \+x^\transp \+\mu_{\star}
        + \frac{n}{2(n+1)(n+2)} \left(
            \+\mu_{\star}^\transp \+\mu_{\star}
            + \+\sigma_\star^\transp \+\sigma_\star
        \right)
    \right)
    \\&\quad
        + \frac{n}{2} \log \frac{ (n+1)(n-1) }{ (n+2)(n-2) }
\,.
\numberthis
\end{align*}
\endgroup
Considering the applied setting for the data generation mechanism parameters, in which
\begin{align}
    \+\Sigma_\star &= s_\star^2 \, \eye_{n},
    \\
    x_{i_\text{out}} &= 1,
    \\
    \mu_{\star,\,i} &= \begin{cases} m_\star & \text{when } i = i_\text{out}\,, \\ 0 & \text{otherwise}, \end{cases}
\end{align}
the LOO-CV error parameters $\+b_\mathrm{err}$ and $c_\mathrm{err}$ simplify into
\begin{align}
    \+b_\mathrm{err}
   &= \frac{1}{\tau^2} \left(
       \left(-\beta_1 \frac{n}{n-1} + m_\star \frac{1}{n+2} \right) \+x
       - m_\star \frac{1}{(n+2)(n+1)} \ones_{n}
   \right)
\,,
\\
    c_\mathrm{err}
    &= \frac{1}{\tau^2} \left(
        - \beta_1^2 \frac{n^2}{(n+1)(n-1)}
        + \beta_1 m_\star \frac{n}{n+1}
        + \frac{n}{2(n+1)(n+2)} \left(
            m_\star^2
            + n s_\star^2
        \right)
    \right)
    \nonumber\\&\quad
    + \frac{n}{2} \log \frac{ (n+1)(n-1) }{ (n+2)(n-2) }
\,.
\end{align}

\paragraph{First Moment}
\label{app_sec_analytic_one_cov_1_moment}

In this section, we formulate the first raw moment $m_1$ in Equation~\eqref{app_eq_analytic_1_moment} for the error $\elpdHatErrC{\Mfd}{\y}$ in the one covariate case defined in Appendix~\ref{app_sec_one_cov}.
The trace of $\+\Sigma_\star^{1/2}\+A_\mathrm{err}\+\Sigma_\star^{1/2} = s_\star^2\+A_\mathrm{err}$ simplifies to
\begin{align*}
\operatorname{tr}\left(\+\Sigma_\star^{1/2}\+A_\mathrm{err}\+\Sigma_\star^{1/2}\right)
    &= \frac{s_\star^2}{\tau^2} n \bigg(
        \frac{n}{2(n-1)(n-2)}
        - \frac{3 n}{(n+2)(n+1)(n-1)(n-2)}
        \\&\qquad\quad
        - \frac{n}{(n+2)(n-2)}
    \bigg)
    \\
    &= - \frac{s_\star^2}{\tau^2} \frac{n^2}{2(n+2)(n+1)}
    \,.
    \numberthis
\end{align*}
Furthermore
\begin{align*}
\+b_\mathrm{err}^\transp \+\mu_\star
    &= \frac{1}{\tau^2} \left(
       \left(-\beta_1 \frac{n}{n-1} + m_\star \frac{1}{n+2} \right) m_\star
       - m_\star \frac{1}{(n+2)(n+1)} m_\star
   \right)
   \\
   & = \frac{1}{\tau^2} \left(
       -\beta_1 m_\star \frac{n}{n-1} + m_\star^2 \frac{n}{(n+2)(n+1)}
   \right)
    \,,
    \numberthis
\end{align*}
and
\begin{align*}
\+\mu_\star^\transp \+A_\mathrm{err} \+\mu_\star
    &= \frac{1}{\tau^2} \Bigg(
        \frac{n}{2(n-1)(n-2)} \+\mu_\star^\transp \+\mu_\star
        \\ &\qquad\quad
        - \frac{3 n}{(n+2)(n+1)(n-1)(n-2)} \+\mu_\star^\transp\ones_{n} \ones_{n}^\transp\+\mu_\star
        \\ &\qquad\quad
        - \frac{n}{(n+2)(n-2)} \+\mu_\star^\transp\+x \+x^\transp\+\mu_\star
    \Bigg)
   \\
   & = -\frac{m_\star^2}{\tau^2} \frac{n}{2(n+2)(n+1)}
    \,.
    \numberthis
\end{align*}
Now Equation~\eqref{app_eq_analytic_1_moment} simplifies to
\begin{align*}
m_1
    &= \operatorname{tr}\left(\+\Sigma_\star^{1/2}\+A_\mathrm{err}\+\Sigma_\star^{1/2}\right)
        + c_\mathrm{err}
        + \+b_\mathrm{err}^\transp \+\mu_\star
        + \+\mu_\star^\transp \+A_\mathrm{err} \+\mu_\star
    \\
    \\
    & = \frac{1}{\tau^2} \left(
        P_{1,0}(n) \beta_1^2
        + Q_{1,-1}(n) \beta_1 m_\star
        + R_{1,-1}(n) m_\star^2
    \right)
    + F_{1}(n)
    \,,
    \numberthis
\intertext{where}
    P_{1,0}(n) &= -\frac{n^2}{(n+1)(n-1)}
    \numberthis\\
    Q_{1,-1}(n) &= -\frac{2n}{(n+1)(n-1)}
    \numberthis\\
    R_{1,-1}(n) &= \frac{n}{(n+2)(n+1)}
    \numberthis\\
    F_{1}(n) &= \frac{n}{2} \log \frac{ (n+1)(n-1) }{ (n+2)(n-2) }
    \,,
    \numberthis
\end{align*}
where the first subscript indicates the corresponding order of the moment, and for the rational functions $P_{1,0}$, $Q_{1,-1}$, and $R_{1,-1}$, the second subscript indicates the degree of the rational as a difference between the degrees of the numerator and the denominator. It can be seen that $m_1$ does not depend on $s_\star$.

\paragraph{Second Moment}
\label{app_sec_analytic_one_cov_2_moment}

In this section, we formulate the second moment $\overline{m}_2$ about the mean in Equation~\eqref{app_eq_analytic_2_moment} for the error $\elpdHatErrC{\Mfd}{\y}$ in the one covariate case defined in Appendix~\ref{app_sec_one_cov}.
The second power of $\+A_\mathrm{err}$ is
\begin{align*}
\+A_\mathrm{err}^2
    &= \frac{1}{\tau^4} \Bigg(
        \frac{n^2}{4(n-1)^2(n-2)^2}\eye_{n}
        \\&\qquad\quad
        - \frac{ 3 n^2 (n^2 + 2)}{(n + 2)^2 (n + 1)^2 (n - 1)^2 (n - 2)^2} \ones_{n} \ones_{n}^\transp
        \\&\qquad\quad
        + \frac{n^2 (n^2 - 2 n - 2)}{(n + 2)^2 (n - 1)  (n - 2)^2} \+x\+x^\transp
    \Bigg)
    \,.
    \numberthis
\end{align*}
The trace in Equation~\eqref{app_eq_analytic_2_moment} simplifies to
\begin{align*}
\operatorname{tr}\left(\left(\+\Sigma_\star^{1/2}\+A_\mathrm{err}\+\Sigma_\star^{1/2}\right)^2\right)
    &= \frac{s_\star^4}{\tau^4} n \bigg(
        \frac{n^2}{4(n-1)^2(n-2)^2}
        - \frac{ 3 n^2 (n^2 + 2)}{(n + 2)^2 (n + 1)^2 (n - 1)^2 (n - 2)^2}
        \\&\qquad\qquad
        + \frac{n^2 (n^2 - 2 n - 2)}{(n + 2)^2 (n - 1)  (n - 2)^2}
    \bigg)
    \\
    &= \frac{s_\star^4}{\tau^4} \frac{n^3 (4 n^3 + 9 n^2 + 5 n - 6)}{4 (n + 2)^2 (n + 1)^2 (n - 1) (n - 2)}
    \,.
    \numberthis
\end{align*}
Furthermore
\begin{align*}
\+b_\mathrm{err}^\transp \+b_\mathrm{err}
    &= \frac{1}{\tau^4} \Bigg(
       \left(- \beta_1 \frac{n}{n-1} + m_\star \frac{1}{n+2} \right)^2 \+x^\transp\+x
       \\&\qquad\quad
       + m_\star^2 \frac{1}{(n+2)^2(n+1)^2} \ones_{n}^\transp\ones_{n}
       \\&\qquad\quad
       -2 \left(-\beta_1 \frac{n}{n-1} + m_\star \frac{1}{n+2} \right) m_\star \frac{1}{(n+2)(n+1)} \+x^\transp\ones_{n}
   \Bigg)
   \\
   & = \frac{1}{\tau^4} \left(
        \beta_1^2 \frac{n^3}{(n-1)^2}
        - \beta_1 m_\star \frac{2n^2}{(n+2)(n-1)}
        + m_\star^2 \frac{n (n^2 + 2 n + 2)}{(n + 2)^2 (n + 1)^2}
   \right)
\,,
\numberthis
\end{align*}
and
\begin{align*}
\+b_\mathrm{err}^\transp \+A_\mathrm{err} \+\mu_\star
    &= \frac{1}{\tau^4} \Bigg(
        \left(-\beta_1 \frac{n}{n-1} + m_\star \frac{1}{n+2} \right) \frac{n}{2(n-1)(n-2)} \+x^\transp \eye_{n} \+\mu_\star
        \\&\qquad\quad
        - \left(-\beta_1 \frac{n}{n-1} + m_\star \frac{1}{n+2} \right) \frac{n}{(n+2)(n-2)} \+x^\transp \+x \+x^\transp \+\mu_\star
        \\&\qquad\quad
        - m_\star \frac{1}{(n+2)(n+1)} \frac{n}{2(n-1)(n-2)} \ones_{n}^\transp \eye_{n} \+\mu_\star
        \\&\qquad\quad
        + m_\star \frac{1}{(n+2)(n+1)} \frac{3 n}{(n+2)(n+1)(n-1)(n-2)} \ones_{n}^\transp \ones_{n} \ones_{n}^\transp \+\mu_\star
    \Bigg)
    \\
    &= \frac{1}{\tau^4} \Bigg(
        \beta_1 m_\star \frac{n^2 (2 n + 1)}{2 (n + 2) (n - 1)^2}
        - m_\star^2 \frac{n^2 (2 n^2 + 5 n + 5)}{2 (n + 2)^2 (n + 1)^2 (n - 1)}
    \Bigg)
\,,
\numberthis
\end{align*}
and
\begin{align*}
\+\mu_\star^\transp \+A_\mathrm{err}^2 \+\mu_\star
    &=
    \frac{1}{\tau^4} \Bigg(
        \frac{n^2}{4(n-1)^2(n-2)^2} \+\mu_\star^\transp \eye_{n} \+\mu_\star
        \\&\qquad\quad
        - \frac{ 3 n^2 (n^2 + 2)}{(n + 2)^2 (n + 1)^2 (n - 1)^2 (n - 2)^2} \+\mu_\star^\transp \ones_{n} \ones_{n}^\transp \+\mu_\star
        \\&\qquad\quad
        + \frac{n^2 (n^2 - 2 n - 2)}{(n + 2)^2 (n - 1)  (n - 2)^2} \+\mu_\star^\transp \+x\+x^\transp \+\mu_\star
    \Bigg)
    \\
    &= \frac{1}{\tau^4} m_\star^2 \frac{n^2 (4 n^3 + 9 n^2 + 5 n - 6)}{4 (n + 2)^2 (n + 1)^2 (n - 1) (n - 2)}
\,.
\numberthis
\end{align*}
Now Equation~\eqref{app_eq_analytic_2_moment} simplifies to
\begin{align*}
\overline{m}_2
    &= 2 \operatorname{tr}\left(\left(\+\Sigma_\star^{1/2}\+A_\mathrm{err}\+\Sigma_\star^{1/2}\right)^2\right)
    + \+b_\mathrm{err}^\transp \+\Sigma_\star \+b_\mathrm{err}
    + 4 \+b_\mathrm{err}^\transp \+\Sigma_\star \+A_\mathrm{err} \+\mu_\star
    + 4 \+\mu_\star^\transp \+A_\mathrm{err} \+\Sigma_\star \+A_\mathrm{err} \+\mu_\star
    \\
    & = \frac{1}{\tau^4} \left(
        P_{2,1}(n) \beta_1^2 s_\star^2
        + Q_{2,0}(n) \beta_1 m_\star s_\star^2
        + R_{2,-1}(n) m_\star^2 s_\star^2
        + S_{2,0}(n) s_\star^4
    \right)
    \,,
    \numberthis
\intertext{where}
    P_{2,1}(n) &= \frac{n^3}{(n-1)^2}
    \numberthis\\
    Q_{2,0}(n) &= \frac{2 n^2}{(n-1)^2}
    \numberthis\\
    R_{2,-1}(n) &= \frac{n}{(n - 1) (n - 2)}
    \numberthis\\
    S_{2,0}(n) &= \frac{n^3 (4 n^3 + 9 n^2 + 5 n - 6)}{2 (n + 2)^2 (n + 1)^2 (n - 1) (n - 2)}
    \,,
    \numberthis
\end{align*}
where the first subscript in the rational functions $P_{2,1}$, $Q_{2,0}$, $R_{2,-1}$, and $S_{2,0}$ indicates the corresponding order of the moment. The second subscript indicates the degree of the rational as a difference between the degrees of the numerator and the denominator.

\paragraph{Third Moment}
\label{app_sec_analytic_one_cov_3_moment}

In this section, we formulate the third moment $\overline{m}_3$ about the mean in Equation~\eqref{app_eq_analytic_2_moment} for the error $\elpdHatErrC{\Mfd}{\y}$ in the one covariate case defined in Appendix~\ref{app_sec_one_cov}.
The third power of $\+A_\mathrm{err}$ is
\begin{align*}
\+A_\mathrm{err}^3
    &= \frac{1}{\tau^6} \Bigg(
        \frac{n^3}{8(n-1)^3(n-2)^3} \eye_{n}
        \\&\qquad\quad
        - \frac{9 n^3 (n^4 + 7 n^2 + 4)}{4 (n + 2)^3 (n + 1)^3 (n - 1)^3 (n - 2)^3} \ones_{n} \ones_{n}^\transp
        \\&\qquad\quad
        - \frac{n^3 (4 n^4 - 14 n^3 + n^2 + 24 n + 12)}{4 (n + 2)^3 (n - 1)^2 (n - 2)^3} \+x\+x^\transp
    \Bigg)
    \,.
    \numberthis
\end{align*}
The trace in Equation~\eqref{app_eq_analytic_3_moment} simplifies to
\begin{align*}
\operatorname{tr}\left(\left(\+\Sigma_\star^{1/2}\+A_\mathrm{err}\+\Sigma_\star^{1/2}\right)^3\right)
    &= \frac{s_\star^6}{\tau^6} n \bigg(
        \frac{n^3}{8(n-1)^3(n-2)^3}
        - \frac{9 n^3 (n^4 + 7 n^2 + 4)}{4 (n + 2)^3 (n + 1)^3 (n - 1)^3 (n - 2)^3}
        \\&\qquad\qquad
        - \frac{n^3 (4 n^4 - 14 n^3 + n^2 + 24 n + 12)}{4 (n + 2)^3 (n - 1)^2 (n - 2)^3}
    \bigg)
    \\
    &= -\frac{s_\star^6}{\tau^6} \frac{n^4 (8 n^6 + 12 n^5 - 35 n^4 - 102 n^3 - 83 n^2 - 36 n + 20)}{8 (n + 2)^3 (n + 1)^3 (n - 1)^2 (n - 2)^2}
    \,.
    \numberthis
\end{align*}
Furthermore
\begingroup
\allowdisplaybreaks
\begin{align*}
\+b_\mathrm{err}^\transp \+A_\mathrm{err} \+b_\mathrm{err}
   &= \frac{1}{\tau^6} \Bigg(
        \left(-\beta_1 \frac{n}{n-1} + m_\star \frac{1}{n+2} \right)^2
        \\ &\qquad\quad
        \left(
                + \frac{n}{2(n-1)(n-2)} \+x^\transp \eye_{n} \+x
                - \frac{n}{(n+2)(n-2)} \+x^\transp \+x \+x^\transp \+x
        \right)
        \\ &\qquad\quad
        + m_\star^2 \frac{1}{(n+2)^2(n+1)^2}
        \\ &\qquad\quad
        \left(
            + \frac{n}{2(n-1)(n-2)} \ones_{n}^\transp \eye_{n} \ones_{n}
            - \frac{3 n}{(n+2)(n+1)(n-1)(n-2)} \ones_{n}^\transp \ones_{n} \ones_{n}^\transp \ones_{n}
        \right)
    \Bigg)
    \\
    &= \frac{1}{\tau^6} \Bigg(
        - \beta_1^2  \frac{n^4 (2 n + 1)}{2 (n + 2) (n - 1)^3}
        + \beta_1 m_\star \frac{n^3 (2 n + 1)}{(n + 2)^2 (n - 1)^2}
        \\ &\qquad\quad
        - m_\star^2 \frac{n^2 (2 n^4 + 7 n^3 + 9 n^2 + 4 n + 2)}{2 (n + 2)^3 (n + 1)^3 (n - 1)}
    \Bigg)
\,,
\numberthis
\end{align*}
\endgroup
and
\begin{align*}
\+b_\mathrm{err}^\transp \+A_\mathrm{err}^2 \+\mu_\star
    &= \frac{1}{\tau^6} \Bigg(
        \left(-\beta_1 \frac{n}{n-1} + m_\star \frac{1}{n+2} \right)
        \\&\qquad\quad
        \left(
            \frac{n^2}{4(n-1)^2(n-2)^2} \+x^\transp \eye_{n} \+\mu_\star
            + \frac{n^2 (n^2 - 2 n - 2)}{(n + 2)^2 (n - 1)  (n - 2)^2}  \+x^\transp\+x\+x^\transp\+\mu_\star
        \right)
        \\&\qquad\quad
        - m_\star \frac{1}{(n+2)(n+1)}
        \\&\qquad\quad
        \left(
            \frac{n^2}{4(n-1)^2(n-2)^2} \ones_{n}^\transp \eye_{n} \+\mu_\star
            - \frac{ 3 n^2 (n^2 + 2)}{(n + 2)^2 (n + 1)^2 (n - 1)^2 (n - 2)^2} \ones_{n}^\transp \ones_{n} \ones_{n}^\transp\+\mu_\star
        \right)
    \Bigg)
    \\
    &= \frac{1}{\tau^6} \Bigg(
        - \beta_1 m_\star \frac{n^3 (2 n + 1)^2}{4 (n + 2)^2 (n - 1)^3}
        + m_\star^2 \frac{n^3 (4 n^4 + 16 n^3 + 25 n^2 + 18 n + 9)}{4 (n + 2)^3 (n + 1)^3 (n - 1)^2}
    \Bigg)
\,,
\numberthis
\end{align*}
and
\begin{align*}
\+\mu_\star^\transp \+A_\mathrm{err}^3 \+\mu_\star
    &= \frac{1}{\tau^6} \Bigg(
        \frac{n^3}{8(n-1)^3(n-2)^3} \+\mu_\star^\transp \eye_{n} \+\mu_\star
        \\&\qquad\quad
        - \frac{9 n^3 (n^4 + 7 n^2 + 4)}{4 (n + 2)^3 (n + 1)^3 (n - 1)^3 (n - 2)^3} \+\mu_\star^\transp \ones_{n} \ones_{n}^\transp \+\mu_\star
        \\&\qquad\quad
        - \frac{n^3 (4 n^4 - 14 n^3 + n^2 + 24 n + 12)}{4 (n + 2)^3 (n - 1)^2 (n - 2)^3} \+\mu_\star^\transp \+x\+x^\transp \+\mu_\star
    \Bigg)
    \\
    &= -\frac{1}{\tau^6} m_\star^2
        \frac{n^3 (8 n^6 + 12 n^5 - 35 n^4 - 102 n^3 - 83 n^2 - 36 n + 20)}{8 (n + 1)^3 (n + 2)^3 (n - 1)^2 (n - 2)^2}
\,.
\numberthis
\end{align*}
Now Equation~\eqref{app_eq_analytic_3_moment} simplifies to
\begin{align*}
\overline{m}_3
    &= 8 \operatorname{tr}\left(\left(\+\Sigma_\star^{1/2}\+A_\mathrm{err}\+\Sigma_\star^{1/2}\right)^3\right)
    + 6 \+b_\mathrm{err}^\transp \+\Sigma_\star \+A_\mathrm{err} \+\Sigma_\star \+b_\mathrm{err}
    \\ &\quad
    + 24 \+b_\mathrm{err}^\transp \+\Sigma_\star \+A_\mathrm{err} \+\Sigma_\star \+A_\mathrm{err} \+\mu_\star
    + 24 \+\mu_\star^\transp \+A_\mathrm{err} \+\Sigma_\star \+A_\mathrm{err} \+\Sigma_\star \+A_\mathrm{err} \+\mu_\star
    \\
    & = \frac{1}{\tau^6} \left(
        P_{3,1}(n) \beta_1^2 s_\star^4
        + Q_{3,0}(n) \beta_1 m_\star s_\star^4
        + R_{3,-1}(n) m_\star^2 s_\star^4
        + S_{3,0}(n) s_\star^6
    \right)
    \,,
    \numberthis
\intertext{where}
    P_{3,1}(n) &= -\frac{3 n^4 (2 n + 1)}{(n + 2) (n - 1)^3}
    \numberthis\\
    Q_{3,0}(n) &= -\frac{6 n^3 (2 n + 1)}{(n + 2) (n - 1)^3}
    \numberthis\\
    R_{3,-1}(n) &= -\frac{3 n^2 (2 n^2 - 5 n - 2)}{(n - 2)^2 (n - 1)^2 (n + 2)}
    \numberthis\\
    S_{3,0}(n) &= -\frac{n^4 (8 n^6 + 12 n^5 - 35 n^4 - 102 n^3 - 83 n^2 - 36 n + 20)}{(n + 2)^3 (n + 1)^3 (n - 1)^2 (n - 2)^2}
    \,,
    \numberthis
\end{align*}
where the first subscript in the rational functions $P_{3,1}$, $Q_{3,0}$, $R_{3,-1}$, and $S_{3,0}$ indicates the corresponding order of the moment. The second subscript indicates the degree of the rational as a difference between the degrees of the numerator and the denominator.

\paragraph{Mean Relative to the Standard Deviation}
\label{app_sec_analytic_three_Rel_SD}

In this section, we formulate the ratio of mean and standard deviation $m_1 \big/ \sqrt{\overline{m}_2}$ for the error $\elpdHatErrC{\Mfd}{\y}$ in the one covariate case defined in Appendix~\ref{app_sec_one_cov}.
Combining results from appendices~\ref{app_sec_analytic_one_cov_1_moment} and~\ref{app_sec_analytic_one_cov_2_moment}, we get
\begingroup
\allowdisplaybreaks
\begin{align*}
    \frac{m_1}{\sqrt{\overline{m}_2}}
    & =
    \frac{
        P_{1,0}(n) \beta_1^2
        + Q_{1,-1}(n) \beta_1 m_\star
        + R_{1,-1}(n) m_\star^2
        + \tau^2 F_{1}(n)
    }{\sqrt{
        P_{2,1}(n) \beta_1^2 s_\star^2
        + Q_{2,0}(n) \beta_1 m_\star s_\star^2
        + R_{2,-1}(n) m_\star^2 s_\star^2
        + S_{2,0}(n) s_\star^4
    }}
\,,
\numberthis
\intertext{where}
    P_{1,0}(n) &= -\frac{n^2}{(n+1)(n-1)}
    \numberthis\\
    Q_{1,-1}(n) &= -\frac{2n}{(n+1)(n-1)}
    \numberthis\\
    R_{1,-1}(n) &= \frac{n}{(n+2)(n+1)}
    \numberthis\\
    F_{1}(n) &= \frac{n}{2} \log \frac{ (n+1)(n-1) }{ (n+2)(n-2) }
    \numberthis\\
    P_{2,1}(n) &= \frac{n^3}{(n-1)^2}
    \numberthis\\
    Q_{2,0}(n) &= \frac{2 n^2}{(n-1)^2}
    \numberthis\\
    R_{2,-1}(n) &= \frac{n}{(n - 1) (n - 2)}
    \numberthis\\
    S_{2,0}(n) &= \frac{n^3 (4 n^3 + 9 n^2 + 5 n - 6)}{2 (n + 2)^2 (n + 1)^2 (n - 1) (n - 2)}
\,,
\numberthis
\end{align*}
\endgroup
where the first subscript in the rational functions $P$, $Q$, $R$, and $S$ indicates the corresponding order of the moment, and the second subscript indicates the degree of the rational as a difference between the degrees of the numerator and the denominator.

Let us inspect the behaviour of $m_1 \big/ \sqrt{\overline{m}_2}$ when $n \rightarrow \infty$. When $\beta_1 \neq 0$, by multiplying numerator and denominator in $m_1 \big/ \sqrt{\overline{m}_2}$ by $n^{-1/2}$, we get
\begin{align*}
    \lim_{n\rightarrow\infty} \frac{m_1}{\sqrt{\overline{m}_2}}
    & =
    \frac{
        \lim_{n\rightarrow\infty} n^{-1/2} \Big(
        P_{1,0}(n) \beta_1^2
        + Q_{1,-1}(n) \beta_1 m_\star
        + R_{1,-1}(n) m_\star^2
        + F_{1}(n) \tau^2
        \Big)
    }{\sqrt{
        \lim_{n\rightarrow\infty} n^{-1} \Big(
        P_{2,1}(n) \beta_1^2 s_\star^2
        + Q_{2,0}(n) \beta_1 m_\star s_\star^2
        + R_{2,-1}(n) m_\star^2 s_\star^2
        + S_{2,0}(n) s_\star^4
        \Big)
    }}
\\
    &= \frac{
        \lim_{n\rightarrow\infty} n^{-1/2} F_{1}(n) \tau^2
    }{\sqrt{
        \lim_{n\rightarrow\infty} n^{-1} P_{2,1}(n) \beta_1^2 s_\star^2
    }}
\\
    &= \frac{
        0 \tau^2
    }{\sqrt{ \beta_1^2 s_\star^2 }}
\\
    &= 0.
    \numberthis
\end{align*}
Similarly, when $\beta_1 = 0$, we get
\begin{align*}
    \lim_{n\rightarrow\infty} \frac{m_1}{\sqrt{\overline{m}_2}}
    & =
    \frac{
        \lim_{n\rightarrow\infty} \Big(
        R_{1,-1}(n) m_\star^2
        + F_{1}(n) \tau^2
        \Big)
    }{\sqrt{
        \lim_{n\rightarrow\infty} \Big(
        R_{2,-1}(n) m_\star^2 s_\star^2
        + S_{2,0}(n) s_\star^4
        \Big)
    }}
\\
    &= \frac{
        0 m_\star^2 + 0 \tau^2
    }{\sqrt{
        0 m_\star^2 s_\star^2 + 2 s_\star^4
    }}
    \\
    &= 0
    \numberthis
\,.
\end{align*}
Now we can summarise
\begin{align}
    \lim_{n\rightarrow\infty} \frac{m_1}{\sqrt{\overline{m}_2}} &= 0.
    \label{app_eq_zscore_lim_n}
\end{align}
\paragraph{Skewness}
\label{app_sec_analytic_three_skew}
In this section, we formulate the skewness $\widetilde{m}_3 = \overline{m}_3 \big/ (\overline{m}_2)^{3/2}$ in Equation~\eqref{app_eq_analytic_3_moment_skew} for the error $\elpdHatErrC{\Mfd}{\y}$ in the one covariate case defined in Appendix~\ref{app_sec_one_cov}. Combining results from appendices~\ref{app_sec_analytic_one_cov_2_moment} and~\ref{app_sec_analytic_one_cov_3_moment}, we get
\begingroup
\allowdisplaybreaks
\begin{align*}
    \widetilde{m}_3 &= \overline{m}_3 \big/ (\overline{m}_2)^{3/2}
\\
    &= \frac{
        P_{3,1}(n) \beta_1^2 s_\star^4
        + Q_{3,0}(n) \beta_1 m_\star s_\star^4
        + R_{3,-1}(n) m_\star^2 s_\star^4
        + S_{3,0}(n) s_\star^6
    }{\Big(
        P_{2,1}(n) \beta_1^2 s_\star^2
        + Q_{2,0}(n) \beta_1 m_\star s_\star^2
        + R_{2,-1}(n) m_\star^2 s_\star^2
        + S_{2,0}(n) s_\star^4
    \Big)^{3/2}}
\,,
\numberthis
\intertext{where}
    P_{2,1}(n) &= \frac{n^3}{(n-1)^2}
    \numberthis\\
    Q_{2,0}(n) &= \frac{2 n^2}{(n-1)^2}
    \numberthis\\
    R_{2,-1}(n) &= \frac{n}{(n - 1) (n - 2)}
    \numberthis\\
    S_{2,0}(n) &= \frac{n^3 (4 n^3 + 9 n^2 + 5 n - 6)}{2 (n + 2)^2 (n + 1)^2 (n - 1) (n - 2)}
    \numberthis\\
    P_{3,1}(n) &= -\frac{3 n^4 (2 n + 1)}{(n + 2) (n - 1)^3}
    \numberthis\\
    Q_{3,0}(n) &= -\frac{6 n^3 (2 n + 1)}{(n + 2) (n - 1)^3}
    \numberthis\\
    R_{3,-1}(n) &= -\frac{3 n^2 (2 n^2 - 5 n - 2)}{(n - 2)^2 (n - 1)^2 (n + 2)}
    \numberthis\\
    S_{3,0}(n) &= -\frac{n^4 (8 n^6 + 12 n^5 - 35 n^4 - 102 n^3 - 83 n^2 - 36 n + 20)}{(n + 2)^3 (n + 1)^3 (n - 1)^2 (n - 2)^2}
\,,
\numberthis
\end{align*}
\endgroup
where the first subscript in the rational functions $P$, $Q$, $R$, and $S$ indicates the corresponding order of the moment, and the second subscript indicates the degree of the rational as a difference between the degrees of the numerator and the denominator. It can be seen that $\tau$ does not affect the skewness.

Let us inspect the behaviour of $\widetilde{m}_3$ when $n \rightarrow \infty$. When $\beta_1 \neq 0$, by multiplying numerator and denominator in $\widetilde{m}_3$ by $n^{-3/2}$, we get
\begin{align*}
    \lim_{n\rightarrow\infty} \widetilde{m}_3 &= \frac{
        \lim_{n\rightarrow\infty} n^{-3/2}\Big(
        P_{3,1}(n) \beta_1^2 s_\star^4
        + Q_{3,0}(n) \beta_1 m_\star s_\star^4
        + R_{3,-1}(n) m_\star^2 s_\star^4
        + S_{3,0}(n) s_\star^6
        \Big)
    }{\Big(
        \lim_{n\rightarrow\infty} n^{-1}\Big(
        P_{2,1}(n) \beta_1^2 s_\star^2
        + Q_{2,0}(n) \beta_1 m_\star s_\star^2
        + R_{2,-1}(n) m_\star^2 s_\star^2
        + S_{2,0}(n) s_\star^4
        \Big)
    \Big)^{3/2}}
    \\
    &= \frac{
        0
    }{\Big(
        \lim_{n\rightarrow\infty} n^{-1} P_{2,1}(n) \beta_1^2 s_\star^2
    \Big)^{3/2}}
    \\
    &= \frac{
        0
    }{\left( \beta_1^2 s_\star^2 \right)^{3/2}}
    \\
    &= 0
    \numberthis
    \label{app_eq_skewness_lim_n_bn0}
\,.
\end{align*}
When $\beta_1 = 0$, we get
\begin{align*}
    \lim_{n\rightarrow\infty} \widetilde{m}_3 &= \frac{
        \lim_{n\rightarrow\infty} \Big(
        R_{3,-1}(n) m_\star^2 s_\star^4
        + S_{3,0}(n) s_\star^6
        \Big)
    }{\Big(
        \lim_{n\rightarrow\infty} \Big(
        R_{2,-1}(n) m_\star^2 s_\star^2
        + S_{2,0}(n) s_\star^4
        \Big)
    \Big)^{3/2}}
    \\
    &= \frac{
        0 m_\star^2 s_\star^4 - 8 s_\star^6
    }{\big(
        0 m_\star^2 s_\star^2 + 2 s_\star^4
    \big)^{3/2}}
    \\
    &= - 2^{3/2}
    \numberthis
    \label{app_eq_skewness_lim_n_b0}
\,.
\end{align*}
Now we can summarise
\begin{align}
    \lim_{n\rightarrow\infty} \widetilde{m}_3 &=
    \begin{dcases}
        - 2^{3/2}, & \text{when }  \beta_1 = 0 \\
        0, & \text{otherwise.}
    \end{dcases}
    \label{app_eq_skewness_lim_n}
\end{align}
It can be seen that the limit does not depend on $m_\star$ or $s_\star$.

Next, similar to the analyses conducted in appendices~\ref{app_sec_analytic_effect_of_model_difference}, \ref{app_sec_analytic_effect_of_outliers}, and~\ref{app_sec_analytic_effect_of_data_variance}, we analyse the behaviour of the skewness as a function of $\beta_1$, $m_\star$, and $s_\star$. Analogous to Equation~\eqref{app_eq_skewness_lim_n_bn0}, inspecting the behaviour of the skewness $\widetilde{m}_3$ as a function of $\beta_1$ gives
\begin{align*}
    \lim_{\beta_1\rightarrow \pm \infty} \widetilde{m}_3 &= \frac{
        \lim_{\beta_1\rightarrow \pm \infty} \beta_1^{-3}\Big(
        P_{3,1}(n) \beta_1^2 s_\star^4
        + Q_{3,0}(n) \beta_1 m_\star s_\star^4
        + R_{3,-1}(n) m_\star^2 s_\star^4
        + S_{3,0}(n) s_\star^6
        \Big)
    }{\Big(
        \lim_{\beta_1\rightarrow \pm \infty} \beta_1^{-2}\Big(
        P_{2,1}(n) \beta_1^2 s_\star^2
        + Q_{2,0}(n) \beta_1 m_\star s_\star^2
        + R_{2,-1}(n) m_\star^2 s_\star^2
        + S_{2,0}(n) s_\star^4
        \Big)
    \Big)^{3/2}}
    \\
    &= \frac{
        0
    }{\left( P_{2,1}(n) s_\star^2 \right)^{3/2}}
    \\
    &= 0
    \numberthis
    \label{app_eq_skewness_lim_beta}
\,.
\end{align*}
Similarly, as a function of $m_\star$, it can be seen that
\begin{align}
    \lim_{m_\star \rightarrow \pm \infty} \widetilde{m}_3 = \frac{
        0
    }{\left( R_{2,-1}(n) s_\star^2 \right)^{3/2}}
    = 0
    \label{app_eq_skewness_lim_mstar}
\,.
\end{align}
As a function of $s_\star$, we get
\begin{align*}
    \lim_{s_\star\rightarrow \infty} \widetilde{m}_3 &= \frac{
        \lim_{s_\star\rightarrow \infty} s_\star^{-6}\Big(
        P_{3,1}(n) \beta_1^2 s_\star^4
        + Q_{3,0}(n) \beta_1 m_\star s_\star^4
        + R_{3,-1}(n) m_\star^2 s_\star^4
        + S_{3,0}(n) s_\star^6
        \Big)
    }{\Big(
        \lim_{s_\star\rightarrow \infty} s_\star^{-4}\Big(
        P_{2,1}(n) \beta_1^2 s_\star^2
        + Q_{2,0}(n) \beta_1 m_\star s_\star^2
        + R_{2,-1}(n) m_\star^2 s_\star^2
        + S_{2,0}(n) s_\star^4
        \Big)
    \Big)^{3/2}}
    \\
    &= \frac{
        S_{3,0}(n)
    }{S_{2,0}(n)^{3/2}}
    \\
    &= - 2^{3/2} \frac{
         8 n^6 + 12 n^5 - 35 n^4 - 102 n^3 - 83 n^2 - 36 n + 20
    }{
        \sqrt{n\left(n^2 - 3 n + 2\right)} \left(4 n^3 + 9 n^2 + 5 n - 6\right)^{3/2}
    }
    \numberthis
    \label{app_eq_skewness_lim_sstar}
\,,
\end{align*}
which approaches the same limit $-2^{3/2}$ from below, when $n \rightarrow \infty$. These limits match with the results obtained in appendices~\ref{app_sec_analytic_effect_of_model_difference}, \ref{app_sec_analytic_effect_of_outliers}, and~\ref{app_sec_analytic_effect_of_data_variance}.
%
%
\section{Additional Results for the Simulated Experiment}
\label{app_sec_additional_experiment_results}
In this appendix, we present some additional results for the simulated linear regression model comparison experiment discussed in Section~\ref{sec_experiments}. Among others, these results illustrate the effect of an outlier in more detail. The outlier observation has a deviated mean of 20 times the standard deviation of $\y_i$ in all experiments.
\begin{figure}[tb!] 
  \centering
  \includegraphics[width=0.85\figurecontrolwidth]{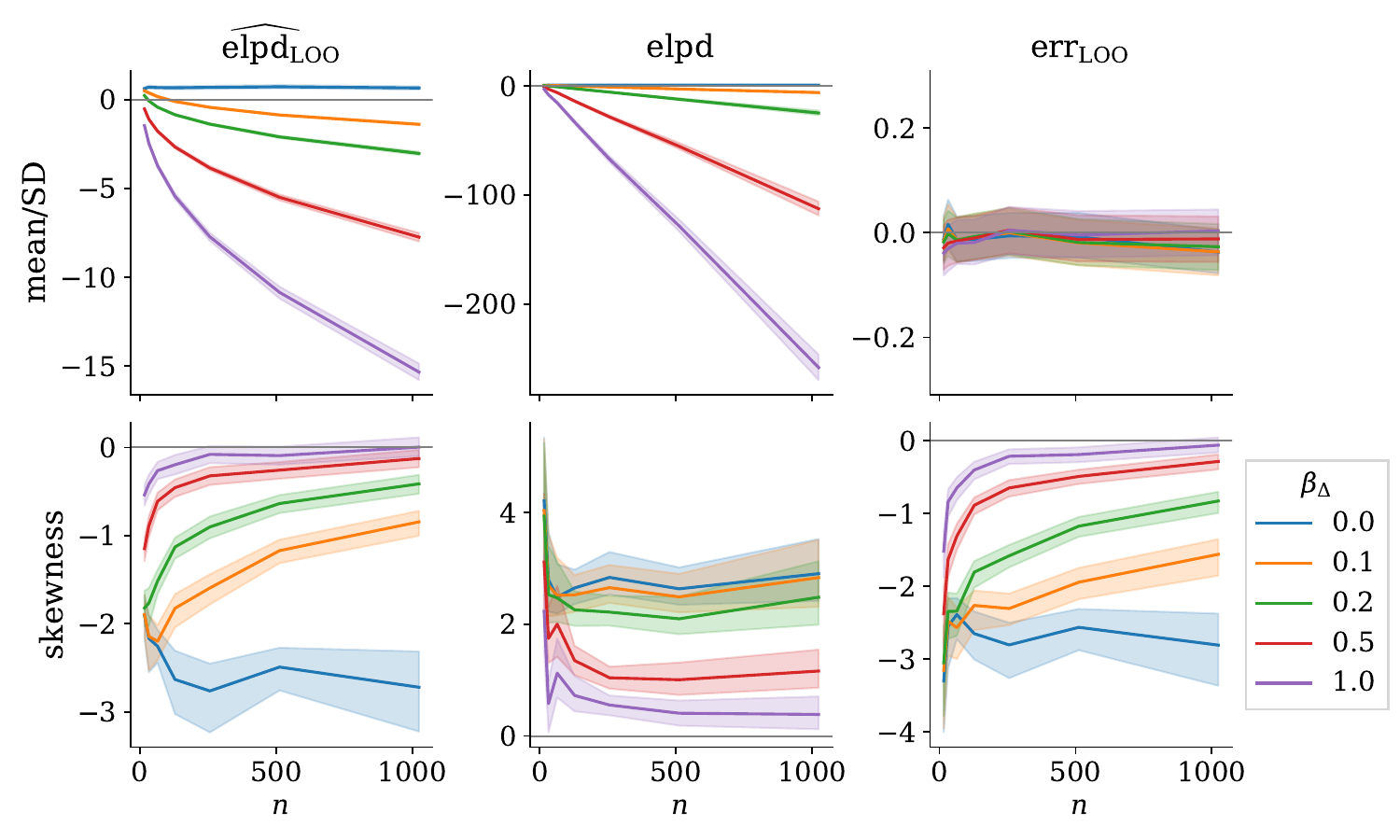}
  \caption{%
    Illustration of the estimated mean relative to the standard deviation and skewness for $\elpdHatC{\Md}{\y}$, $\elpdC{\Md}{\y}$, and for the error $\elpdHatErrC{\Md}{\y}$ as a function of the data size $n$ for various non-shared covariate effects $\beta_\Delta$. The solid lines correspond to the median and the shaded area to the 95 \% confidence interval from the Bayesian bootstrap (BB) sample of size 2000 using the weighted moment estimators presented by \citet{Rimoldini_2013_weighted_moments}. As the effect of $\beta_\Delta$ is symmetric, the problem is simulated only with positive $\beta_\Delta$. Similar behaviour can be observed in Figure~\ref{fig_analytic_zscore_skew_n_b} for analogous experiment conditional for the design matrix $X$ and model variance $\tau^2$. In this case, however, while not greatly affecting the skewness of the error $\elpdHatErrC{\Md}{\y}$, the skewness of the $\elpdC{\Md}{\y}$ decreases when $\beta_\Delta$ grows.
  }\label{fig_moments_n_b}
\end{figure}

Figure~\ref{fig_moments_n_b} illustrates the relative mean and skewness for the sampling distribution $\elpdHatC{\Md}{\y}$, for the distribution of the estimand $\elpdC{\Md}{\y}$, and for the error distribution $\elpdHatErrC{\Md}{\y}$ estimated from the simulated experiments as a function of the data size $n$ for different non-shared covariates' effects $\beta_\Delta$. These results indicate that the moments behave quite similarly as in the analysis conditional on the design matrix $\+X$ and model variance $\tau$ in Section~\ref{sec_analytic_case}. 
Similar to the situation with conditionalised design matrix $\+X$ and model variance $\tau$, it can be seen from the figure that when the non-shared covariate effect $\beta_\Delta$ grows, the difference in the predictive performance grows and the LOO-CV method becomes more likely to pick the correct model. Similar behaviour can be observed when the data size $n$ grows and $|\beta_\Delta|>0$. However, when $\beta_\Delta=0$, the difference in the predictive performance stays zero, and the LOO-CV method is slightly more likely to pick the simpler model regardless of $n$. The relative mean of the error confirms that the bias of the LOO-CV estimator is relatively small with all applied $n$ and $\beta_\Delta$.

By analysing the estimated skewness in Figure~\ref{fig_moments_n_b}, it can be seen that the absolute skewness of $\elpdHatC{\Md}{\y}$ and $\elpdHatErrC{\Md}{\y}$ is bigger when $\beta_\Delta$ is closer to zero. The models are more similar in predictive performance. While in the case of conditionalised design matrix $\+X$ and model variance $\tau$ in Section~\ref{sec_analytic_case}, the skewness of $\elpdC{\Md}{\y}$ is similar with all $\beta_\Delta$, in the simulated experiment this skewness decreases when $\beta_\Delta$ grows. When $|\beta_\Delta|>0$, the absolute skewness of $\elpdHatC{\Md}{\y}$ and $\elpdHatErrC{\Md}{\y}$ decreases towards zero when $n$ grows. Otherwise, when $\beta_\Delta=0$, similar to the problem setting in the analytic case study in Section~\ref{sec_analytic_case}, the skewness does not fade off when $n$ grows. These results show that problematic skewness can occur when the models are close in predictive performance and with smaller sample sizes. 
\begin{figure}[tb!]
  \centering
  \includegraphics[width=0.85\figurecontrolwidth]{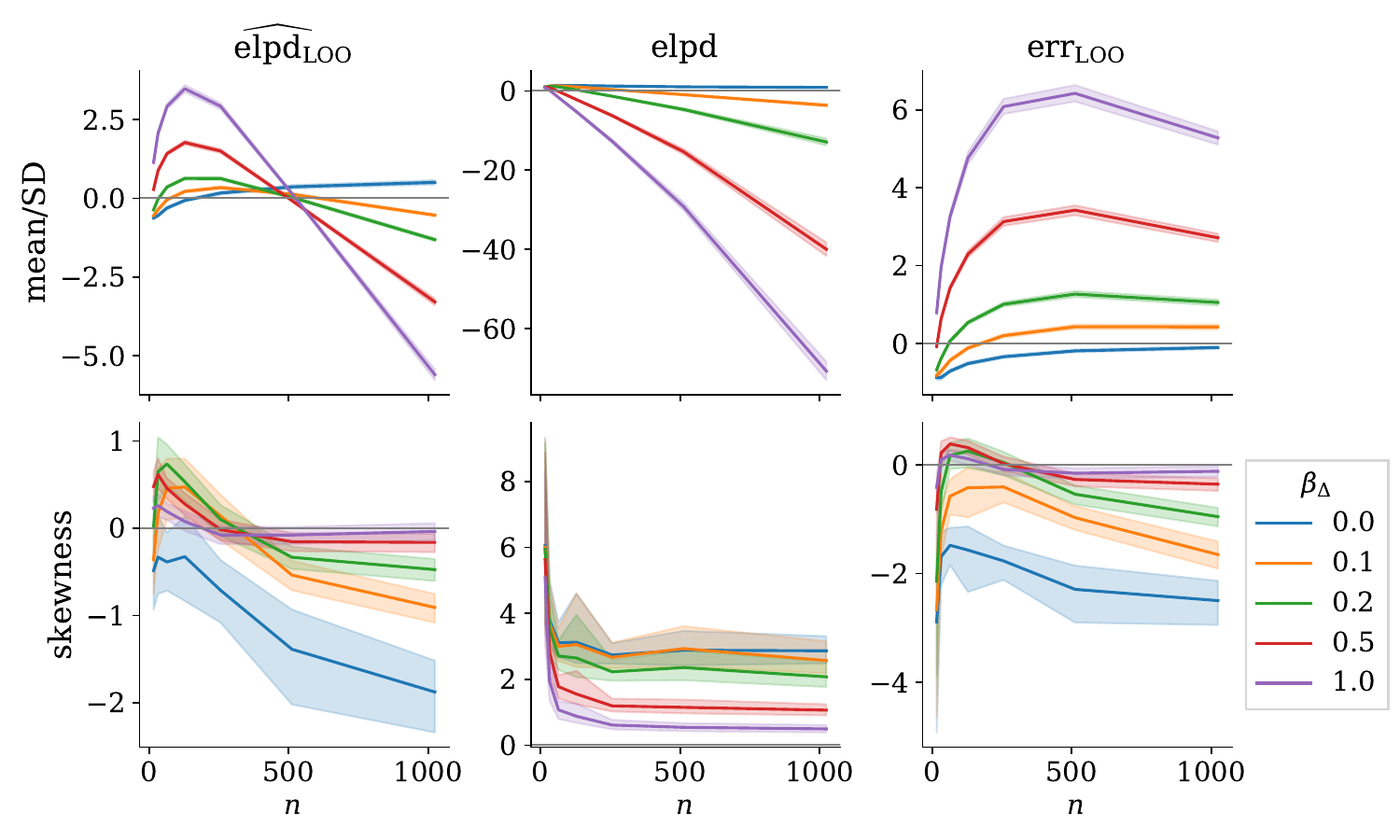}
  \caption{%
   Illustration of the estimated mean relative to the standard deviation and skewness for $\elpdHatC{\Md}{\y}$, $\elpdC{\Md}{\y}$, and for the error $\elpdHatErrC{\Md}{\y}$ as a function of the data size $n$ for various non-shared covariate effects $\beta_\Delta$, when there is an outlier observation in the data. The solid lines correspond to the median and the shaded area to the 95 \% confidence interval from the Bayesian bootstrap (BB) sample of size 2000 using the weighted moment estimators presented by \citet{Rimoldini_2013_weighted_moments}. As the effect of $\beta_\Delta$ is symmetric, the problem is simulated only with positive $\beta_\Delta$.
  }\label{fig_moments_n_b_out}
\end{figure}

Figure~\ref{fig_moments_n_b_out} illustrates the relative mean and skewness for the sampling distribution $\elpdHatC{\Md}{\y}$, for the distribution of the estimand $\elpdC{\Md}{\y}$, and for the error distribution $\elpdHatErrC{\Md}{\y}$ estimated from the simulated experiments as a function of the data size $n$ for different non-shared covariates' effects $\beta_\Delta$ when there is an outlier observation in the data. Compared to the analogous plot without the outlier in Figure~\ref{fig_moments_n_b}, introducing the outlier affects the distribution of $\elpdHatC{\Md}{\y}$ more than of the distribution of $\elpdC{\Md}{\y}$. This is plausible considering the leave-one-out technique used in the estimator. Due to the difference in the distribution of $\elpdHatC{\Md}{\y}$, the error $\elpdHatErrC{\Md}{\y}$ is also affected. The effect is greater when the non-shared covariates effect $\beta_\Delta$ is bigger.
\begin{figure}[tb!]
  \centering
  \includegraphics[width=0.75\figurecontrolwidth]{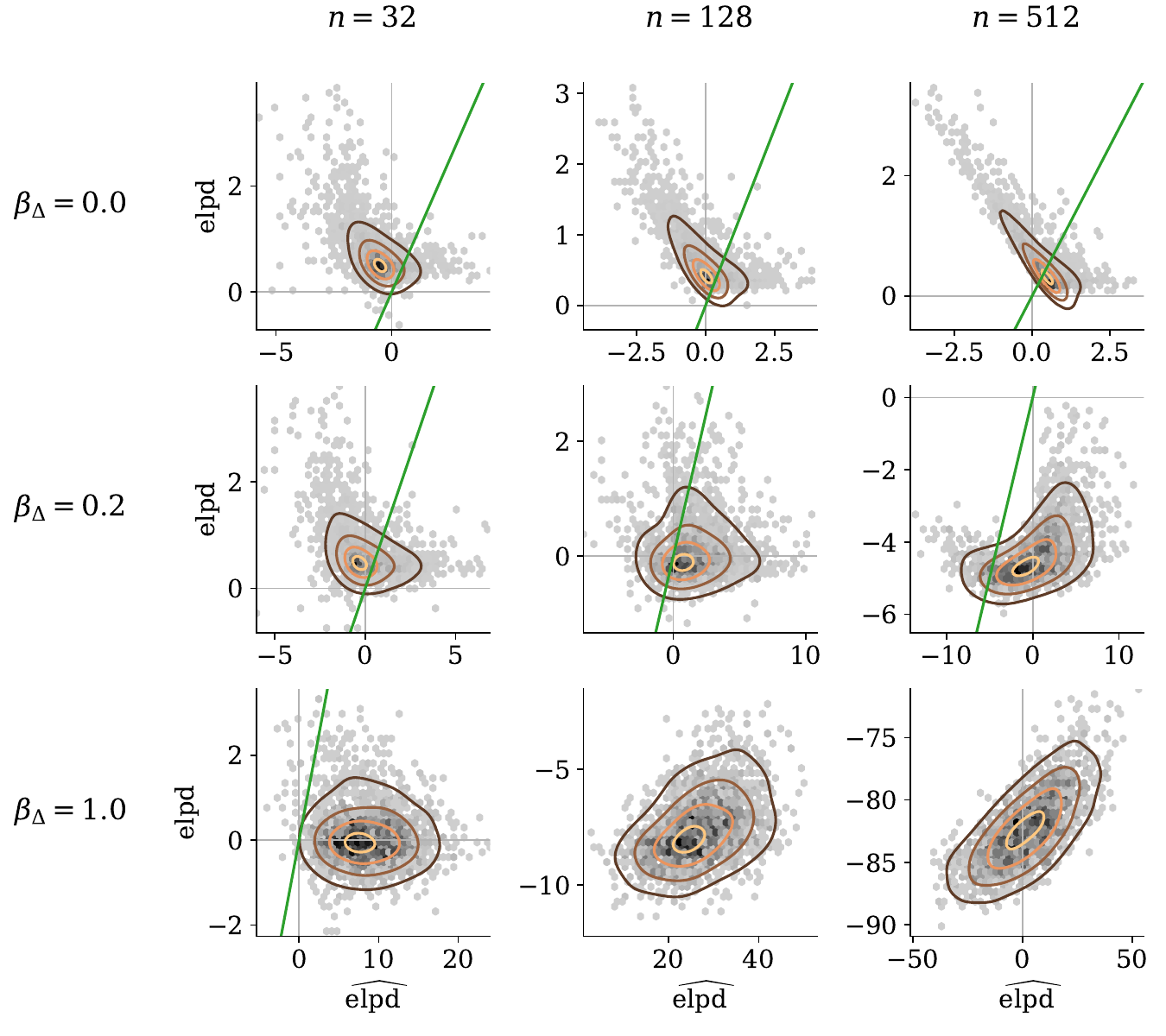}
  \caption{%
    Illustration of the joint distribution of $\elpdHatC{\Md}{\y}$ and $\elpdC{\Md}{\y}$ for various data sizes $n$, non-shared covariate effects $\beta_\Delta$, and an outlier in the data. The outlier scaling coefficient is set to $\mu_{\star \mathrm{r}} = 20$. Green diagonal line indicates where $\elpdHat{\Md}{\yobs} = \elpd{\Md}{\yobs}$.
    }\label{fig_joint_out}
\end{figure}

Figure~\ref{fig_joint_out} illustrates the joint distribution of the estimator $\elpdHatC{\Md}{\y}$ and the estimand $\elpdC{\Md}{\y}$ when there is an outlier observation present.
Similar to the case without an outlier illustrated in Figure~\ref{fig_joint}, although to a slightly lesser degree, the estimator and the estimand get negatively correlated when the models' predictive performances get more similar. In the outlier case, however, the estimator is biased and using the LOO-CV method is problematic. For example, in the case where $n=128$ and $\beta_\Delta=1.0$, the distributions of $\elpdHatC{\Md}{\y}$ and $\elpdC{\Md}{\y}$ lie in the opposite sides of sign and LOO-CV method will almost surely pick the wrong model.

Figure~\ref{fig_err_n_b_both} illustrates the behaviour of the error relative to the standard deviation $\elpdHatErr{\Md}{\y} \big/ \allowbreak \operatorname{SD}\bigl(\elpd{\Md}{\y}\bigr)$ for various non-shared covariate effects $\beta_\Delta$ and data sizes $n$ with and without an outlier observation. It can be seen from the figure that without outliers, the mean of the relative error is near zero in all settings, so the bias in the LOO-CV estimator is small. When an outlier is present in the data (Scenario~2), the relative error's mean usually deviates from zero, and the estimator is biased. Whether LOO-CV estimates the difference in the predictive performance to be further away or closer to zero or of a different sign depends on the situation.
Figure~\ref{fig_errdirection_n_b_both} illustrates the behaviour of
\begin{align*}
  \operatorname{sign}\Bigl(\elpdC{\Md}{\y}\Bigr) \frac{\elpdHatErrC{\Md}{\y}}{\operatorname{SD}\Bigl(\elpdC{\Md}{\y}\Bigr)} \,,
\end{align*}
the relative error directed towards $\elpdC{\Md}{\y} = 0$, for various non-shared covariate effects $\beta_\Delta$ and data sizes $n$ with and without an outlier observation. It can be seen from the figure that with an outlier observation,  LOO-CV often estimates the difference in the predictive performance to be smaller or of the opposite sign than the estimand $\elpdC{\Md}{\y}$.
\begin{figure}[tb!]
 \centering
 \includegraphics[width=0.80\figurecontrolwidth]{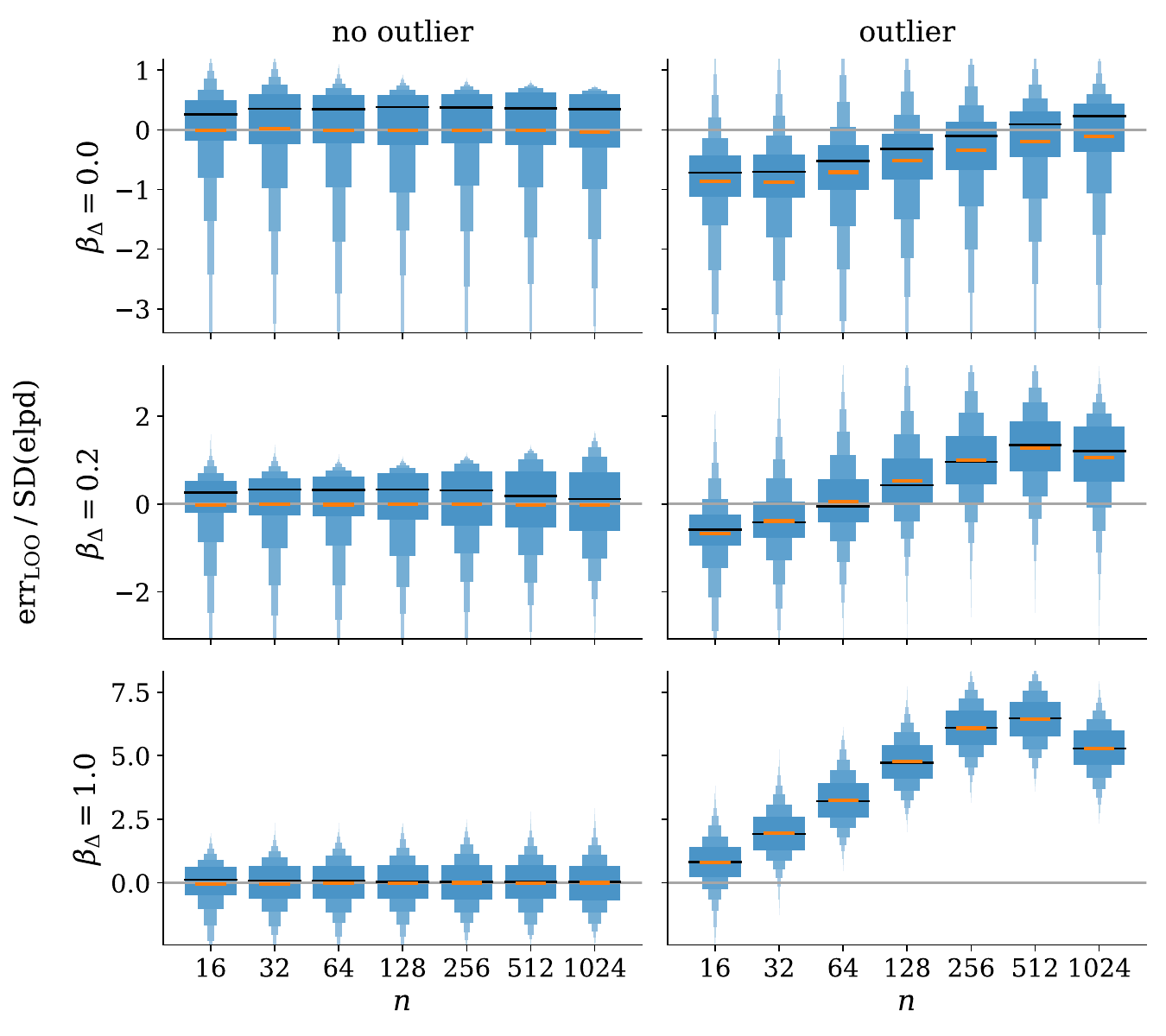}
 \caption{Distribution of the relative error $\elpdHatErr{\Md}{\y} \big/ \operatorname{SD}\left(\elpd{\Md}{\y}\right)$ for different data sizes $n$ and non-shared covariate effects $\beta_\Delta$. In the left column, there are no outliers in the data, and in the right column, there is one extreme outlier with a deviated mean of 20 times the standard deviation of $\y_i$. The distributions are visualised using letter-value plots or boxenplots~\citep{letter-value-plot}. The black lines correspond to the distribution's median, and the yellow lines indicate the mean. The bias can be considerable with an extreme outlier in the data (Scenario~2). Whether LOO-CV estimates the difference in the predictive performance to be further away or closer to zero or of different sign depends on the situation.} \label{fig_err_n_b_both}
\end{figure}
\begin{figure}[tb!]
  \centering
  \includegraphics[width=0.82\figurecontrolwidth]{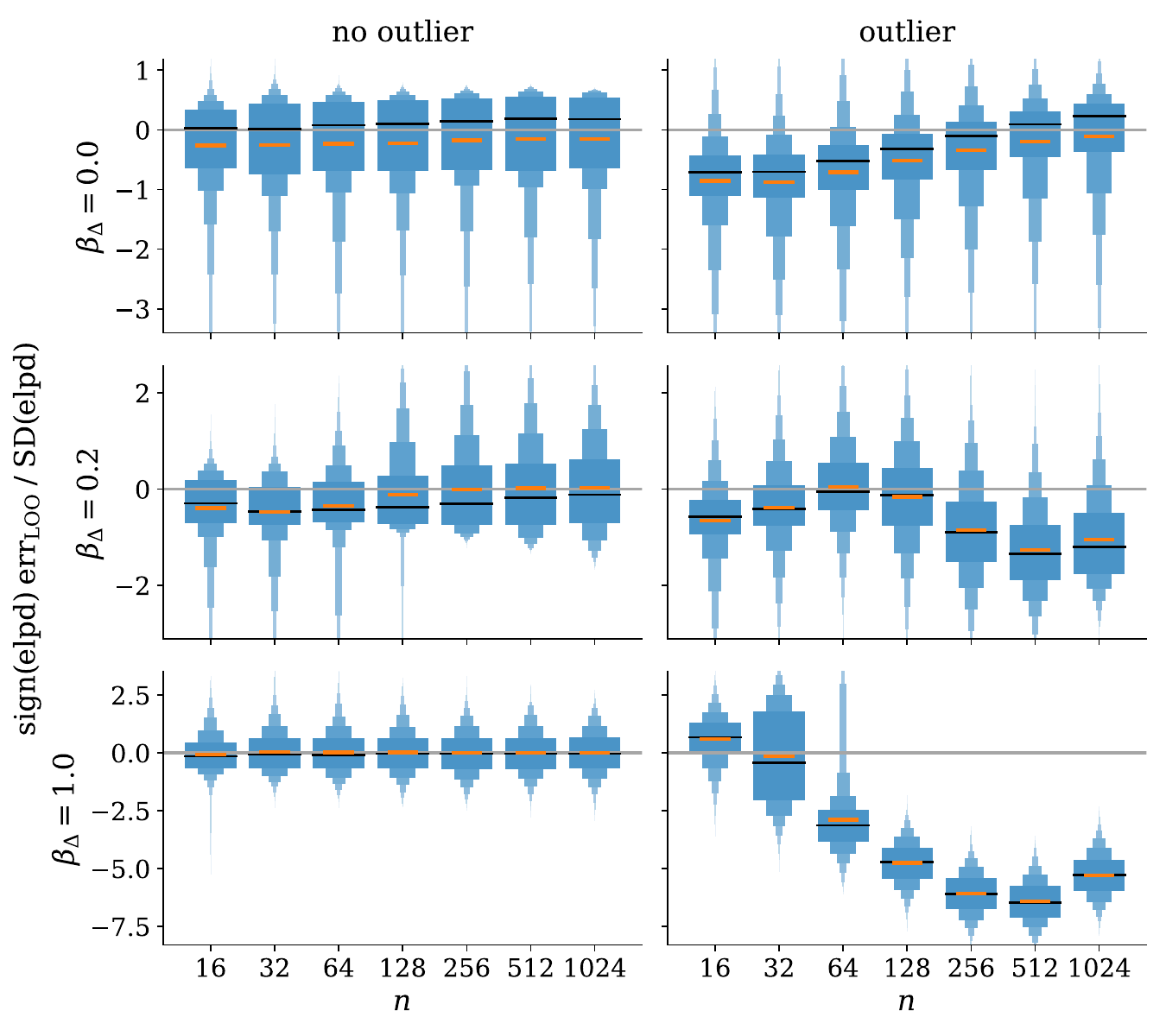}
  \caption{%
    Distribution of $\operatorname{sign}\left({}^\mathrm{sv}\elpdPlain\right) {}^\mathrm{sv}\elpdHatErrPlain \big/ \operatorname{SD}\left({}^\mathrm{sv}\elpdPlain\right)$, the relative error directed towards ${}^\mathrm{sv}\elpdPlain = 0$, in a model comparison setting (omitting arguments $(\Md \mid \y)$ for clarity) for different data sizes $n$ and non-shared covariate effects $\beta_\Delta$. Negative values indicate that LOO-CV estimates the difference in the predictive performance to be smaller or of the opposite sign and positive values indicate the difference is larger. In the left column, there are no outliers in the data, and in the right column, there is one outlier with deviated mean of 20 times the standard deviation of $\y_i$. The distributions are visualised using letter-value plots or boxenplots~\citep{letter-value-plot}. The black lines correspond to the median of the distribution, and the yellow lines indicate the mean. With an outlier observation, the directed relative error is typically negative.
  }\label{fig_errdirection_n_b_both}
\end{figure}

%
\begin{figure}[tb!]
  \centering
  \includegraphics[width=0.92\figurecontrolwidth]{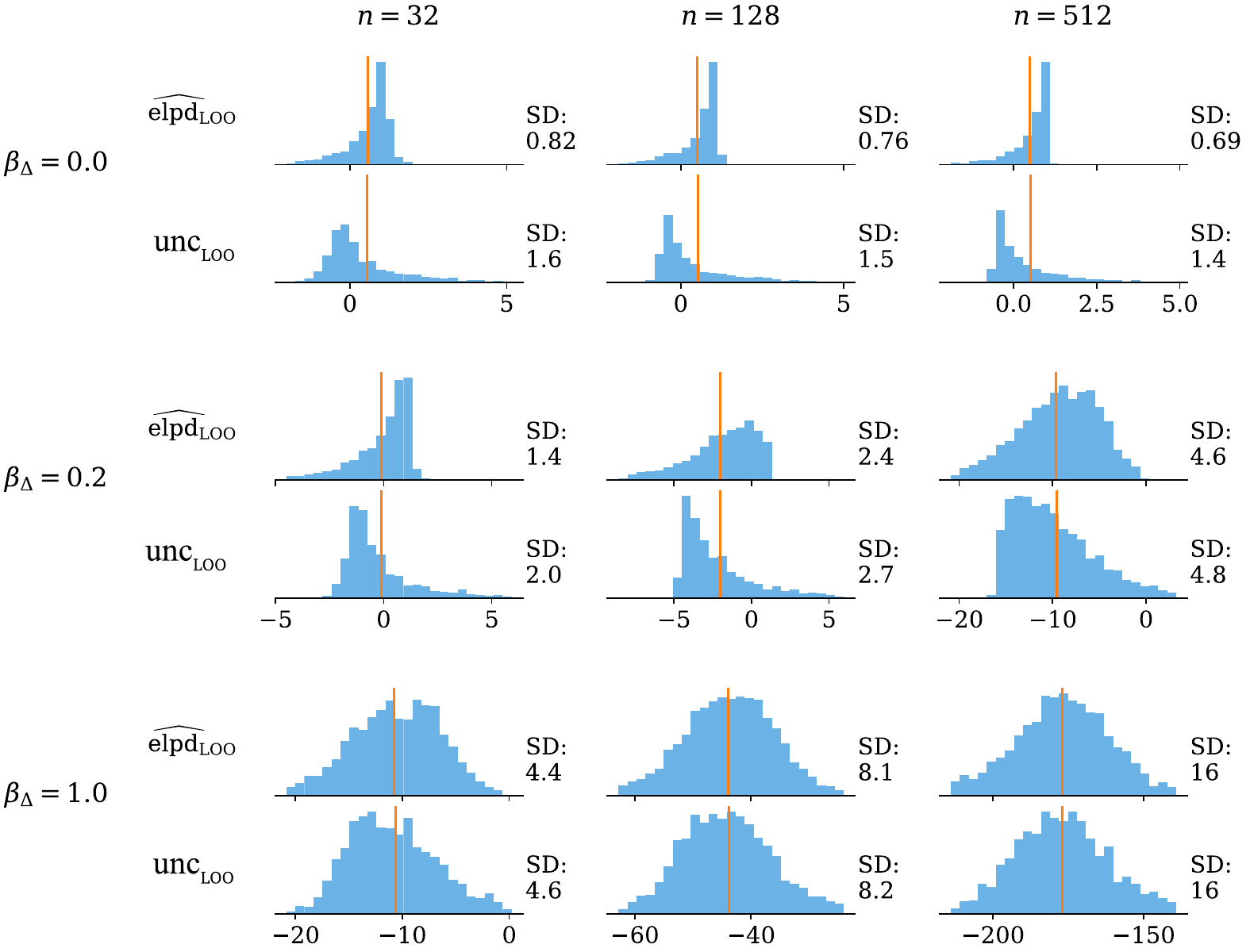}
  \caption{%
   Illustration of the distributions of $\elpdHatC{\Md}{\y}$ and $\elpdHatUnkC{\Md}{\yobs}$, where $\yobs$ is such that $\elpdHat{\Md}{\yobs} = \E\left[\elpdHatC{\Md}{\y}\right]$, for various data sizes $n$ and non-shared covariate effects $\beta_\Delta$. The yellow lines show the means of the distributions, and the corresponding sample standard deviation is displayed next to each histogram. In the problematic cases with small $n$ and $\beta_\Delta$, there is a weak connection in the skewness of the sampling and the error distributions. Thus, even with a better estimator for the sampling distribution, the uncertainty estimation is badly calibrated. For brevity, model labels are omitted in the notation in the figure.
  }\label{fig_err_vs_sampdist}
\end{figure}

Figure~\ref{fig_err_vs_sampdist} illustrates the difference between the sampling distribution $\elpdHatC{\Md}{\y}$ and the uncertainty distribution
\begin{equation}
 \elpdHatUnkC{\Md}{\yobs} = \elpdHat{\Md}{\yobs} - \elpdHatErrC{\Md}{\y} \,.
\end{equation}
Here $\yobs$ is selected such that $\elpdHat{\Md}{\yobs} = \E\left[\elpdHatC{\Md}{\y}\right]$ so that, in addition to the shape, the location of the former distribution can be directly compared to the location of the latter one. It can be seen from the figure that the distributions match when one model is clearly better than the other. When the models are more similar in predictive performance, however, the distribution of $\elpdHatC{\Md}{\y}$ has smaller variability than in the distribution of the uncertainty $\elpdHatUnkC{\Md}{\yobs}$ and the distribution is skewed to the wrong direction. Nevertheless, as the bias of the approximation is small, the means of the distributions are close in all problem settings.
%
\begin{figure}[tb!]
  \centering
  \includegraphics[width=0.92\figurecontrolwidth]{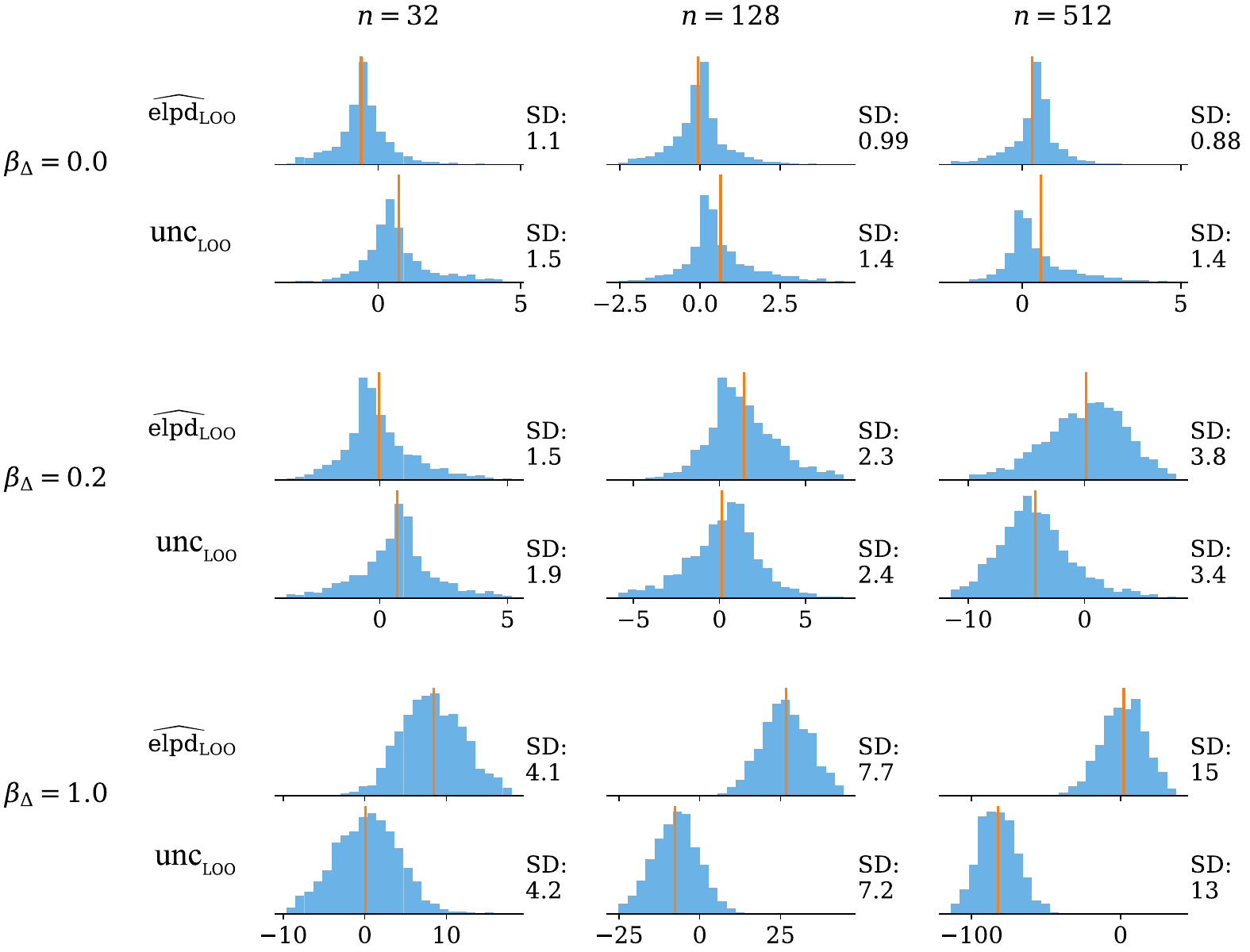}
  \caption{%
   Illustration of the distributions of $\elpdHatC{\Md}{\y}$ and $\elpdHatUnkC{\Md}{\yobs}$, where $\yobs$ is such that $\elpdHat{\Md}{\yobs} = \E\left[\elpdHatC{\Md}{\y}\right]$, for various data sizes $n$, non-shared covariate effects $\beta_\Delta$, and an outlier in the data. The yellow lines show the means of the distributions, and the corresponding sample standard deviation is displayed next to each histogram. For brevity, model labels are omitted in the notation in the figure.
  }\label{fig_err_vs_sampdist_out}
\end{figure}
Figure~\ref{fig_err_vs_sampdist_out} illustrates the difference between the sampling distribution $\elpdHatC{\Md}{\y}$ and the uncertainty distribution $\elpdHatUnkC{\Md}{\yobs}$ when there is an outlier observation present. Compared to the non-outlier case shown in Figure~\ref{fig_err_vs_sampdist}, in this model misspecification setting, the distributions are not notably skewed to the opposite directions anymore, but as the approximations are significantly biased, the means are clearly different.

\clearpage

Figure~\ref{fig_var_ratioi_n_b} illustrates the problem of underestimation of the variance with small data sizes $n$ (Scenario~3) and models with more similar predictive performances (Scenario~1). 
\begin{figure}[tb!]
 \centering
 \includegraphics[width=0.62\figurecontrolwidth]{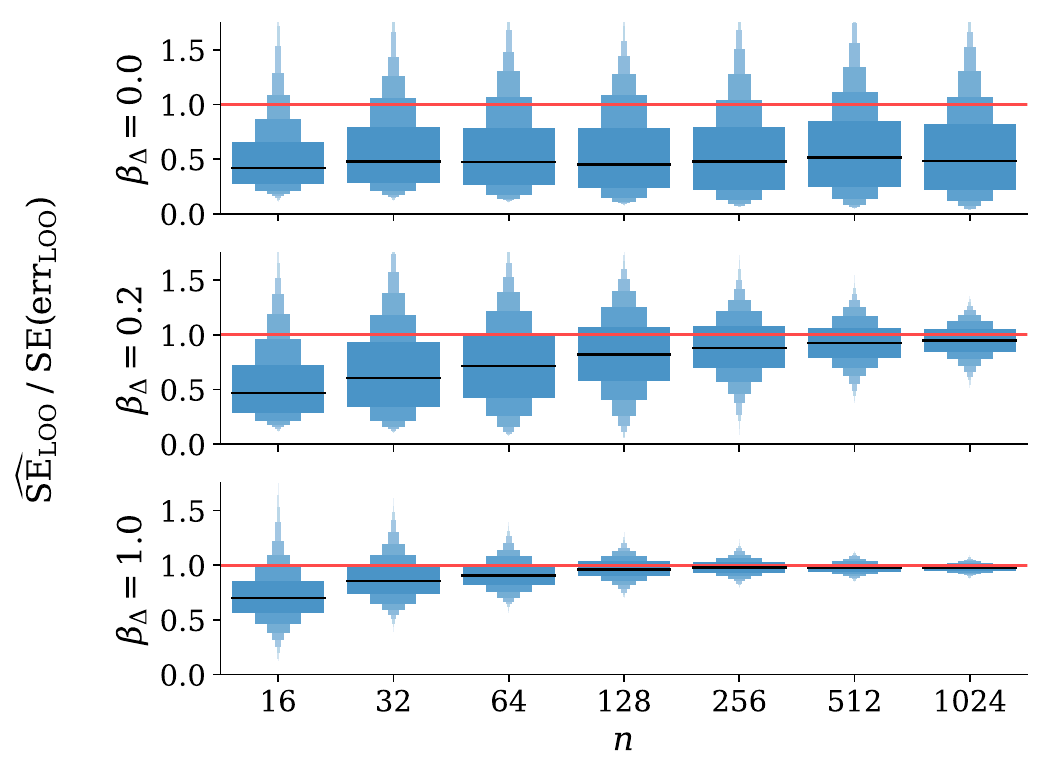}
 \caption{Distribution of the ratio $\seHat{\Md}{\y} \,\Big/\, \operatorname{SE}\Bigl(\elpdHatErr{\Md}{\y}\Bigr)$ for different data sizes $n$ and non-shared covariate effects $\beta_\Delta$. The red line highlights the target ratio of 1. The distributions are visualised using letter-value plots or boxenplots~\citep{letter-value-plot}. The black lines correspond to the median of the distribution. The variability is predominantly underestimated, with small $\beta_\Delta$ (Scenario~1) and small $n$ (Scenario~3).}
 \label{fig_var_ratioi_n_b}
\end{figure}
\begin{figure}[tb!]
  \centering
  \includegraphics[width=0.6\figurecontrolwidth]{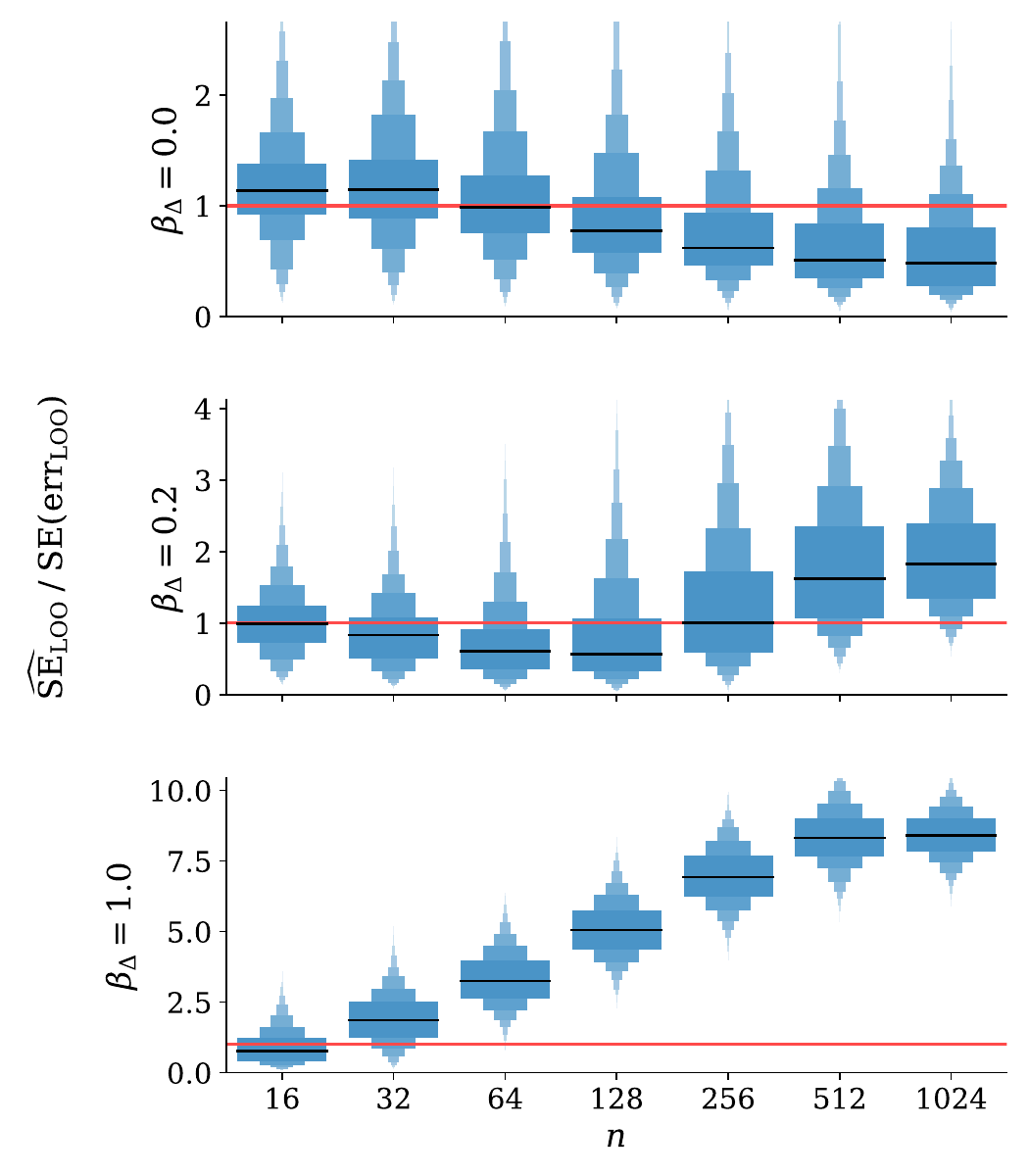}
  \caption{%
   Distribution of the ratio $\seHatC{\Md}{\y} \,\Big/\, \operatorname{SE}\Bigl(\elpdHatErrC{\Md}{\y}\Bigr)$ for different data sizes $n$ and non-shared covariate effects $\beta_\Delta$, when there is an outlier observation in the data. The red line highlights the target ratio of 1. The distributions are visualised using letter-value plots or boxenplots~\citep{letter-value-plot}. The black lines correspond to the median of the distribution.
  }\label{fig_var_ratioi_n_b_out}
\end{figure}

Figure~\ref{fig_var_ratioi_n_b_out} illustrates the problem of underestimation of the variance with small data sizes $n$ and models with more similar predictive performances when there is an outlier observation in the data. Compared to the non-outlier case shown in Figure~\ref{fig_var_ratioi_n_b}, in this model misspecification setting, the ratio is situationally also significantly larger than one so that the uncertainty is overestimated. In these situations, as demonstrated in Figure~\ref{fig_joint_out} the estimator is biased so that the overestimation is understandable and acceptable.

\clearpage
\begin{figure}[tbp!]
  \centering
  \includegraphics[width=0.75\figurecontrolwidth]{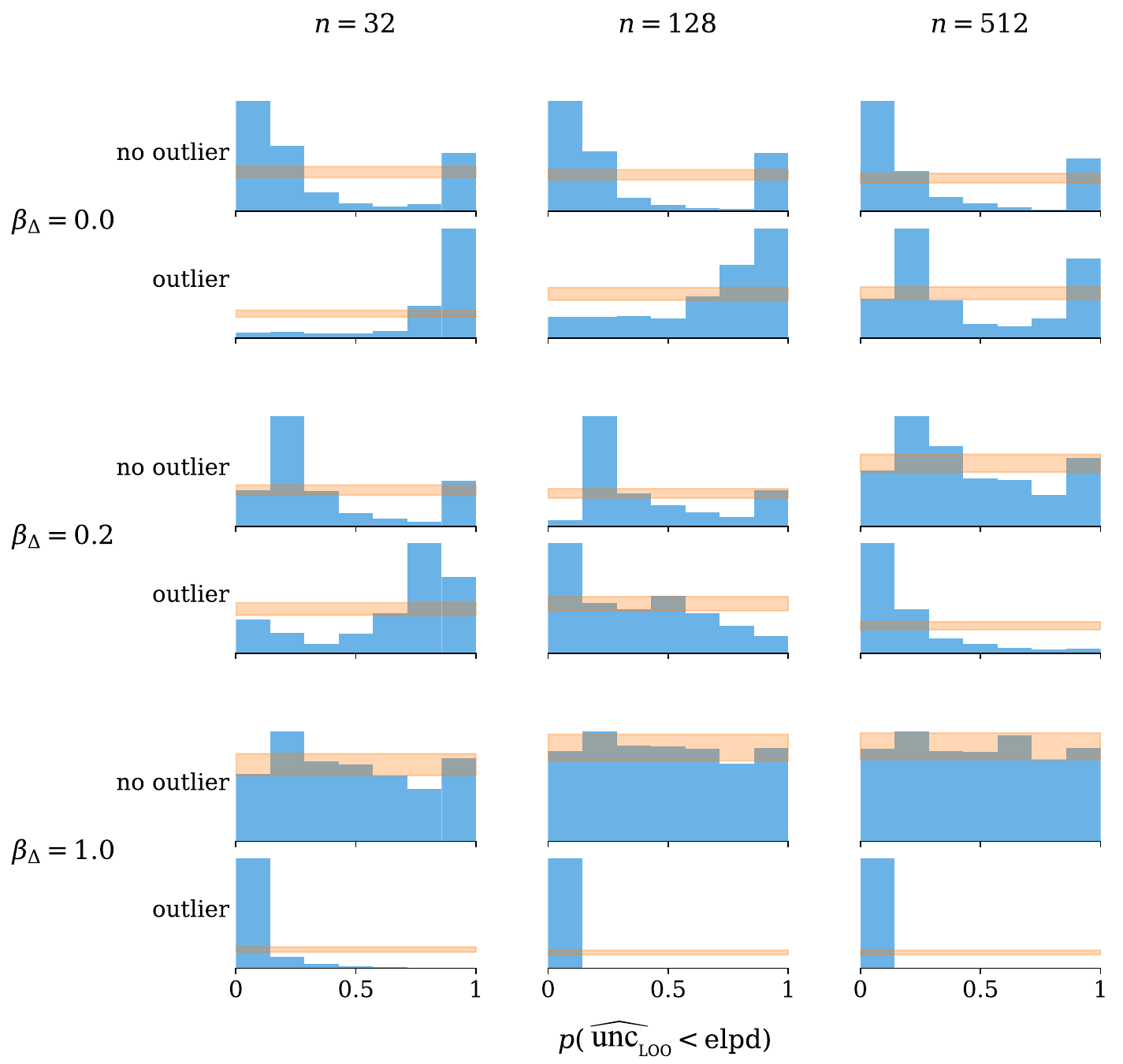}
  \caption{
   Calibration of the theoretical approximation based on $\elpdHatC{\Md}{\y}$ centred around $\elpdHat{\Md}{\yobs}$ for various data sizes $n$ and non-shared covariate effects $\beta_\Delta$, when there is an outlier observation in the data. The histograms show the distribution of $\p\Bigl(\elpdHatUnkHatC{\Md}{\y} < \elpdC{\Md}{\y}\Bigr)$, which would be uniform in a case of optimal calibration. 
  }\label{fig_calib_n_b_sampdist}
\end{figure}

Figure~\ref{fig_calib_n_b_sampdist} illustrates the calibration of the theoretical estimate based on the true distribution of $\elpdHatC{\Md}{\y}$ centred around $\elpdHatC{\Md}{\yobs}$ in various problem settings when there is an outlier observation in the data. It can be seen that the sampling distribution provides a good calibration only in the case of no outlier and large $\beta_\Delta$ or, to some degree, large $n$.

Figures~\ref{fig_boxen_adexp52_n512_b05} and~\ref{fig_unc_adexp53_n128_b05} provide additional information related to the experiments discussed in Section~\ref{subsec_additional_experiments}.
Figure~\ref{fig_boxen_adexp52_n512_b05} shows the relative error $\elpdHatErr{\Md}{\y} \big/  \operatorname{SD}\left(\elpd{\Md}{\y}\right)$ for different data sizes, the non-shared coefficient is equal to 0.5, and without outlier observations. The relative errors are symmetrical, and the mean and median are close to zero, confirming that also, in the extended examples, the bias goes asymptotically to zero (Section 3 of this paper;\citealp[][Section 5.1]{arlot_celisse_2010_cv_survey}; \citealp{Watanabe:2010d}). 
\begin{figure}[!ht]
 \centering
 \includegraphics[width=0.88\figurecontrolwidth]{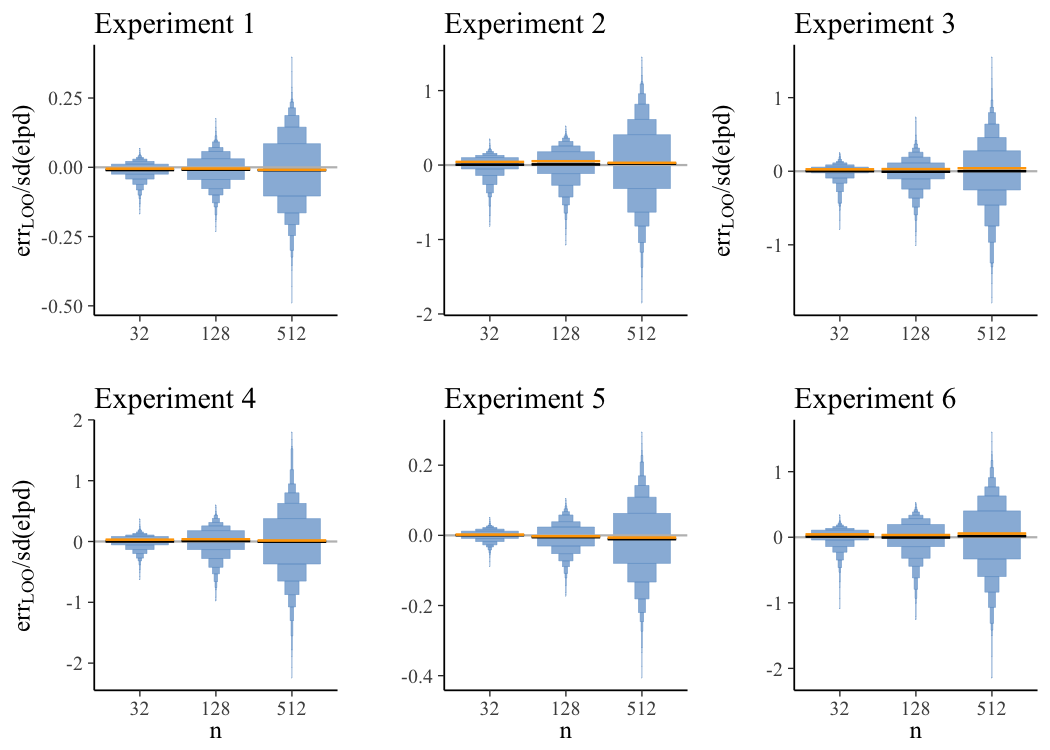}
 \caption{Distribution of the relative error for different data sizes $n = 32,128,512$ and for the non-shared covariate effect 0.5. The distributions are visualised using letter-value plots or boxplots. The black lines correspond to the median of the distribution, yellow lines indicate the mean, and the x-axis indicates the different data sizes n}
 \label{fig_boxen_adexp52_n512_b05}
\end{figure}
Figure~\ref{fig_unc_adexp53_n128_b05} compares the normal uncertainty approximation for data size $n = 128$, with a non-shared covariate effect  $\beta_\Delta = 0.5$. The results show that when models differ in their predictive performance slightly, the normal approximation provides a good fit for the LOO-CV uncertainty even in problematic scenarios where the number of observations is relatively small (Scenario 3). 
\begin{figure}[!ht]
 \centering
 \includegraphics[width=0.92\figurecontrolwidth]{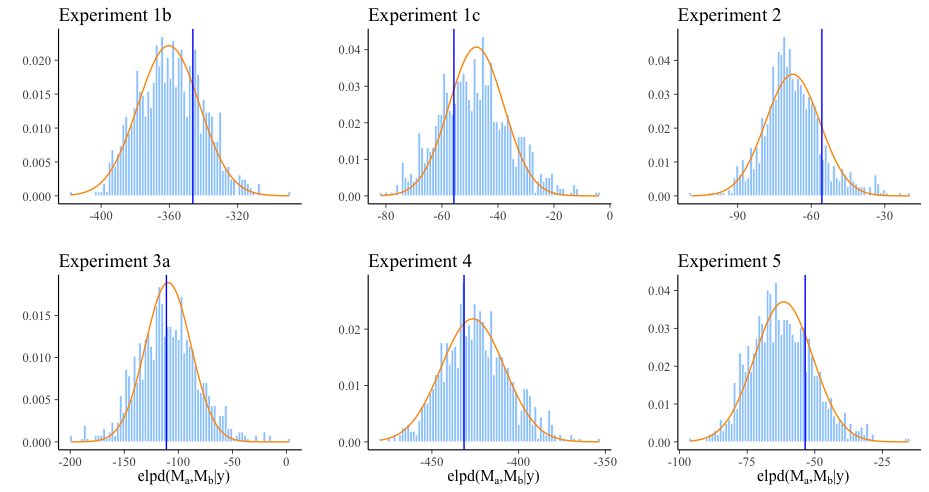}
 \caption{Approximated uncertainty using a normal distribution. The histogram represents the calculated uncertainty defined in equation~\eqref{eq_elpdhat_err} shifted by its mean, the orange line represents the normal approximation defined in Section~\ref{sec_intro_uncertainty}, and the vertical line corresponds to the $\elpd{\Md}{\yobs}$}
 \label{fig_unc_adexp53_n128_b05}
\end{figure}

\section{Other model variants}
\label{subsec_additional_experiments}

In this section, we present empirical results for six model variants, illustrating that the theoretical results generalise beyond the simplest case. We study models with 2) more covariates, 3) non-Gaussianity, 4) hierarchy, and 5) splines. We also demonstrate the behaviour with 1) fixed covariate values and 6) $K$-fold-CV. All the additional experiments have two nested regression models with data-generating mechanisms similar to~\eqref{eq_analytic_data_gen_process}, where $d=3$, $\+\beta = [0, 1, \beta_\Delta]$, and $\Sigma_\star = \eye$. The model \Ma{} is a (generalised) linear model with intercept and one covariate, following the structure as in~\eqref{eq_analytic_case_models}. The model \Mb{} follows the data-generating process by including the additional covariate. For simplicity, we only present the data-generating processes, as the model \Mb{} follows the same structure. 
\begin{enumerate}[topsep=0pt,partopsep=2pt,itemsep=2pt,parsep=2pt]
  \item \textbf{A linear model with fixed (non-random) covariate values.} The models are the same as in Equation~\eqref{eq_analytic_case_models}, but covariate $X_2$ is defined as a fixed uniform sequence $X_2 = -1 + 2k/n$, for $k =1,2,\ldots, n$. 
  \item \textbf{Linear model with more common covariates.}
  $$Y = 1 + \sum_{k=1}^5 Z_k + \beta_\Delta X_2 + \varepsilon,$$ 
  where $Z_k,X_2 \sim \N(0,1)$, $\varepsilon \sim \N(0,\tau^2)$, and $\tau$ unknown.
  
  \item \textbf{A linear hierarchical model with $k = 4$ groups.}
  $$Y = 1 + X_1 + \beta_\Delta \alpha_{j} +\varepsilon,$$
  $$\alpha_j \sim \N(\alpha_0,\sigma^2),$$
  where $\varepsilon \sim \N(0,\tau^2)$, $\tau$ unknown, $\alpha_0 \sim \N(0,1)$, and $j = 1,2,3,k$. 
  
  \item \textbf{A Poisson generalised linear model.} $Y \sim \Poisson(\mu)$, where $\mu = \exp(1 + X_1 + \beta_\Delta X_2)$, and $X_1,X_2 \sim \N(0,1)$. 
  
  \item \textbf{A spline model.} The data-generating process includes a non-linear dependency
  $$Y = X_1 + \beta_\Delta X_2 \cos(X_2) + \varepsilon,$$
  where, $X_1,X_2 \sim \N(0,1)$, $\varepsilon \sim \N(0,\tau^2)$, and $\tau$ unknown. The spline model
  is based on a linear combination of non-linear basis functions, which do not match the data-generating process, i.e. 
  $$M_b: Y = \beta_0 + \beta_1X_1 + \beta_2 s(X_2) + \varepsilon,$$ 
  where $s(X_2)$ represent the penalised B-spline matrix obtained for the covariate $X_2$. 

  \begin{figure}[!htp]
   \centering
   \includegraphics[width=0.84\figurecontrolwidth]{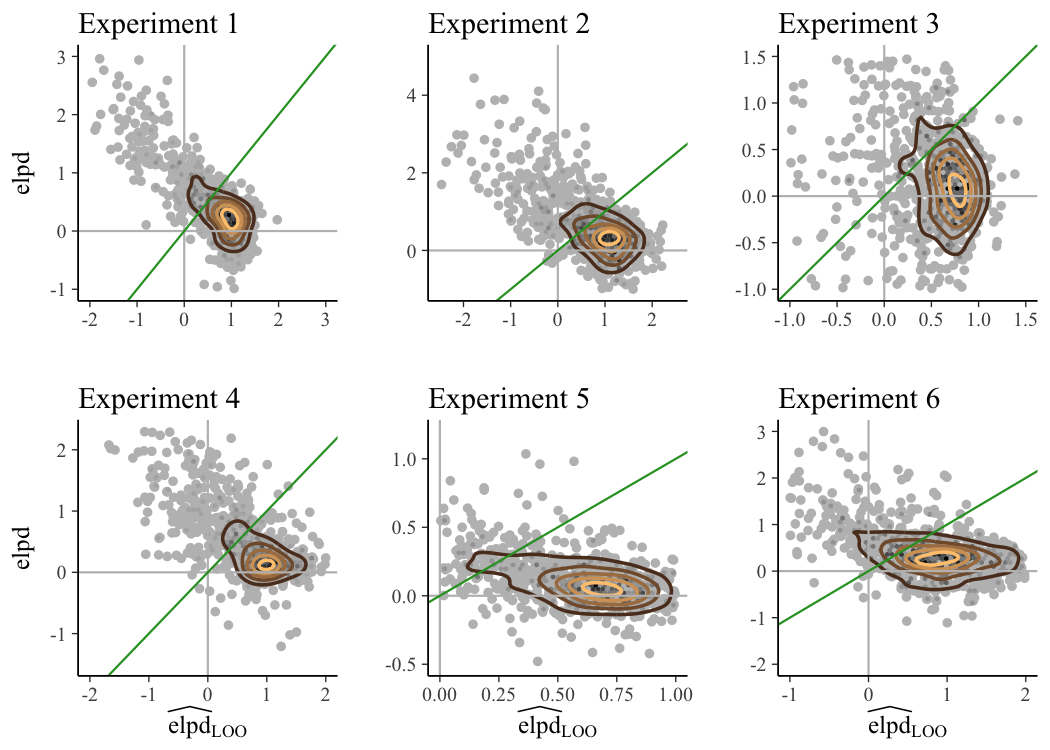}
   \vspace{-3pt}
   \caption{Illustration of the joint distribution for the LOO-CV estimator and $\elpd{\Mfd}{\yobs}$ for sample size of $n = 32$, and non-shared covariate effect $\beta_\Delta = 0.0$. The green diagonal line indicates where the variables match. }\label{fig_joint_adexp51_n32_b00}
   \end{figure}
\begin{figure}[!htp]
 \centering
 \includegraphics[width=0.84\figurecontrolwidth]{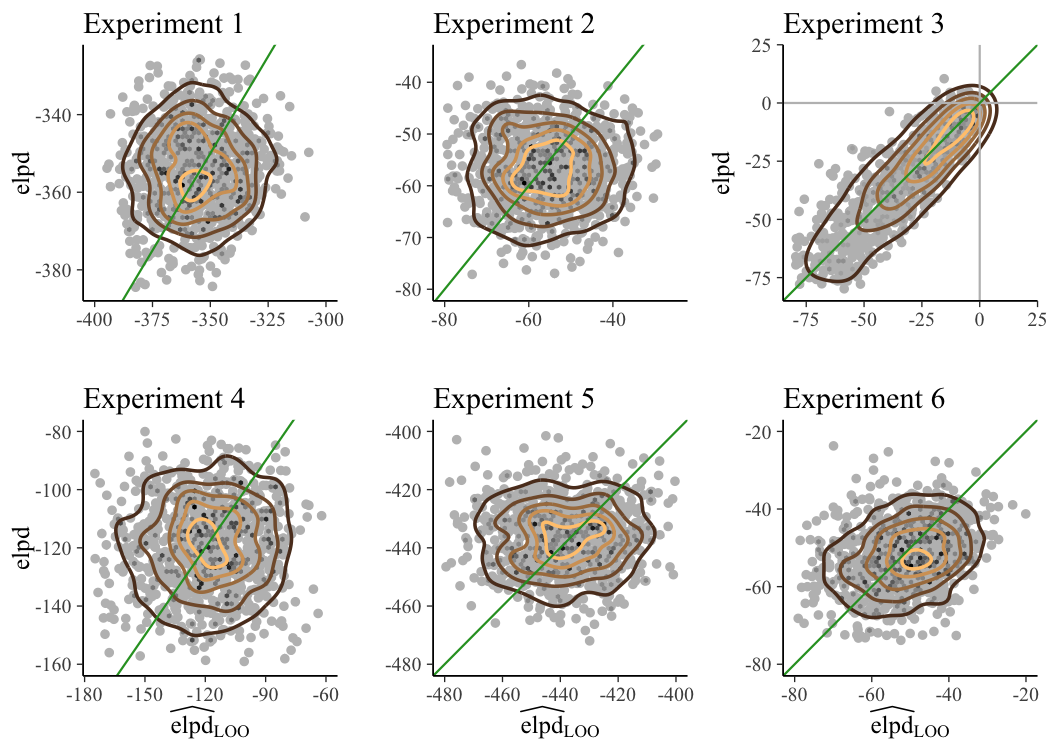}
   \vspace{-3pt}
 \caption{Illustration of the joint distribution for the LOO-CV estimator and $\elpd{\Mfd}{\yobs}$ for sample size of $n = 512$, and non-shared covariate effect $\beta_\Delta = 0.5$. The green diagonal line indicates where the variables match.}
\label{fig_joint_adexp51_n512_b05}
\end{figure}
  
  \item \textbf{10-fold-CV.} The model and data are the same as in the normal linear regression case, but 10-fold-CV with a random complete block design is used. 
  As the observations left out in each fold are likely not neighbours, we get a reasonable approximation of LOO-CV. Globally, as only $n-n/K$ (rounded to an integer) observations are used for the posterior, the predictive performance will likely be slightly worse than when using $n-1$ observations. We could correct this bias, but this is rarely done, as the bias is often small, and the bias correction increases the variance  \citep{Vehtari+Lampinen:2002}. We assume that a small bias doesn't change the general behaviour. If K-fold-CV is used to perform leave-one-group-out cross-validation, the behaviour is much different from LOO-CV, and we leave that for future research.
\end{enumerate}

In every experiment, we generate $1000$ data sets, and for each trial, we obtain pointwise LOO-CV (or 10-fold-CV) estimates $\elpdHat{\Mfd}{\yobs}$ and $\seHat{\Mfd}{\yobs}$. The respective target values $\elpd{\Mfd}{\yobs}$ are obtained using a separate test set of 4000 data sets of the same size simulated from the same data-generating process.

Figures~\ref{fig_joint_adexp51_n32_b00} and~\ref{fig_joint_adexp51_n512_b05} illustrate the joint distribution of the LOO-CV estimator and $\elpd{\Md}{\yobs}$ for different data sizes $n$ and non-shared covariate effects $\beta_\Delta$. Figure~\ref{fig_joint_adexp51_n32_b00} shows the results with small $n$ and models with similar predictions ($n = 32$ and  $\beta_\Delta = 0$). Figure~\ref{fig_joint_adexp51_n512_b05} shows the results with large $n$ and models with different predictions ($n = 512$ and  $\beta_\Delta = 0.5$). The results match the theoretical and previous experimental results. 
In the case of the hierarchical example (Experiment~3), there is a clear positive correlation, as the random realisations of data have variations in how strongly the groups differ, and thus, both the estimate and true value have more variation, but the error distribution doesn't get wider. 
Additional results are shown in Figures~\ref{fig_boxen_adexp52_n512_b05} and~\ref{fig_unc_adexp53_n128_b05} in Appendix~\ref{app_sec_additional_experiment_results}.

\end{document}